\documentclass{article}
\usepackage[utf8]{inputenc}
\usepackage[margin=1in]{geometry}

\usepackage{multirow}
\usepackage{booktabs}
\usepackage{amssymb,amsmath,amsthm,centernot,graphicx}
\usepackage{dsfont}
\usepackage[sort]{natbib}
\usepackage{enumitem}
\usepackage[size = small]{caption}
\usepackage[colorlinks = true, citecolor = blue]{hyperref}
\usepackage[capitalise]{cleveref}

\usepackage{refcount}
\usepackage{fmtcount}
\usepackage{mathtools}
\usepackage{algcompatible}
\usepackage{algorithm}

\newcommand{\E}{\mathbb{E}}

\newcommand{\R}{\mathbb{R}}

\newcommand{\F}{\mathcal{F}}

\renewcommand{\P}{\mathbb{P}}

\newcommand{\full}{\textup{full}}
\newcommand{\msm}{\mathsf{MSM}}
\newcommand{\Pobs}{P}
\newcommand{\Pfull}{P_{\full}}
\newcommand{\indep}{\rotatebox[origin=c]{90}{$\models$}}
\renewcommand{\d}{\textup{d}}
\newcommand{\T}{1}
\newcommand{\C}{0}
\newcommand{\Var}{\textup{Var}}
\newcommand{\sign}{\operatorname{sign}}
\newcommand{\ATE}{\textup{ATE}}
\newcommand{\ATT}{\textup{ATT}}

\DeclareMathOperator*{\argmin}{argmin}

\newcommand{\eg}{\textit{e.g.}}
\newcommand{\ie}{\textit{i.e.}}

\newtheorem{condition}{Condition}
\newtheorem{definition}{Definition} 
\newtheorem{theorem}{Theorem} 
\newtheorem{corollary}{Corollary} 
\newtheorem{lemma}{Lemma} 
\newtheorem{proposition}{Proposition} 
\allowdisplaybreaks

\title{Doubly-Valid/Doubly-Sharp Sensitivity Analysis for\\Causal Inference with Unmeasured Confounding}
\date{First version: December 21, 2021\qquad This version: July 22, 2022}
\newcommand*\samethanks[1][\value{footnote}]{\footnotemark[#1]}
\author{Jacob Dorn\thanks{Alphabetical Order} \\ Princeton University
\and Kevin Guo\samethanks \\ Stanford University
\and Nathan Kallus\samethanks \\ Cornell University}

\begin{document}
\maketitle

\begin{abstract}
We consider the problem of constructing bounds on the average treatment effect (ATE) when unmeasured confounders exist but have bounded influence.  Specifically, we assume that omitted confounders could not change the odds of treatment for any unit by more than a fixed factor.  We derive the sharp partial identification bounds implied by this assumption by leveraging distributionally robust optimization, and we propose estimators of these bounds with several novel robustness properties.  The first is \emph{double sharpness}:  our estimators consistently estimate the sharp ATE bounds when one of two nuisance parameters is misspecified and achieve semiparametric efficiency when all nuisance parameters are suitably consistent.  The second is \emph{double validity}: even when most nuisance parameters are misspecified, our estimators still provide valid but possibly conservative bounds for the ATE and our Wald confidence intervals remain valid even when our estimators are not asymptotically normal. As a result, our estimators provide a highly credible method for sensitivity analysis of causal inferences.
\end{abstract}

\section{Introduction}

Investigating causal relationships using only observational data is often necessary when experimentation is infeasible.  In the absence of a natural experiment, estimation often proceeds under the ``unconfoundedness" assumption that all relevant confounders have been measured.  Since this assumption is fundamentally untestable, it is imperative to conduct sensitivity analyses that explore how unobserved confounders might affect our inferences.

We consider sensitivity analysis under a nonparametric relaxation of unconfoundedness known as the \emph{marginal sensitivity model} (MSM).  This assumption allows for the existence of arbitrarily many unmeasured confounders, but posits that within each stratum of the observed covariates, measuring these confounders would not change the odds of receiving treatment by more than some factor. The posited factor can be varied by the data analyst to assess the robustness of conclusions to different magnitudes of unobserved confounding.  The MSM was introduced by \citet{tan2006} and has been applied in \citet{kallus2018interval, zsb2019, kallus2020confoundingrobust, kallus_zhou2020, lee2021causal, rosenman2020combining, dornGuo2021sharp, rosenman2021designing,soriano2021interpretable, jin2021sensitivity, yin2021conformal, nie2021covariate}, among others. 

This paper draws a connection between sensitivity analysis under the MSM and classical results on distributionally robust optimization.  In particular, we show that the problem of computing sharp upper and lower bounds on counterfactual means and treatment effects in the MSM reduces to the dual problem defining the \emph{conditional value at risk} of a probability distribution \citep{rockafellar2000}.  Using this observation, we derive new ``outcome regression" formulae for the optimal bounds on the average treatment effect (ATE) under the MSM.  These formulae complement the ``propensity weighting" formulae given in \cite{dornGuo2021sharp, jin2021sensitivity}. 

Based on this characterization, we propose new estimators for the sharp ATE bounds which we call the Doubly-Valid/Doubly-Sharp (DVDS) estimators.  These estimators combine three nuisances --- a quantile regression, a propensity model, and a transformed-outcome regression 
--- to yield a number of desirable robustness properties:

\begin{enumerate}
    \item \textbf{Double sharpness}:\quad
    The DVDS estimators have multiple opportunities to converge to the sharp bounds on the ATE.  In particular, if the quantile regression is well-specified, consistency holds as long as \emph{either} the propensity score \emph{or} the transformed-outcome regression is estimated consistently.  This partial double robustness also implies that the DVDS estimators are locally robust: even if all first-stage nuisance functions are consistently estimated at slow (nonparametric) rates of convergence, the DVDS estimators will still converge to the sharp ATE bounds at fast parametric rates.  Moreover, they will be asymptotically normal with semiparametrically efficient variance, so simple Wald confidence intervals are asymptotically calibrated and optimal.
    
    \item \textbf{Double validity}:\quad 
    Even when all opportunities for the DVDS estimators to consistently estimate the sharp ATE bounds fail, there are still multiple opportunities for these estimators to converge to \emph{valid} (but possibly conservative) bounds on the ATE.  This happens even when the quantile regression is misspecified, as long as one of the other two nuisances is well-specified.  This special structure also implies a novel inference result:  if both the propensity score and transformed-outcome regression are consistently estimated at nonparametric rates, then Wald confidence intervals remain asymptotically valid even though the DVDS point estimators may fail to be asymptotically normal or even consistent. Typically, Wald intervals based on such poorly-behaved estimators do not have any guarantees.
\end{enumerate}

No efficient or doubly-robust estimators were previously known under the MSM.  Moreover, all previously-proposed confidence intervals under the MSM \citep[e.g.][]{tan2006, zsb2019, dornGuo2021sharp, soriano2021interpretable} have coverage guarantees only when first-stage nuisance functions are estimated by parametric models.  In contrast, our method's estimation and inference can be valid even when nuisance functions are estimated by black-box machine learning techniques, possibly even inconsistently.  

Double validity is --- to our knowledge --- a never-before-seen robustness property that uniquely enables credible sensitivity analysis.  It is a refinement of the single validity guarantee in \citet{dornGuo2021sharp}.  Their sensitivity analysis estimates sharp bounds when both the propensity score and quantile regression are well-specified and estimates valid bounds when only the propensity score is well-specified.  Our estimator offers additional protection by having an honest shot at valid inference when the quantile regression estimator is inconsistent and at valid estimates when the propensity score is inconsistent as well.  These validity guarantees are particularly relevant as quantile regression may be more challenging than the usual propensity score and outcome regression problems.

Sensitivity analysis using DVDS estimators is especially convenient when the outcome variable is binary.  The binary-outcome DVDS estimators have closed-form expressions in terms of the outcome regression and the propensity score, meaning they can be computed ``for free" alongside the usual Augmented Inverse Propensity Weighting (AIPW) estimator of the ATE \citep{RRZDoubleRobust}.  Moreover, the validity of the DVDS confidence intervals holds under the same product-of-rates conditions typically used in the analysis of the AIPW estimator. 

The rest of this paper is organized as follows:  \cref{sec:problem_statement} introduces the sensitivity assumption; \cref{sec:partial_id} characterizes the sharp partial identification region; \cref{sec:estimator} defines the DVDS estimators; Sections \ref{sec:point_estimation_guarantees}, \ref{sec:InferenceOnBounds} present our statistical guarantees for point estimates and confidence intervals, respectively;  \cref{sec:experiments} illustrates the DVDS estimators empirically; and \cref{sec:discussion} concludes.  All proofs are deferred to the appendix.

\section{Sensitivity assumption and problem statement} \label{sec:problem_statement}

In this section, we describe the mathematical setting of this paper and introduce the MSM.

We consider the Neyman-Rubin potential outcomes model with a binary treatment \citep{neyman, rubin1974}.  We assume $n$ units $(X_i, Y_i(0), Y_i(1), Z_i, U_i)$ are sampled independently from a distribution $\Pfull$.  Here, $(X_i, U_i) \in \R^d \times \R^k$ is a vector of confounders, $( Y_i(0), Y_i(1) )$ are real-valued potential outcomes, and $Z_i \in \{ 0, 1 \}$ is a binary treatment.  However, the data analyst only observes $(X_i, Y_i, Z_i)$ where $Y_i = Y_i(Z_i)$.  The confounders $U_i$ are also never observed, and their dimension $k$ may be unknown.  We use $\Pobs$ to denote the distribution of the observables, so that $(X_i,Z_i,Y_i)\sim \Pobs$ are independent and identically distributed. We omit the subscript $i$ to denote a generic draw from $\Pfull$ or $\Pobs$.

Our goal is to use the observed data to draw inferences about the average treatment effect $\psi_{\ATE}(\Pfull) = \E_{\Pfull}[Y(1) - Y(0)]$ and the counterfactual means $\psi_z(\Pfull) = \E_{\Pfull}[Y(z)], z \in \{ 0, 1 \}$.  Extensions to the average treatment effect on the treated are presented in the Appendix.

Since we allow for unmeasured confounders, these quantities are not point-identified (meaning, not functions of $\Pobs$ alone) and cannot be consistently estimated.  Nevertheless, they may be bounded if we impose some assumptions on the unmeasured confounders $U_i$.  In this paper, we assume that measuring unobserved confounders could not change the odds of treatment by more than some factor $\Lambda \geq 1$.  Formally, this restricts the distribution $\Pfull$ to the \emph{marginal sensitivity model} \citep{tan2006}.

\begin{definition}[Marginal sensitivity model]
\label{definition:msm}
Given an observed-data distribution $\Pobs$ and an odds ratio bound $\Lambda \geq 1$, the \emph{marginal sensitivity model} $\msm(\Pobs, \Lambda)$ is the set of distributions $\Pfull$ on $(X,Y(0),Y(1),Z,U)$ with $U\in\R^k$ for some $k$ satisfying the following properties:
\begin{enumerate}[label=(\alph*), itemsep=-1ex, topsep=-1ex]
    \item \emph{Unconfoundedness}. $( Y(0), Y(1) ) \, \indep \, Z \mid (X, U)$ under $\Pfull$. \label{condition:latent_unconfoundedness}
    \item \emph{Bounded odds ratio}. Measuring $U$ does not change the odds of treatment by more than a factor of $\Lambda$:
    \begin{align}
        \Lambda^{-1} \leq \frac{\Pfull(Z = 1 \mid X, U)}{\Pfull(Z = 0 \mid X, U)} \bigg/ \frac{\Pfull(Z = 1 \mid X)}{\Pfull(Z = 0 \mid X)} \leq \Lambda,\quad\text{$\Pfull$-almost surely.} \label{eq:odds_ratio}
    \end{align} \label{condition:bounded_odds_ratio}
    \item \emph{Matches observable law}. The distribution of $(X, Y(Z), Z)$ induced under $\Pfull$ is equal to $\Pobs$. \label{condition:observable_implications}
\end{enumerate}
\end{definition}

The identifying power of this assumption come from the bounded odds ratio condition \ref{condition:bounded_odds_ratio}.  When $\Lambda = 1$, this condition implies that measuring $U_i$ would not change the odds of treatment at all, \ie, the observed data is already unconfounded. As $\Lambda \rightarrow \infty$, stronger and stronger confounding is allowed.  Condition \ref{condition:latent_unconfoundedness} is innocuous as it may always be satisfied by taking $U = ( Y(0), Y(1) )$.  We frame the assumption in terms of an unobserved variable $U$ since it may be helpful to have some specific confounders in mind when choosing $\Lambda$.  Meanwhile, any distribution not satisfying \ref{condition:observable_implications} evidently cannot be the true complete-data distribution.

For any finite $\Lambda$, the assumption $\Pfull \in \msm(\Pobs, \Lambda)$ restricts the estimands $\psi_0, \psi_1, \psi_{\textup{ATE}}$ to a closed interval known as the \emph{partially identified set} \citep[Theorems 1 and 2]{dornGuo2021sharp}:  
\begin{align}
[\psi_{\diamond}^-,~ \psi_{\diamond}^+] &= \left[ \inf_{\Pfull \in \msm(\Pobs, \Lambda)} \psi_{\diamond}(\Pfull),~ \sup_{\Pfull \in \msm(\Pobs, \Lambda)} \psi_{\diamond}(\Pfull) \right] \quad \text{for} \quad \diamond \in \{ 0, 1, \textup{ATE} \}. \label{eq:partially_identified_set}
\end{align}
The quantities $\psi_{\diamond}^-, \psi_{\diamond}^+$ are called the sharp upper and lower bounds for $\psi_{\diamond}$, as they are the smallest upper bound and largest lower bound on $\psi_{\diamond}$ that can be derived from observations from $\Pobs$, assuming the MSM holds.  The main goals of this paper are to develop robust estimators of these bounds and reliable confidence intervals for the partially identified set. For the rest of the paper, excluding \cref{sec:experiments}, we will fix some $\Lambda\geq1$ and study the corresponding sharp bounds and partially identified set.

\bigskip 
\noindent \textbf{Notations}.  
We set $\tau=\frac{\Lambda}{\Lambda + 1}\in[0.5,1)$ throughout the paper.
For a real number $t$, we set $\{ t \}_+ = \max \{ t, 0 \}$, $\{ t \}_- = \min \{ t, 0 \}$, and $\sign(t) = 1$ if $t \geq 0$ and $-1$ otherwise.  For a real-valued function $f(x, y, z)$, we set $|| f ||_q = (\int |f|^q \, \d P)^{1/q}$ if $q \in [1, \infty)$ and $|| f ||_{\infty} = \inf \{ t \in \R \, : \, P( | f(X,Y,Z)| \leq t ) = 1 \}$.  We also define the following functions:
\begin{align*}
e(x) &= \Pobs(Z = 1 \mid X = x) &&\text{Propensity score}\\
\mu(x, z) &= \E[Y \mid X = x, Z = z] &&\text{Outcome regression}\\
F(y \mid x, z) &= \Pobs(Y \leq y \mid X = x, Z = z) &&\text{Conditional outcome distribution}\\
Q_{\alpha}(x, z) &= \inf \{ q \, : \, F(q \mid x, z) \geq \alpha \} &&\text{Conditional outcome quantile}
\end{align*}
We assume throughout that $Y$ has a finite mean and that the propensity score is bounded away from zero and one.  These conditions ensure that the above functions are well-defined.  Finally, we adopt the following convention regarding the symbols $\pm$ and $\mp$:  any equation containing these symbols should be read twice, first with $\pm = +$ and then with $\pm = -$.  Meanwhile, $\mp$ should always read with the opposite sign as $\pm$.  This convention helps to unify formulae for upper and lower bounds.

\section{Partial identification and singly-valid estimation} \label{sec:partial_id}

In this section, we present two characterizations of the sharp partial identification bounds.  The first is based on inverse propensity weighting (IPW) and extends the results of  \citet{dornGuo2021sharp, jin2021sensitivity}.  The second is based on regression adjustment and is derived using classical results from distributionally robust optimization \citep{rockafellar2000}.  These characterizations motivate certain plug-in estimators of the identified set, which we combine in \Cref{sec:estimator} to form the DVDS estimators.

For simplicity, this section focuses on $\psi_1^+$, the sharp upper bound for $\E_{\Pfull}[Y(1)]$.  The same arguments apply to characterize $\psi_1^-, \psi_0^+$ or $\psi_0^-$, after exchanging $Y$ with $-Y$ and/or $Z$ with $1 - Z$.  These bounds can then be subtracted to obtain sharp bounds on the ATE \citep[Theorem 2]{dornGuo2021sharp}.  Note that, in contrast, in other sensitivity models, subtracting sharp bounds on counterfactual means does not always give sharp bounds for the ATE \citep[Remark 3]{yadlowsky2018bounds}.  This makes the MSM especially convenient for sensitivity analysis.

\subsection{Formulation using adversarial propensities} \label{sec:eform}

In the absence of unobserved confounding, counterfactual averages can be identified using the IPW formula $\E_{\Pfull}[Y(1)] = \E[ YZ / e(X)]$.  \citet[Proposition 2]{dornGuo2021sharp} shows that a similar formula identifies $\psi_1^+$ once the propensity score $e(X)$ is replaced by an ``adversarial" propensity score.%

\begin{proposition}
\label{prop:ipw}
There exists a function $e_+(x, y)$ satisfying $\E[ Z / e_+(X, Y) \mid X ] = 1$ and
\begin{align}
\frac{e_+(x, y)}{1 - e_+(x, y)} \bigg/ \frac{e(x)}{1 - e(x)} &= \left\{
\begin{array}{ll}
\Lambda^{-1} &\text{when } y > Q_{\tau}(x, 1)\\
\Lambda &\text{when } y < Q_{\tau}(x, 1)
\end{array}
\right. .
\label{eq:saturating_or_bounds}
\end{align}
This function is the adversarial propensity score in the sense that $\psi_1^+ = \E[YZ/e_+(X, Y)]$.
\end{proposition}

The interpretation of Proposition \ref{prop:ipw} is the following.  To maximize $\E_{\Pfull}[Y(1)]$, the worst-case confounding structure in $\msm(\Pobs, \Lambda)$ treats all ``small" values of $Y(1)$ with the largest allowed probability and treats all ``large" values of $Y(1)$ with the smallest allowed probability.  The cutoff between large and small is determined separately for each covariate level $x$.  Similar cutoff phenomena are observed in many other partial identification problems \citep{AronowLeeInterpretable, manski2003partial, lee2009training, jin2021sensitivity}.

Despite its interpretability, it is not obvious how the identification formula from \Cref{prop:ipw} might be used for estimation.  One issue is that \Cref{eq:saturating_or_bounds} does not specify the value of the worst-case propensity score $e_+(x, y)$ on the boundary event $y = Q_{\tau}(x, 1)$.  

We now derive a new ``propensity weighting" identification formula from \Cref{prop:ipw} that is more useful for estimation purposes.  Since the worst-case propensity score $e_+$ satisfies $\E[ Z / e_+(X, Y) \mid X] = 1$, we may add $0 = \E[ Q_{\tau}(X, 1) Z / e(X) - Q_{\tau}(X, 1) Z / e_+(X, Y)]$ in \cref{prop:ipw} to obtain:
\begin{align}
\psi_1^+ &= \E \left[ \frac{Q_{\tau}(X, 1) Z}{e(X)} + \frac{( Y - Q_{\tau}(X, 1) ) Z}{e_+(X, Y)} \right] \nonumber \\
&= \E \left[ \frac{Q_{\tau}(X, 1) Z}{e(X)} + ( Y - Q_{\tau}(X, 1) ) Z \left( 1 + \frac{1 - e(X)}{e(X)} \Lambda^{\sign( Y - Q_{\tau}(X, 1) )} \right) \right]. \label{eq:robust_ipw_formula}
\end{align}
In the final step, we plugged in the value of $e_+(X,Y)$ specified by \Cref{eq:saturating_or_bounds}.  This was possible even when $Y = Q_{\tau}(X, 1)$ because, on that event, the second term in \cref{eq:robust_ipw_formula} is zero anyway.

The identification formula in \Cref{eq:robust_ipw_formula} is fully explicit and suggests a natural strategy for estimating $\psi_1^+$:  first, estimate the propensity score $e$ and the conditional quantile $Q_{\tau}$ from data, then plug these estimated quantities into \Cref{eq:robust_ipw_formula}.

The following Proposition shows that the natural strategy
is \emph{singly valid}:  if the analyst's estimate of the propensity score is correct, then the plug-in estimator based on \Cref{eq:robust_ipw_formula} yields a valid (but possibly conservative) upper bound for $\psi_1^+$ even if the analyst estimates the quantile function $Q_{\tau}$ incorrectly.  This property will play a crucial role in our estimation results to come.

\begin{proposition} \label{prop:ipw_valid}
Plugging any fixed integrable function $\hat{Q}_{\tau}$ into \Cref{eq:robust_ipw_formula} gives an upper bound for $\psi_1^+$:
\begin{align*}
\psi_1^+ \leq \E \left[ \frac{\hat{Q}_{\tau}(X, 1) Z}{e(X)} + ( Y - \hat{Q}_{\tau}(X, 1) ) Z \left(1 + \frac{1 - e(X)}{e(X)} \Lambda^{\sign ( Y - \hat{Q}_{\tau}(X, 1) )} \right) \right].
\end{align*}
Equality is attained if $\hat{Q}_{\tau}$ is the true conditional quantile $Q_{\tau}$.
\end{proposition}

The source of this robustness is a simple sign-matching argument:  $( Y - \hat{Q}_{\tau}(X, 1) ) \Lambda^{\sign ( Y - \hat{Q}_{\tau}(X, 1) )}$ is always at least as large as $( Y - \hat{Q}_{\tau}(X, 1) ) \Lambda^{\sign ( Y - Q_{\tau}(X, 1) )}$, since $\Lambda^{\sign ( Y - \hat{Q}_{\tau}(X, 1) )}$ takes on its largest (smallest) value whenever $Y - \hat{Q}_{\tau}(X, 1)$ is positive (negative).  Thus, using a misspecified cutoff is at least as conservative as using the true cutoff.

\subsection{Formulation using adversarial regressions} \label{sec:kform}

A second way of identifying counterfactual averages under unconfoundedness is via the regression adjustment formula $\E_{\Pfull}[Y(1)] = \E[ZY + (1 - Z) \mu(X, 1)]$.  In this section, we show that a similar formula will identify $\psi_1^+$ in the presence of unobserved confounding, after replacing the outcome regression $\mu(X, 1)$ by an ``adversarial" outcome regression.  The formula motivates another singly-valid plug-in estimator.

We start with the identity $\E_{\Pfull}[Y(1)] = \E[ZY + (1 - Z) \E_{\Pfull}[Y(1) \mid X, Z = 0]]$, which holds even in the presence of unmeasured confounding.  The only unknown quantity here is the counterfactual regression $\E_{\Pfull}[ Y(1) \mid X = x, Z = 0]$, which can be written as an integral relative to the distribution $F_{\full}(y \mid x, 0) = \Pfull(Y(1) \leq y \mid X = x, Z = 0)$.
\begin{align}
\E_{\Pfull}[Y(1) \mid X = x, Z = 0] = \int y \, \d F_{\full}(y \mid x, 0) \label{eq:integral_form}
\end{align}

While the counterfactual distribution $F_{\full}(y \mid x, 0)$ is unobservable, the marginal sensitivity model allows us to bound it in terms of the distribution of the observed outcome, $F(y \mid x, 1)$.  In particular, the bounded odds ratio condition \ref{condition:bounded_odds_ratio} and Bayes' theorem imply that their likelihood ratio must satisfy:
\begin{align}
    \Lambda^{-1} \leq \frac{\d F_{\full}(y \mid x, 0)}{\d F(y \mid x, 1)} \leq \Lambda \label{eq:lr_constraint}
\end{align}
for almost every $x$.  Moreover, these bounds are achievable when $\Pfull$ makes the potential outcomes conditionally independent of one another.  Therefore, the largest allowed counterfactual regression can be computed by maximizing the integral (\ref{eq:integral_form}) over measures $F_{\full}$ satisfying the likelihood ratio constraint (\ref{eq:lr_constraint}).  This optimization problem was considered but not solved in \citet{tan2006}.

After a simple reparameterization, the counterfactual regression maximization problem can be solved using results from distributionally robust optimization.  The likelihood ratio constraint on $F_{\full}$ is equivalent to the condition that $F_{\full}(y \mid x, 0) = \Lambda^{-1} F(y \mid x, 1) + (1 - \Lambda^{-1}) G(y)$ for some distribution $G$ with $\d G(y)/\d F(y \mid x, 1) \leq 1/(1- \tau)$. Thus, the largest value of the counterfactual regression allowed under the MSM is:
\begin{align}
\Lambda^{-1} \int y \, \d F(y \mid x, 1) + (1 - \Lambda^{-1}) \sup_{\frac{\d G(y)}{\d F(y \mid x, 1)} \leq \frac{1}{1 - \tau}} \int y \, \d G(y). \label{eq:regression_cvar_mixture}
\end{align}
In robust optimization, the problem of maximizing $\int y \, \d G(y)$ subject to the likelihood ratio constraint $\d G / \d F \leq 1/(1 - \tau)$ is known as the dual problem defining the level-$\tau$ \emph{conditional value at risk (CVaR)} of the distribution $F$ \citep{ang_etal2017, rockafellar2000, Rockafellar2006Generalized, ruszczynski2006optimization}.  It can be solved explicitly in terms of the conditional quantiles of $F$:
\begin{align}
\begin{split}
\textup{CVaR}_{\tau}(x, z) &= \sup_{\frac{\d G(y)}{\d F(y \mid x, z)} \leq \frac{1}{1 - \tau}} \int y \, \d G(y) = \int [ Q_{\tau}(x, z) + \tfrac{1}{1 - \tau} \{ y - Q_{\tau}(x, z) \}_+] \, \d F(y \mid x, z). \label{eq:cvar_definition}
\end{split}
\end{align}
Plugging the CVaR into \Cref{eq:regression_cvar_mixture} yields our adversarial ``regression adjustment" formula for $\psi_1^+$.

\begin{proposition}
\label{prop:dm}
Define $\rho_+(x, z)$ to be the following outcome regression/CVaR mixture:
\begin{align*}
\rho_+(x, z) &= \Lambda^{-1} \mu(x, z) + (1 - \Lambda^{-1}) \textup{CVaR}_{\tau}(x, z)
\end{align*}
Then $\rho_+(\cdot, 1)$ is the adversarial counterfactual regression in the sense that $\psi_1^+ = \E[ ZY + (1 - Z) \rho_+(X, 1)]$.
\end{proposition}

The interpretation of \cref{prop:dm} is the following. The adversarial counterfactual regression $\rho_+$ is a mixture between the average of the distribution $F(y\mid x,z)$ and the average of only the largest $100(1-\tau)\%$ of values under that distribution.
As $\Lambda$ grows, the CVaR term both gets more weight and uses a more extreme tail average.  In the limit, the CVaR converges to the upper bound of the support of $F(y \mid x, z)$.  Because of its interpretation as an extremal average, the CVaR appears naturally in partial identification problems \citep{lee2009training, semenova2021generalized}.

We now discuss how this identification formula might be used to form a plug-in estimator for $\psi_1^+$. Observe from \cref{eq:regression_cvar_mixture,eq:cvar_definition} that $\rho_+$ is the conditional mean of a certain transformed outcome:
\begin{align*}
\rho_+(x, z) = \int ( \Lambda^{-1} y + (1 - \Lambda^{-1})[ Q_{\tau}(x, z) + \tfrac{1}{1 - \tau} \{ y - Q_{\tau}(x, z) \}_+] \, \d F(y \mid x, z).
\end{align*}
Therefore, to estimate $\rho_+$, one can first learn $Q_{\tau}$ by quantile regression and then learn the conditional mean of the \emph{estimated} transformed outcome $\Lambda^{-1} Y_i + (1 - \Lambda^{-1}) [ \hat{Q}_{\tau}(X_i, Z_i) + \tfrac{1}{1 - \tau} \{ Y_i - \hat{Q}_{\tau}(X_i, Z_i)]$.  
The resulting estimate of $\rho_+$ may then be plugged into the identification formula in \cref{prop:dm}.

The following Proposition shows that this plug-in estimator for $\psi_1^+$ also has a single validity property:  even if the quantile regression is misspecified, we still obtain a valid (but possibly conservative) upper bound for $\psi_1^+$ as long as the transformed outcome regression correctly learns the conditional mean of the \emph{estimated} transformed outcome.  %

\begin{proposition} \label{prop:dm_valid}
For any integrable function $\hat{Q}_{\tau}(x, z)$, let $\varrho_+(x, z; \hat{Q}_{\tau})$ denote the conditional mean of the estimated transformed outcome based on $\hat{Q}_{\tau}$:
\begin{align*}
\varrho_+(x, z;\hat Q_\tau) &= \int (\Lambda^{-1} y + (1 - \Lambda^{-1}) [ \hat Q_{\tau}(x, z) + \tfrac{1}{1 - \tau} \{ y - \hat Q_{\tau}(x, z) \}_+]) \, \d F(y \mid x, z).
\end{align*}
For any fixed $\hat{Q}_{\tau}$, we have $\psi_1^+ \leq \E[ ZY + (1 - Z) \varrho_+(X, 1; \hat Q_\tau)]$, with equality when $\hat{Q}_{\tau} = Q_{\tau}$.
\end{proposition}

The explanation for this robustness is that the CVaR is the optimal value of the following minimization problem solved by the true quantile $Q_{\tau}$:
\begin{align*}
\textup{CVaR}_{\tau}(x, z) = \inf_{q \in \R} \int q + \tfrac{1}{1 - \tau} \{ y - q \}_+ \, \d F(y \mid x, z) = \int Q_{\tau}(x, 1) + \tfrac{1}{1 - \tau} \{ y - Q_{\tau}(x, 1) \}_+ \, \d F(y \mid x, z) 
\end{align*}
See \citet{koenker_bassett_1978} or \citet[Theorem 1]{rockafellar2000} for a proof.  Therefore, using any function other than $Q_\tau$ yields a bound that is equal or larger.

As an aside, this derivation also directly motivates the consideration of more general sensitivity models.  \Cref{eq:lr_constraint} may be replaced by other ambiguity sets for $F_{\full}(y \mid x, 0)$, such as $f$-divergence or Wasserstein balls around $F(y \mid x, 1)$.  The Rosenbaum model considered in \citet{yadlowsky2018bounds} also falls in this class.  The duality between robust optimization and convex risk measures \citep{ruszczynski2006optimization} then implies the resulting adversarial regression would be the corresponding risk measure on the conditional outcome distribution.  For example, had we used the Kullback-Leibler divergence, we would have obtained the entropic value at risk \citep{ahmadi2012entropic} in place of CVaR.  Moreover, like CVaR, the resulting risk measure will be expressible as the value of a minimization problem.  Thus, identification formulae obtained in this way are always ``singly valid."

\section{The doubly-valid/doubly-sharp estimators} \label{sec:estimator}

In this section, we introduce the DVDS estimators and give some intuition for their properties, proving them formally in subsequent sections.  These estimators combine the ``propensity weighting" strategy from Section \ref{sec:eform} with the ``regression adjustment" strategy from Section \ref{sec:kform} to simultaneously enjoy the robustness properties of both Propositions \ref{prop:ipw_valid} and \ref{prop:dm_valid}.  This is analogous to how the AIPW estimator combines propensity weighting and regression adjustment to achieve double robustness under unconfoundedness.

The first step in computing the DVDS estimators is to learn the nuisance functions from Section \ref{sec:partial_id}.  We re-define these below in a generalized notation that accommodates either upper or lower bounds:
\begin{align*}
Q_{+}(x, z)&=Q_{\tau}(x, z) \\
Q_{-}(x, z)&=Q_{1-\tau}(x, z),\\
\varrho_{\pm}(x, z;\hat Q_\pm) &= \int \Lambda^{-1} y + (1 - \Lambda^{-1}) ( \hat Q_{\pm}(x, z) + \tfrac{1}{1 - \tau} \{ y - \hat Q_{\pm}(x, z) \}_{\pm}) \, \d F(y \mid x, z),\\
\rho_{\pm}(x, z)&=\varrho_{\pm}(x, z; Q_\pm).
\end{align*}
Let $\eta = (e, Q_+, Q_-, \rho_+, \rho_-)$ denote the full vector of DVDS nuisance functions.

While estimating the propensity score $e$ is a standard task in causal inference, estimating the other nuisances merits further discussion. 

For continuous outcomes, we recommend estimating $Q_{\pm}, \rho_{\pm}$ in the manner described in \cref{sec:kform}.  First, $Q_{\pm}$ is estimated by any quantile regression method \citep{quantile_neural_networks,quantile_random_forest,generalized_random_forests,belloni2011quantile,koenker_bassett_1978}. Then, $\rho_{\pm}$ is estimated using regression to learn the conditional mean of the \emph{estimated} transformed outcome using the estimated quantile.  In the Appendix, we show that this two-step approach can yield accurate estimates of $\rho_{\pm}$ despite first-step estimation error. If this second-step regression is performed by combining separate regression estimates of the $\mu$ and CVaR components, then, as $\Lambda$ tends to one, the DVDS estimator will tend to the AIPW estimator that uses the same $\mu$ and $e$-estimates.

For binary outcomes, estimating the additional nuisances is considerably simpler.  The quantiles and CVaR of a Bernoulli($\mu(x, z)$)-distributed random variable are explicit functions of $\mu(x, z)$, so the nuisances $Q_{\pm}, \rho_{\pm}$ can be computed by directly plugging an estimate $\hat{\mu}$ of the outcome regression into the following formulae:
\begin{align}
\begin{aligned}
Q_-(x, z) &=\mathbb{I} \{ \mu(x, z) > \tau \}, && \rho_-(x, z) = \max \{ 1 - \Lambda + \mu(x, z) \Lambda,\, \mu(x, z) / \Lambda \},\\
    Q_+(x, z) &= \mathbb{I} \{ \mu(x, z) > 1 - \tau\}, && \rho_+(x, z) = \min \{ 1 - 1/\Lambda + \mu(x,z)/\Lambda,\,  \mu(x, z) \Lambda \}.
\end{aligned}
\label{eq:binary_formulas}
\end{align}
When these formulae are used, the DVDS estimators tend to the AIPW estimator as $\Lambda$ tends to one and they tend to \citet{Manski1990}'s ``no assumptions" bounds as $\Lambda$ grows large.

For either continuous or binary outcomes, we assume that the nuisances are estimated using a standard sample-splitting technique known as cross-fitting.  This technique is described in \Cref{algo:dvds} and has been applied in \citet{schick1986, robins_etal2008, zheng2011cross, doubleML} (among others).  Cross-fitting allows black-box machine learning methods to be used for nuisance estimation.

The second step in computing the DVDS estimators is to aggregate the estimated nuisances using the following recentered influence functions:
\begin{align*}
\phi_1^{\pm}(x, y, z; \hat{\eta}) &= zy + (1 - z) \hat{\rho}_{\pm}(x, 1) + \frac{(1 - \hat{e}(x)) z}{\hat{e}(x)} ( \hat{Q}_{\pm}(x, 1) + \Lambda^{\pm \sign(  y - \hat{Q}_{\pm}(x, 1) )} ( y - \hat{Q}_{\pm}(x, 1) ) - \hat{\rho}_{\pm}(x, 1))\\
\phi_0^{\pm}(x, y, z; \hat{\eta}) &= (1 - z) y + z \hat{\rho}_{\pm}(x, 0) + \frac{\hat{e}(x)(1 - z)}{1 - \hat{e}(x)} ( \hat{Q}_{\pm}(x, 0) + \Lambda^{ \pm \sign ( y - \hat{Q}_{\pm}(x, 0) )} ( y - \hat{Q}_{\pm}(x, 0) ) - \hat{\rho}_{\pm}(x, 0))\\
\phi_{\ATE}^{\pm}(x, y, z; \hat{\eta}) &= \phi_1^{\pm}(x, y, z; \hat{\eta}) - \phi_0^{\mp}(x, y, z; \hat{\eta}).
\end{align*} \cref{algo:dvds} gives a complete description of this aggregation.

To provide some intuition for these recentered influence functions, we present two decompositions of $\phi_1^+$.  For conceptual clarity, we assume in the discussion below that the nuisances $\hat{\eta}$ are fixed rather than estimated functions.
\begin{enumerate}[label = (\roman*)]
    \item The first decomposition highlights the adversarial propensity-weighting characterization presented in Section \ref{sec:eform}:
\begin{align}
\begin{split}\label{eq:firstDecomposition}
\phi_1^+(x, y, z; \hat{\eta}) &= \underbrace{\frac{z \hat{Q}_+(x, 1)}{\hat{e}(x)} + \{ y - \hat{Q}_+(x, 1) \} z \left( 1 + \frac{1 - \hat{e}(x)}{\hat{e}(x)} \Lambda^{\sign \{ y - \hat{Q}_+(x, 1) \}} \right)}_a\\
&+ \underbrace{\left( 1 - \frac{z}{\hat{e}(x)} \right) \hat{\rho}_+(x, 1)}_b
\end{split}
\end{align}
If $\hat{e} = e$, then term $a$ recovers Proposition \ref{prop:ipw_valid}'s singly-valid ``weighting" formula for $\psi_1^+$, meaning $\E[a] = \psi_1^+$ if $\hat{Q}_+$ is correctly-specified while $\E[a] \geq \psi_1^+$ even if $\hat{Q}_+$ is misspecified.  Meanwhile, the term $b$ has mean zero when $\hat{e} = e$, no matter how accurate $\hat{\rho}_+$ is.  Thus, the DVDS estimator $\hat{\psi}_1^+$ should inherit both the sharpness and single validity properties from Section \ref{sec:eform}.

\item The second decomposition highlights the adversarial regression adjustment characterization presented in Section \ref{sec:kform}.
\begin{align}
\begin{split}\label{eq:secondDecomposition}
\phi_1^+(x, y, z; \hat{\eta}) &= \underbrace{zy + (1 - z) \hat{\rho}_+(x, 1)}_a\\
&+ \underbrace{\frac{\{ 1 - \hat{e}(x) \} z}{\hat{e}(x)} ( \Lambda^{-1} y + (1 - \Lambda^{-1}) [ \hat{Q}_{\tau}(x, 1) + \tfrac{1}{1- \tau} \{ y - \hat{Q}_{\tau}(x, 1) \}_+] - \hat{\rho}_+(x, 1))}_b
\end{split}
\end{align}
If $\hat{\rho}_+ = \varrho_+(\cdot,\cdot;\hat Q_\tau)$, then term $a$ is Proposition \ref{prop:dm_valid}'s singly-valid ``regression adjustment" formula for $\psi_1^+$, meaning $\E[a] = \psi_1^+$ if $\hat{Q}_+$ is correctly-specified while $\E[a] \geq \psi_1^+$ even if $\hat{Q}_+$ is misspecified.  Meanwhile, $b$ is has mean zero if $\hat{\rho}_+ = \varrho_+(\cdot,\cdot;\hat Q_\tau)$, no matter how accurate $\hat{e}$ is.  Thus, the DVDS estimator $\hat{\psi}_1^+$ should also inherit the properties from Section \ref{sec:kform}.
\end{enumerate}
In summary, we expect the estimator $\hat{\psi}_1^+$ to have two chances at estimating sharp bounds and, if quantiles are misspecified, two chances at estimating valid bounds.  

\begin{algorithm}[t!]
\caption{DVDS estimators}
\label{algo:dvds}
\begin{algorithmic}[1]
\STATEx\textbf{Input:} Data $\{(X_i,Z_i,Y_i):i=1,\dots,n\}$, number of folds $K$, estimand $\diamond \in \{ \T, \C, \ATE \}$
\STATE{Split the observations into $K$ (approximately) even-sized folds $\F_1, \ldots, \F_K$.}
\FOR{$k=1,\dots,K$}
\STATE{Estimate $\hat\eta^{(-k)} = (\hat{e}^{(-k)}, \hat{Q}_+^{(-k)}, \hat{Q}_-^{(-k)}, \hat{\rho}_+^{(-k)}, \hat{\rho}_-^{(-k)})$ using observations not in $\F_k$.  Use quantile regression followed by transformed-outcome regression if $Y$ is continuous, and use \cref{eq:binary_formulas} if $Y$ is binary.}
\ENDFOR
\STATE{\textbf{Return} the following estimators of $(\psi_{\diamond}^-, \psi_{\diamond}^+)$:
\begin{align*}
\begin{pmatrix}\hat{\psi}_{\diamond}^-\\ \hat{\psi}_{\diamond}^+\end{pmatrix} &= \frac{1}{n} \sum_{k = 1}^K \sum_{i \in \mathcal{F}_k} \begin{pmatrix} \phi_{\diamond}^-(X_i, Y_i, Z_i; \hat{\eta}^{(-k)})\\ \phi_{\diamond}^+(X_i, Y_i, Z_i; \hat{\eta}^{(-k)})\end{pmatrix},
\end{align*}}
\end{algorithmic}
\end{algorithm}

\section{Theory for estimation of bounds} \label{sec:point_estimation_guarantees}

In this section, we present the formal asymptotic properties of the DVDS point estimators.  For simplicity, we state our results only for the ATE bounds $(\hat{\psi}_{\ATE}^-, \hat{\psi}_{\ATE}^+)$, although the same results hold for the counterfactual mean bounds as well.

\subsection{Sharpness and validity}

Our first main result on estimation shows that the robustness properties intuitively derived in \Cref{sec:estimator} hold even when the nuisances are estimated from data.  This requires only mild regularity conditions, which are standard in semiparametric causal inference \citep{Kennedy2016, doubleML, TargetedLearningVDLRose}.

\begin{condition}[Primitives] \label{condition:primitives}
For some $q>2,\epsilon > 0$, we have $\E[|Y|^q] < \infty$ and $\P(\epsilon < e(X) < 1 - \epsilon) = 1$.
\end{condition}

\begin{condition}[Nuisance estimators] \label{condition:first_stage}
We have 
$\|\hat Q_\pm\|_q=O_P
(1)$, 
$\|\hat\rho_\pm\|_q=O_P
(1)$, 
$\P( \epsilon < \inf_x \hat{e}(x), \sup_x \hat{e}(x) < 1 - \epsilon) \rightarrow 1$.
\end{condition}

\noindent We suppress cross-fitting notation from this assumption and the ones that follow.  It should be understood that conditions stated in terms of ``$\hat \eta$'' are  imposed on each of $\hat\eta^{(-k)}$, for $k=1,\dots,K$.

\Cref{theorem:double_valid_double_sharp} below states that correct specification of either $e$ or $\rho_{\pm}$ suffices for the DVDS estimators to estimate the sharp ATE bounds when quantiles are correctly specified.  However, if quantiles are misspecified, then the DVDS estimators may still estimate \emph{valid} bounds on the ATE as long as the propensity score is well-specified or the transformed outcome regression accurately learns the conditional mean of the \emph{estimated} transformed outcome.

\begin{theorem}[Double validity and double sharpness]
\label{theorem:double_valid_double_sharp}
Assume Conditions \ref{condition:primitives} and \ref{condition:first_stage}.  Then:
\begin{enumerate}[label = (\roman*), itemsep=-1ex, topsep=-1ex]
    \item If quantiles are estimated consistently, then the DVDS estimators $(\hat{\psi}_{\ATE}^-, \hat{\psi}_{\ATE}^+)$ converge to the sharp bounds on the ATE if either $|| \hat{e} - e ||_{2} = o_P(1)$ or $|| \hat{\rho}_{\pm} - \rho_{\pm} ||_{2} = o_P(1)$. \label{item:double_sharpness}
    \item Even if quantiles are not estimated consistently, we still have $\hat{\psi}_{\ATE}^+ \geq \psi_{\ATE}^+ - o_P(1)$ and $\hat{\psi}_{\ATE}^- \leq \psi_{\ATE}^- + o_P(1)$ if either $\| \hat{e} - e \|_{2} = o_P(1)$ or $\| \hat{\rho}_{\pm} - \varrho_{\pm}(\cdot,\cdot;\hat Q_\pm) \|_{2} = o_P(1)$. \label{item:double_validity}
\end{enumerate}
\end{theorem}

While the ``double sharpness" result \ref{item:double_sharpness} is a relatively standard multiple robustness property, the ``double validity" result \ref{item:double_validity} appears to be new.  In our view, it enables uniquely credible sensitivity analysis by offering two chances at valid bounds even when consistency fails.  

For binary outcomes, the condition $|| \hat{\rho}_+ - \varrho_{\pm}(\cdot,\cdot;\hat Q_\pm) ||_2 = o_P(1)$ will be satisfied whenever the outcome regression $\mu$ is consistently estimated.  As a result, Theorem \ref{theorem:double_valid_double_sharp} for binary outcomes may be phrased entirely in terms of the outcome regression and propensity score. 

\begin{corollary} \label{corollary:binary_double_valid_double_sharp}
Suppose that $Y$ is binary and that Conditions \ref{condition:primitives} and \ref{condition:first_stage} hold, with nuisances estimated according to \Cref{eq:binary_formulas}.  Then the binary DVDS estimators yield sharp bounds on the ATE if $\hat{\mu}$ is consistent and yield valid bounds if either $\hat{\mu}$ or $\hat{e}$ are consistent.
\end{corollary}

\subsection{Efficiency}

Our second main result on estimation shows that the DVDS estimators are asymptotically efficient when all first-stage nuisance parameters are estimated consistently with certain rates of convergence.  These convergence rates accommodate nonparametric nuisance learners based on machine learning.

This result requires certain density assumptions.  In the continuous case, we assume that $Y$ has a bounded conditional density, as required by many standard quantile regression methods (\eg,  \citealt[Example 2]{generalized_random_forests}, or \citealt[Condition D.1]{belloni2011quantile}).  In the binary case, we assume $\mu(X, Z)$ has a bounded density.  This ensures that thresholding an accurate estimate of $\mu$ yields accurate estimates of $Q_{\pm}$. 

\begin{condition}[Densities] \label{condition:density}
Assume one of the following holds:
\begin{enumerate}[itemsep=-1ex, topsep=-1ex]
    \item \emph{Continuous case}.  For each $(x, z)$, $F(y \mid x, z)$ is continuous with a bounded density $f(y \mid x, z)$ that is positive on the interior of its support.
    \item \emph{Binary case}.  $Y$ is binary, and the regression function $\mu(X, Z)$ has a bounded density.
\end{enumerate}
\end{condition}

\begin{theorem}[Semiparametric efficiency] \label{theorem:efficiency}
Assume Conditions \ref{condition:primitives}, \ref{condition:first_stage}, and \ref{condition:density} and that all first-stage nuisance parameters are consistently estimated with the following convergence rates:
\begin{enumerate}[itemsep=-1ex]
    \item \emph{Continuous case}. $|| \hat{Q}_{\pm} - Q_{\pm} ||_{2} = o_P(n^{-1/4})$ and $|| \hat{e} - e ||_{2} \times || \hat{\rho}_{\pm} - \rho_{\pm} ||_{2} = o_P(n^{-1/2})$.
    \item \emph{Binary case}. $|| \hat{\mu} - \mu ||_{\infty} = o_P(n^{-1/4})$ and $ || \hat{e} - e ||_{2} \times || \hat{\mu} - \mu ||_{2} = o_P(n^{-1/2})$.
\end{enumerate}
Then the DVDS estimators are $\sqrt{n}$-consistent, asymptotically normal, and attain the semiparametric efficiency bound:
\begin{align}
\sqrt{n} \binom{\hat{\psi}_{\ATE}^- - \psi_{\ATE}^-}{\hat{\psi}_{\ATE}^+ - \psi_{\ATE}^+} &= \frac{1}{\sqrt{n}} \sum_{i = 1}^n \binom{\phi_{\ATE}^-(X_i, Y_i, Z_i; \eta) - \psi_{\ATE}^-}{\phi_{\ATE}^+(X_i, Y_i, Z_i; \eta) - \psi_{\ATE}^+} + o_P(1) \rightsquigarrow N ( 0, \mathbf{\Sigma}). \label{eq:asymptotic_normality}
\end{align}
\end{theorem}

The semiparametric efficiency bound $\mathbf{\Sigma}$ is the smallest asymptotic variance attainable by any estimator which converges locally uniformly to its asymptotic distribution.  Moreover, $\textup{trace}(\mathbf{\Sigma})$ is an asymptotic lower bound on the local minimax mean squared error attainable by \emph{any} estimator.  If the rate conditions and bounds assumed in Theorem \ref{theorem:double_valid_double_sharp} hold uniformly in a nonparametric model $\mathcal{P}$, the convergence in \cref{eq:asymptotic_normality} can be improved from locally uniform to globally uniform as in \citet[Theorem 3.1]{doubleML} or \citet[Theorem 1]{kallus2020localized}.  We skip this detail in the analysis for brevity.

The rate conditions in \cref{theorem:efficiency} will be satisfied if all first-stage nuisance parameters are estimated at rates faster than $n^{-1/4}$.  This is a standard requirement in semiparametric inference.  For the outcome regression, propensity score, and quantile regression nuisances, primitive conditions that imply this rate can be found in the rich literature on nonparametric function estimation (see \citealp{generalized_random_forests, belloni2011quantile, gyorfi_etal2002, wainwright_2019}, and references therein). 
Even the $L^{\infty}$ rate requirement on $\hat\mu$ needed to estimate (discrete) quantiles in the binary case is achievable in many cases \citep[\eg,][]{stone1982optimal}.

Convergence rates for $\hat{\rho}_{\pm}$ in the continuous case are slightly nonstandard since the transformed outcome is itself estimated from the data.  Fortunately, this regression problem satisfies a conditional Neyman orthogonality \citep{doubleML, foster2020orthogonal}:
\begin{align}
\frac{\partial}{\partial q} \int ( \Lambda^{-1} y + (1 - \Lambda^{-1})( q + \tfrac{1}{1 - \tau} \{ y - q \}_{\pm} )) \, f(y \mid x, z) \d y \, \bigg|_{q = Q_{\pm}(x, z)} = 0. \label{eq:kappa_neyman_orthogonal}
\end{align}
In the Appendix, we show that this property allows us to essentially ignore estimation error in the transformed outcome when deriving convergence rates for $\hat{\rho}_{\pm}$, meaning the rate $|| \hat{\rho}_{\pm} - \rho_{\pm} ||_2 = o_P(n^{-1/4})$ may be achieved even when $\hat{Q}_{\pm}$ is itself estimated by black-box methods.  \citet{olma2021truncated} obtains similar results for two-step CVaR estimation using local polynomial regression methods.

\section{Theory for inference on bounds} \label{sec:InferenceOnBounds}

In this section, we explain how the DVDS estimators can be used to set confidence limits on the endpoints of the partially-identified set or the entire set itself.  While all previous approaches for inference in the MSM have required computationally-intensive bootstrap procedures \citep{dornGuo2021sharp, soriano2021interpretable, zsb2019}, we propose using simple Wald-type intervals based on the following standard errors:
\begin{align*}
\hat{\sigma}_{\pm}^2 &= \frac{1}{n(n - 1)} \sum_{k = 1}^K \sum_{i \in \F_k} ( \phi_{\ATE}^{\pm}(X_i, Y_i, Z_i; \hat{\eta}^{(-k)}) - \hat{\psi}_{\ATE}^{\pm})^2
\end{align*}
Here, we recall that $\phi_{\ATE}^{\pm}$ is the estimated recentered influence function defined in \cref{sec:estimator}.

The following Theorem shows that confidence limits based on these standard errors are asymptotically valid and optimal under the conditions of Theorem \ref{theorem:efficiency}. This is a straightforward consequence of asymptotic normality and semiparametric efficiency, so we view it as the inferential analog of double sharpness.

\begin{theorem}[Sharp inference] \label{theorem:sharp_inference}
Assume the conditions of Theorem \ref{theorem:efficiency}.  Let $z_{1 - \alpha}$ be the $100(1 - \alpha)\%$ quantile of the standard normal distribution.  Then we have:
\begin{align}
\lim_{n \rightarrow \infty} \P ( \psi_{\ATE}^+ \leq \hat{\psi}_{\ATE}^+ + z_{1 - \alpha} \hat{\sigma}_+ ) = 1 - \alpha \quad \text{ for any } \alpha \in (0, 1). \label{eq:sharp_inference}
\end{align}
Moreover, for testing the null hypothesis $\psi_{\ATE}^+ \leq 0$ at level $\alpha$, the test that rejects when $\hat{\psi}_{\ATE}^+$ exceeds $z_{1 - \alpha} \hat{\sigma}_+$ is asymptotically most powerful in the sense of \citet[Theorems 25.44--45]{asymptotic_statistics}.
\end{theorem}

The theoretical guarantees for lower confidence bounds of the form $\hat{\psi}_{\ATE}^- - z_{1 - \alpha} \hat{\sigma}_-$ are analogous.  Optimal one-sided confidence limits for the partially identified also provide optimal one-sided confidence limits for the true (unidentified) ATE.  To obtain an asymptotic $100(1 -\alpha)\%$ two-sided confidence region for the entire identified set, one can intersect the $100(1 - \alpha/2)\%$ upper and lower confidence regions.  By the union bound, this will cover the identified set with asymptotic probability at least $1 - \alpha$.  It is possible to construct slightly refined intervals with asymptotic coverage exactly $1 - \alpha$ by accounting for the correlations between $\hat{\psi}_{\ATE}^+, \hat{\psi}_{\ATE}^-$ \citep[Corollary 3]{kallus2021assessing}.  However, we expect the scope of this refinement to be limited in practice as we typically see strong positive correlations between the bounds in our simulations.

One dissatisfying feature of \cref{theorem:sharp_inference} is that it requires stronger rate conditions than those typically used for inference under unconfoundedness.  For example, inference based on the usual AIPW estimator does not require any quantile regression or $L^{\infty}$ rate conditions.

Thus, our second inference result considers what happens when these extra rate conditions are not satisfied.  In other words, we allow for the quantiles $Q_{\pm}$ to be estimated at a rate slower than $n^{-1/4}$, or even to be completely misspecified.  In the binary case, this amounts to assuming only the standard product-of-rates condition used in the analysis of the AIPW estimator, with no $L^{\infty}$ rate condition.  Because this result does not assume quantile consistency, we view it as the inferential analogue of double validity.

\begin{theorem}[Valid inference] \label{theorem:valid_inference} Assume Conditions \ref{condition:primitives}, \ref{condition:first_stage}, and \ref{condition:density} and the following:
\begin{enumerate}[itemsep=-1ex]
    \item \emph{Continuous case}.  $|| \hat{e} - e ||_2$ and $|| \hat{\rho}_{\pm} - \varrho_{\pm}(\cdot, \cdot; \hat{Q}_{\pm} ) ||_2$ converge to zero at rate $|| \hat{e} - e ||_2 \times || \hat{\rho}_{\pm} - \varrho_{\pm}(\cdot, \cdot; \hat{Q}_{\pm} ) ||_2 = o_P(n^{-1/2})$ and $\hat{Q}_{\pm}$ tends to some limit $\bar{Q}_{\pm}$ at any rate.
    \item \emph{Binary case}. $|| \hat{e} - e ||_2$ and $|| \hat{\mu} - \mu ||_2$ tend to zero at rate $|| \hat{e} - e ||_2 \times || \hat{\mu} - \mu ||_2 = o_P(n^{-1/2})$.
\end{enumerate}
Then $\hat{\psi}_{\ATE}^+$ may have non-Gaussian fluctuations, converge at an arbitrarily slow rate, or even be inconsistent.  Nevertheless, Wald confidence intervals remain asymptotically valid:
\begin{align}
\liminf_{n \rightarrow \infty} \P( \psi_{\ATE}^+ \leq \hat{\psi}_{\ATE}^+ + z_{1 - \alpha} \hat{\sigma}_+) \geq 1 - \alpha \quad \text{for any } \alpha \in (0, 1). \label{eq:valid_inference}
\end{align}
\end{theorem}

In the continuous case, \cref{theorem:valid_inference} assumes that $\hat{\rho}_{\pm}$ is not too far from $\varrho_{\pm}(\cdot, \cdot, \hat{Q}_{\tau})$, the conditional mean of the \emph{estimated} transformed outcome.  We show in the appendix that this is achievable even when the quantile regression is misspecified or converges at a much slower rate.

In our view, \cref{theorem:valid_inference} is surprising since Wald-type confidence intervals are typically invalid when asymptotic normality fails.  We prove this result by showing that the Wald upper confidence bound based on $\hat{\psi}_{\ATE}^+$ is first-order equivalent to the confidence bound obtained from a certain bootstrap scheme.  Then, we couple the bootstrap distribution of the DVDS estimator with that of an infeasible but asymptotically normal estimator $\bar{\psi}_{\ATE}^+$ to show that the quantiles of the former bootstrap distribution are always larger than the quantiles of the latter.  Since $\bar{\psi}_{\ATE}^+$ is a well-behaved estimator, its bootstrap quantiles are asymptotically valid confidence limits.  Therefore, the bootstrap quantiles based on the DVDS estimator are as well.

We make a technical remark on uniform validity.  Under the stronger assumptions of \Cref{theorem:sharp_inference}, the DVDS estimators are regular so the convergence in \Cref{eq:sharp_inference} is automatically uniform over local (contiguous) neighborhoods of $P$.  Under the weaker assumptions of \Cref{theorem:valid_inference}, the DVDS bounds may no longer be regular, so one may expect that the convergence in \Cref{eq:valid_inference} is only pointwise.  However, this is not the case.  Even when the DVDS estimators are nonregular, Wald confidence bounds retain their aforementioned uniform validity. We sketch the argument in our proof of \Cref{theorem:valid_inference}.

\section{Numerical illustration} \label{sec:experiments}

We next demonstrate DVDS in simulated examples and a case study.

\subsection{Simulated examples}

In this section, we illustrate the DVDS estimators in two simulated examples.
\begin{itemize}
    \item \emph{Binary outcomes}.  In our first example, we compute DVDS bounds for the ATE based on $n = 1{,}000$ observations from the following data-generating process:
    \begin{align*}
    X &\sim \textup{Uniform}( [-1, 1]^5)\\
    Z \mid X &\sim \textup{Bernoulli} \left( \frac{1}{1 + \exp(X_1 + 0.5 \mathbb{I} \{ X_2 > 0 \} + 0.5 X_2 X_3)} \right)\\
    Y \mid X, Z &\sim \textup{Bernoulli} \left( \frac{1}{1 + \exp(0.5 X_1 + X_2 + 0.25 X_2 X_3)} \right).
    \end{align*}
    We compute bounds for a range of $\Lambda$ values between one and three.  For nuisance estimation, we use random forest regression to estimate the outcome regression and propensity score, then form quantile and CVaR estimates using \Cref{eq:binary_formulas}.  
    
    \item \emph{Continuous outcomes}. Our second example mimics our first example, except the conditional distribution of $Y$ is replaced by a heteroscedastic normal distribution:
    \begin{align*}
    Y \mid X, Z &\sim N ( 2 \sign(X_1) + X_2 + X_2 X_3,\, (1 + X_4^2)^2).
    \end{align*}
    In this case, we first estimate $Q_{\pm}$ and then fit $\textup{CVaR}_{\tau}(x, z)$ using transformed-outcome regression.  Then, we combine the CVaR estimate with a separate random forest regression estimate of $\mu(x, z)$ to learn $\rho_{\pm}(x, z)$.  
    This separate-regressions approach ensures that the DVDS estimators recover the AIPW estimator when $\Lambda = 1$. 
    To avoid re-estimating a separate regression for each value of $\Lambda$, we use the same random forest weights for all quantile regressions and the same random forest weights for all transformed-outcome regressions, with transformed-outcome weights coming from the regression corresponding to $\Lambda=2$.  
\end{itemize}
In both cases, in a slight departure from our formal results, we use out-of-bag predictions instead of cross-fitting.  This approach is also taken in \citet{athey2017efficient}.
All forests are fit using \textit{R} package \texttt{grf} \citep{generalized_random_forests}.

These computations were performed $1,000$ times in each example.  For each replication, we also computed 95\% Wald confidence intervals based on the DVDS estimators.  For comparison, we also obtained bounds using the AIPW sensitivity analysis for the MSM proposed in \citet[Section 6.2]{zsb2019}.  The estimated bounds are plotted in \Cref{fig:simulation_results}.  

In the binary case, DVDS estimators are approximately unbiased for all considered values of $\Lambda$.  Moreover, nominal level $95\%$ two-sided Wald confidence bounds had coverage ranging from $92.2\%$ (when $\Lambda = 1$) to $96.5\%$ (when $\Lambda = 2.1$) with average coverage of $95.3\%$.  Meanwhile, the \citet{zsb2019} bounds are quite conservative in the binary outcome simulation. 

In the continuous case, the DVDS upper bounds perform well but the lower bounds are somewhat conservative.  Here, it appears that the random forest methods used for quantile and regression estimation did not adapt well to the outcome model used.  Nevertheless, nominal level 95\% two-sided Wald confidence bounds had reasonable coverage ranging from $91.7\%$ (when $\Lambda = 1$) to $98.1\%$ (when $\Lambda = 3$) with average coverage of $96.1\%$.  The DVDS bounds improve only slightly over the \citet{zsb2019} bounds in this example. The difference may be more pronounced in examples with more heteroscedasticity.

In summary, these experiments largely validate the theoretical results presented in Sections \ref{sec:point_estimation_guarantees} and \ref{sec:InferenceOnBounds}.  When nuisances are estimated well, DVDS point estimates and confidence intervals yield sharp inferences on the partially identified set even when machine learning methods are used.  Meanwhile, when some nuisances are estimated poorly, the procedure errs on the side of conservatism.

\begin{figure}
    \centering
    \includegraphics[width=16cm]{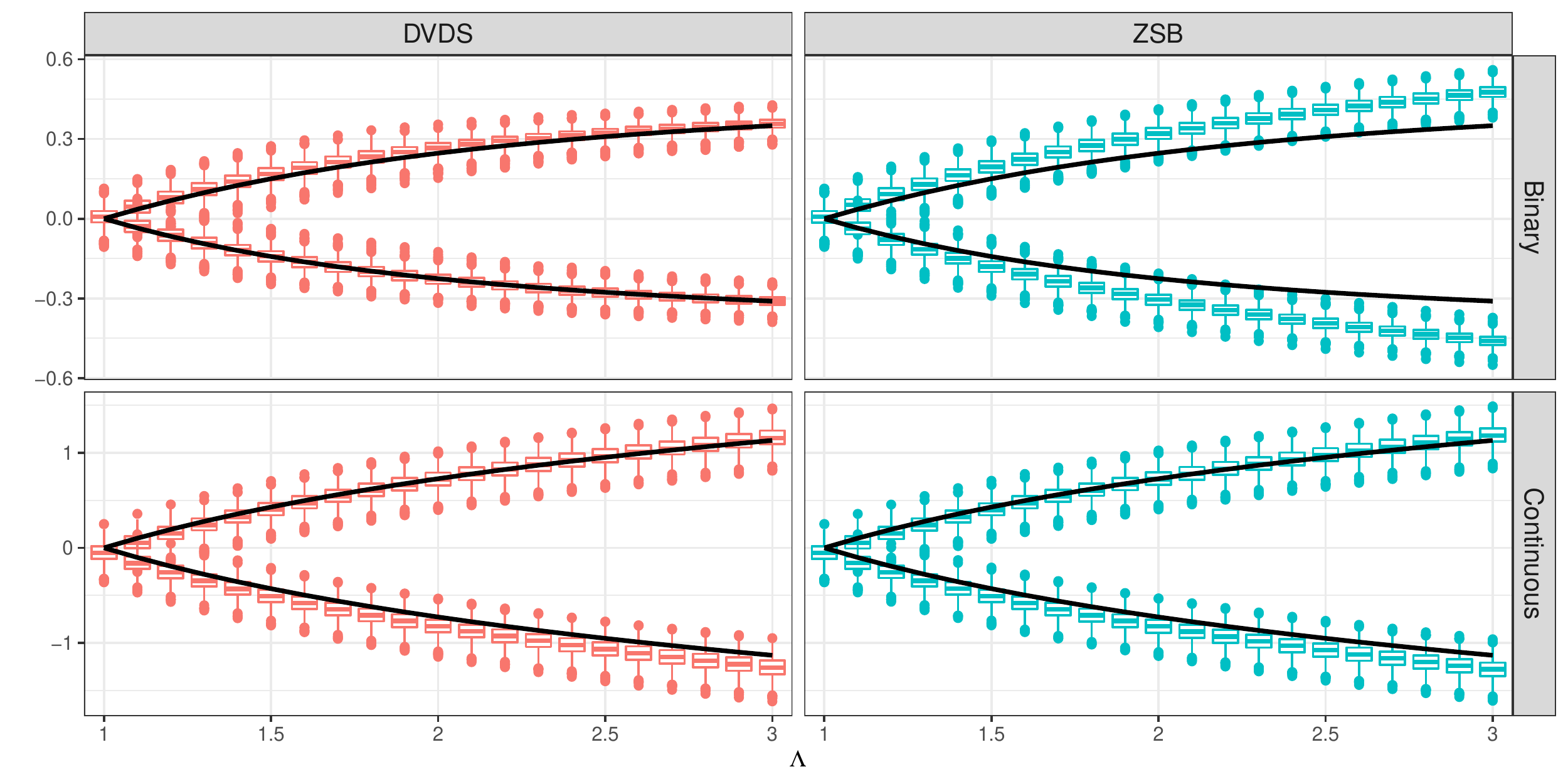}
    \caption{Boxplots of the estimated ATE bounds in simulated examples.  The solid line indicates the upper and lower bounds of the true partially identified set.  \emph{DVDS} is the doubly-valid/doubly-sharp method introduced in this paper.  \emph{ZSB} is the method from Section 6.2 of \citet{zsb2019}.}
    \label{fig:simulation_results}
\end{figure}

\subsection{Application to right heart catheterization}

We now apply the DVDS estimators in a real-data example.  Specifically, we revisit \citet{ConnorsEtAl1996}'s influential finding that right heart catheterization (RHC) leads to lower survival rates among critically ill patients.  Their data has been extensively analyzed in the causal inference literature (\textit{e.g.} \citet{tan2020, ATTViaIPW, cui_tchetgen2019, lin1998}), and was the first real-data application ever studied using the MSM \citep{tan2006}.

We use a version of the data from \citet{BWSnooping}.  It consists of 5{,}735 observations on adult patients from five US hospital centers.  The treatment $Z$ indicates whether a patient received RHC within 24 hours of admission, and the outcome $Y$ is 30-day survival rate.  The data also contains measurements on a rich set of demographic and medical factors $X$ measured prior to treatment.  For a complete list, see \citet[Tables 1 and 2]{ATTViaIPW}.  Our estimand of interest is the ATE.

In our analysis, we estimate the propensity score $\hat{e}$ and outcome regression $\hat{\mu}$ using the Super Learner ensemble method \citep{SuperLearner}.  As base learners, we use logistic regression, random forest, gradient boosted trees, and support vector machine with Platt scaling.

The AIPW estimator based on our estimated nuisances yields an ATE estimate of $-5.4\%$ with a 95\% confidence interval of $[-7.8\%, -3\%]$.  Thus, assuming unconfoundedness, we find that RHC has a negative average effect on 30-day survival.  \citet{cui_tchetgen2019, tan2020} also apply machine learning methods to this dataset and obtain similar estimates.

We compare the DVDS method to the AIPW-based sensitivity analysis from \citet{zsb2019}.  Since the outcome $Y$ in this example is binary, no nuisance functions beyond the propensity score and outcome regression are needed for either sensitivity analysis.  We also compute (pointwise) 95\% confidence intervals based on the DVDS estimator.  

Estimated bounds for various values of $\Lambda$ are shown in \Cref{fig:RHC_Sensitivity}.  Both sensitivity analyses find that fairly small amounts of unobserved confounding would suffice to explain away the negative point estimate ($\Lambda = 1.17$ for the \citet{zsb2019} method, $\Lambda = 1.23$ for the DVDS method).  If we further account for statistical variability in the estimated bounds, the DVDS estimators would find statistical evidence for a negative treatment effect only under the assumption that unobserved confounders cannot change the odds of treatment by up to a factor of $\Lambda = 1.12$.  Our bound estimates are comparable to those of \cite{tan2006} and suggest more caution is warranted than is conveyed by \cite{ConnorsEtAl1996}'s original sensitivity analysis.

\begin{figure}
    \centering
    \includegraphics[width=.7\textwidth]{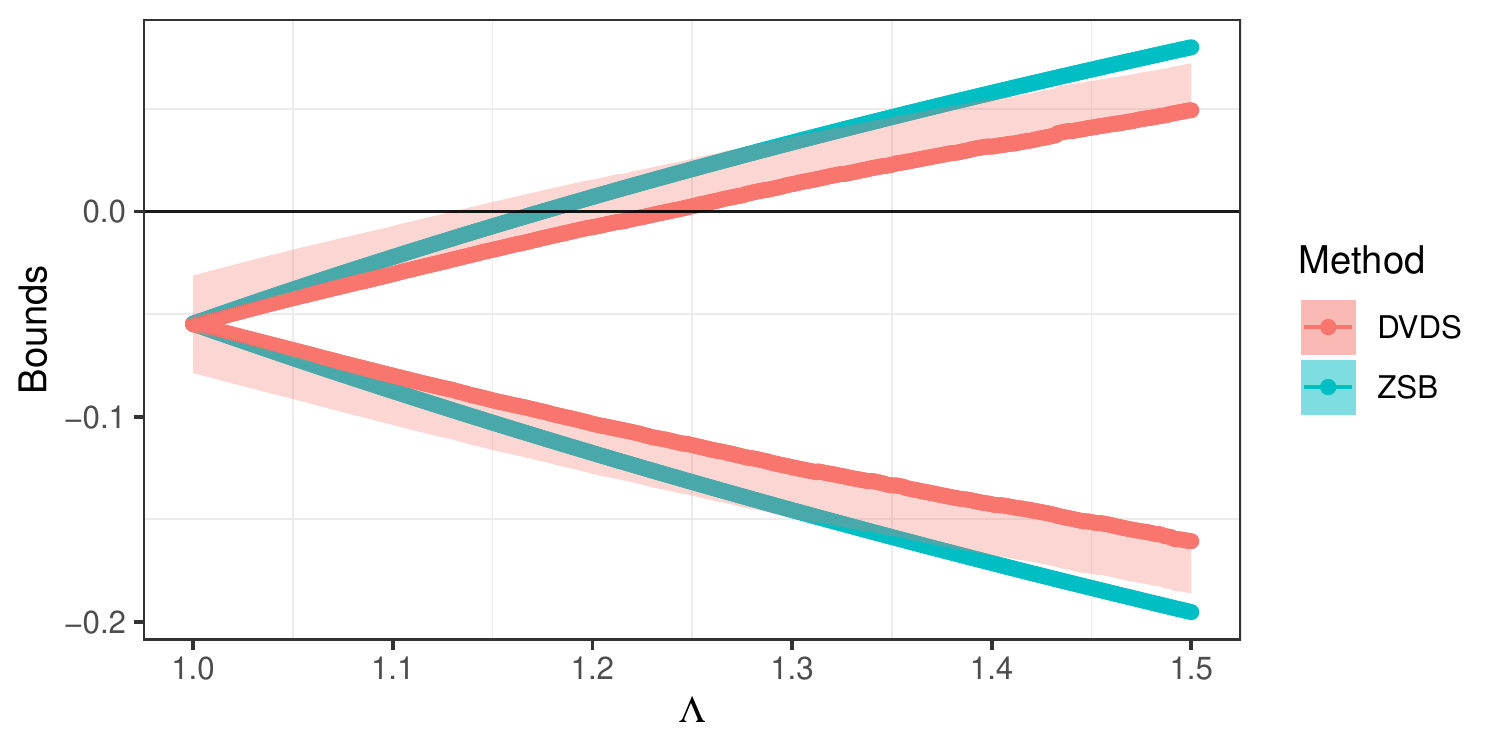}
    \caption{Sensitivity analysis results for ATE. Dots reflect point estimates for bounds, with 95\% DVDS confidence interval indicated by shaded region. Zero effect is indicated with the solid horizontal line.}
    \label{fig:RHC_Sensitivity}
\end{figure}

\section{Discussion} \label{sec:discussion}

This paper considered estimation of bounds on the ATE (and other causal quantities) under the assumption that measuring unobserved confounders does not change the odds of treatment by more than a factor of $\Lambda$.  We characterized the sharp partial identification bounds in terms of CVaR regression and used this characterization to propose robust estimators of the partially identified set.

Perhaps the most novel phenomenon observed in our work is double validity and the associated inference-without-normality result of \Cref{theorem:valid_inference}.  As far as we can tell, neither of these properties has been observed before, possibly because they do not arise in point-identified problems.

Since the literature on semiparametric estimation of partial identification bounds is relatively nascent \citep{tchetgen_tchetgen_shpitser2012, yadlowsky2018bounds, BonviniKennedy, semenova2021generalized, doubleMLSensitivity, kallus2021assessing}, we expect that single and possibly double validity will be recurring phenomenon as the field grows.
Since the first draft of this paper, double validity has even been observed in partial identification outside of sensitivity analysis \citep{kallus2022treatment,kallus2022s}.
For sensitivity analysis, as discussed in \cref{sec:kform}, single validity should arise in any sensitivity model that bounds a convex distributional divergence between $\Pfull(Y(1) \mid X = x, Z = 0)$ and $\Pobs(Y(1) \mid X = x, Z = 1)$.
The MSM falls into this class, but numerous other choices have been popularized in the growing literature on distributionally robust optimization \citep[\eg,][]{ben2013robust,bertsimas2018robust,esfahani2018data}.  This yields a rich family of sensitivity models that may be appropriate for different applications.  
Since the first draft of this paper, \citet{jin_etal2022} have explored such alternatives using $f$-divergences and indeed obtain single validity guarantees.

\section*{Acknowledgements}

This material is based upon work supported by the National
Science Foundation under Grant No. 1939704 and by the National Science Foundation Graduate Research Fellowship Program under Grant No. DGE-2039656. Any opinions, findings, and conclusions or recommendations expressed in this material are those of the authors and do not necessarily reflect the views of the National Science Foundation. The authors thank Angela Zhou for helpful discussions while conceiving the idea for this paper and are grateful for comments from Bo Honor{\'e}, Michal Koles\'{a}r, and Mikkel Plagborg-Møller.

\newpage
\bibliographystyle{chicago}
\bibliography{bibliography.bib}

\newpage
\appendix

\section{Average treatment effect on the treated}

In this appendix, we extend our method to the average treatment effect on the treated (ATT) $\psi_{\ATT}(\Pfull) := \E_{\Pfull}[Y(1) - Y(0) \mid Z = 1]$.  We let $\psi_{\ATT}^-$ and $\psi_{\ATT}^+$ denote the sharp lower and upper bounds for the ATT under the marginal sensitivity model.

The DVDS point estimators for the sharp ATT bounds are defined as:
\begin{align*}
\hat{\psi}_{\ATT}^+ = \{ \bar{Y} - \hat{\psi}_0^- \} / \bar{Z} \quad \text{and} \quad \hat{\psi}_{\ATT}^- = \{ \bar{Y} - \hat{\psi}_0^+ \} / \bar{Z}.
\end{align*}
Here, $\bar{Y} = \sum_{i = 1}^n Y_i / n$, $\bar{Z} = \sum_{i = 1}^n Z_i / n$, and $\hat{\psi}_0^-, \hat{\psi}_0^+$ are defined in \Cref{algo:dvds} of \Cref{sec:estimator}.  These estimators are based on the relation $\psi_0^{\pm} = \{ \E[Y] - \psi_0^{\mp} \} / \E[Z]$ established in \citet[Corollary 1]{dornGuo2021sharp}.  

For statistical inference, we recommend using Wald confidence intervals based on the following standard errors:
\begin{align*}
\hat{\sigma}_{\ATT,-}^2 &= \frac{1}{N_1} \frac{1}{N_1 - 1} \sum_{k = 1}^K \sum_{i \in \F_k} (Y_i - \phi_0^+(X_i, Y_i, Z_i; \hat{\eta}^{(-k)} ) - Z_i \hat{\psi}_{\ATT}^-)^2\\
\hat{\sigma}_{\ATT,+}^2 &= \frac{1}{N_1} \frac{1}{N_1 - 1} \sum_{k = 1}^K \sum_{i \in \F_k} (Y_i - \phi_0^-(X_i, Y_i, Z_i ; \hat{\eta}^{(-k)}) - Z_i \hat{\psi}_{\ATT}^+)^2
\end{align*}
Here, $N_1 = \sum_{i = 1}^n Z_i$ is the number of treated units.  These standard errors can be naturally derived from \Cref{theorem:efficiency} and the delta method.

All of our theoretical results for the DVDS estimators of the sharp ATE bounds also extend to the DVDS estimators of the sharp ATT bounds.  In particular, the estimators $\hat{\psi}_{\ATT}^-, \hat{\psi}_{\ATT}^+$ are doubly sharp and doubly valid, and the associated confidence intervals are efficient under nonparametric rate conditions and valid even when some of these rate conditions fail.  Since the proofs of these facts are nearly identical to those for the ATE, we omit the details.

\section{Rates for transformed-outcome regression}

In this appendix, we derive convergence rates for the transformed-outcome regression $\hat{\rho}_+$ when $Y$ has a continuous distribution.  Our main goal is to illustrate the following points:
\begin{enumerate}[label=(\roman*), itemsep=-1ex]
    \item Standard assumptions and proof techniques for nonparametric regression/empirical risk minimization apply in this problem, even though the transformed-outcome is itself estimated from the data.
    \item It is possible to estimate $\rho_+$ and $\varrho_+(\cdot, \cdot; \hat{Q}_+)$ with a relatively fast convergence rate, even when the conditional quantile $Q_+$ is estimated by a black-box model with a relatively slow convergence rate.
\end{enumerate}

The specific estimation strategy we analyze is sieved least-squares with sample splitting.  We assume that the observations are divided into two sets of indices $\mathcal{I}_1, \mathcal{I}_2$ with $| \mathcal{I}_2 | \asymp n$.  The analyst first uses observations with indices in $\mathcal{I}_1$ to estimate a quantile model $\hat{Q}_+$, and then uses observations with indices in $\mathcal{I}_2$ to fit the following least-squares regression:
\begin{align}
\hat{\rho}_+ &= \argmin_{g \in \mathcal{G}_n} \sum_{i \in \mathcal{I}_2} (  \Lambda^{-1} Y_i + (1 - \Lambda^{-1})[ \hat{Q}_+(X_i, 1) + \tfrac{1}{1 - \tau} \{ Y_i - \hat{Q}_+(X_i, 1) \}_+] - g(X_i, Z_i) )^2. \label{eq:sieved_least_squares}
\end{align}
In \Cref{eq:sieved_least_squares}, the hypothesis class $\mathcal{G}_n$ may be a parameterized class (\eg, the neural networks from \citealt{farrell_etal2021}) or a nonparametric class (\eg, all monotone functions, or functions of bounded sectional variation as in \citealt{bibaut2019fast}).

We use sample splitting to get convergence rates under only high-level assumptions on the first-step quantile regression estimator.  \citet{foster2020orthogonal, kennedy2020optimal, yadlowsky2018bounds, nie_wager2020} do the same.  In practice, the non-sample-split version of (\ref{eq:sieved_least_squares}) may work well under further assumptions on $\hat{Q}_+$.  See \citet{olma2021truncated} for an example.

In our analysis, we make use of the following technical assumptions.

\begin{condition}
\label{condition:erm}
Assume the following conditions hold for some $\gamma \in (0, \infty)$, $\alpha \in (2, 4)$ and $B < \infty$.
\begin{enumerate}[label=(\alph*), itemsep=-1ex, topsep=-1ex]
    \item \emph{Boundedness}.  $|Y|$, $|| \hat{Q}_+ ||_{\infty}$, and $\sup_{g \in \mathcal{G}_n} || g ||_{\infty}$ are almost surely bounded by $B$.
    \item \emph{Capacity}. $\mathcal{G}_n$ is pointwise measurable, and its $L^2$ bracketing entropy satisfies $H(s, \mathcal{G}_n) \leq \log(n)^{\gamma} / s^{\alpha - 2}$.
    \item \emph{Quantile convergence}.  $|| \hat{Q}_+ - \bar{Q}_+ ||_2 = o_P(1)$ for some square-integrable function $\bar{Q}_+$.
    \item \emph{Approximation}. $\inf_{g \in \mathcal{G}_n} || g - \bar{\varrho}_+ ||_2 = O(n^{-1/\alpha})$, where $\bar{\varrho}_+ = \varrho_+(\cdot, \cdot; \bar{Q}_+)$.
\end{enumerate}
\end{condition}

The only nonstandard item is the approximation condition.  It asks that the sieve $\mathcal{G}_n$ approximates the conditional mean of the \emph{limiting} estimated transformed outcome.  An alternative approach might directly assume that the sieve $\mathcal{G}_n$ approximates $\varrho_+(\cdot, \cdot; \hat{Q}_+)$.  However, we feel that assumption is less interpretable, since $\varrho_+(\cdot, \cdot; \hat{Q}_+)$ is itself a random function which depends on a black-box quantile model.  In contrast, $\bar{\varrho}_+$ is a fixed function that depends only on the limit of the quantile regression, which will often be the true quantile $Q_+$.  In this case, the approximation requirement in Condition \ref{condition:erm} licences taking $\mathcal{G}_n$ to be a sieve of smooth functions (e.g. splines) as long as the true outcome regression and CVaR are smooth, even if $\hat{Q}_+$ is discontinuous (e.g. the output of a quantile regression forest).

\begin{proposition}
\label{theorem:transformed_outcome_rate}
Assume Condition \ref{condition:density} (continuous case) and Condition \ref{condition:erm}.  Then the following hold:
\begin{enumerate}[itemsep=-1ex, topsep=-1ex, label=(\roman*)]
    \item If $|| \hat{Q}_+ - Q_+ ||_2 = O_P(n^{-1/\alpha})$, then $|| \hat{\rho}_+ - \rho_+ ||_2 = o_P(n^{-1/4})$.
    \item If $|| \hat{Q}_+ - Q_+ ||_4 = O_P(n^{-1/2\alpha})$ or $|| \hat{Q}_+ - \bar{Q}_+ ||_2 = O_P(n^{-1/\alpha})$, then $|| \hat{\rho}_+ - \varrho_+(\cdot, \cdot; \hat{Q}_+) ||_2 = o_P(n^{-1/4})$.
\end{enumerate}
\end{proposition}

The first conclusion of this Theorem is useful for our results on sharp estimation and inference.  It states that if the quantile regression converges at a rate slightly faster than $n^{-1/4}$, then least-squares regression based on a sieve with good approximation properties also converges faster than $n^{-1/4}$.

The second conclusion is useful for valid estimation and inference.  It shows that a fast rate on $|| \hat{\rho}_+ - \varrho_+(\cdot, \cdot; \hat{Q}_+) ||_2$ can be attained as long as the initial quantile regression converges to the right limit at rate slightly faster than $n^{-1/8}$ or to the wrong limit at a rate slightly faster than $n^{-1/4}$.  This latter case is relevant when the data analyst uses an insufficiently expressive parametric model for quantile regression.

\section{Proofs}

In this appendix, we prove all the results from the main text.  We adopt the following notational conventions.  The $K$ folds used in cross fitting at sample size $n$ will be denoted by $\F_{1,n}, \ldots, \F_{K,n}$, although the $n$ subscript will typically be suppressed.  We assume that $| \F_{k,n} | / n \sim 1/K$ so that the folds are approximately equally sized.  We write $\E_n[\cdot]$ as shorthand for the empirical expectation $\sum_{i = 1}^n[\cdot]_i/n$, $\E_{n,k}[\cdot]$ as shorthand for the fold-$k$ average $\sum_{i \in \F_{k,n}} [\cdot]_i / | \F_{k,n}|$, and $|| \cdot ||_{2,k}$ for the empirical $L^2$ norm in fold $k$, $|| \cdot ||_{2,k} = \E_{n,k}[(\cdot)^2]^{1/2}$.  We let $\F_{-k} = \sigma( \{ (X_i, Y_i, Z_i) \, : \, i \leq n, i \not \in \F_k \})$ denote the information contained in observations not in $\F_k$.  For two sequences $A_n, B_n$, we write $A_n \precsim B_n$ if $A_n = O(B_n)$ and $A_n \precsim_P B_n$ if $A_n = O_P(B_n)$.

\subsection{Preparation}

\begin{lemma}
\label{lemma:transformed_outcome_identity}
The equality $q + \Lambda^{\sign(y - q)}(y - q) = \Lambda^{-1} y + (1 - \Lambda^{-1}) \{ q + \tfrac{1}{1 - \tau}(y - q)_+ \}$ holds for any $y$ and $q$.
\end{lemma}

\begin{proof}
This can be verified by separately considering the cases $y > q$, $y = q$, and $y < q$.
\end{proof}

\begin{lemma}
\label{lemma:cvar_neyman_orthogonal}
Assume Conditions \ref{condition:primitives}, \ref{condition:density}, and that $Y$ is continuous.  Let $M < \infty$ be a uniform bound on the density $f(y \mid x, z)$.  Then for any integrable function $\hat{Q}_+$, we have:
\begin{align*}
\left| \int_{\R} [\hat{Q}_+(x, 1) + \tfrac{1}{1 - \tau} \{ y - \hat{Q}_+(x, 1) \}_+ - Q_+(x, 1) - \tfrac{1}{1 - \tau} \{ y - Q_+(x, 1) \}_+] f(y \mid x, 1) \, \d y \right| \leq \Lambda M | Q_+(x, 1) - \hat{Q}_+(x, 1)|^2 
\end{align*}
\end{lemma}

\begin{proof}
Define the function $g$ by $g(q) = \int q + \tfrac{1}{1 - \tau} \{ y - q \}_+ f(y \mid x, 1) \, \d y$.  Since the integrand in the definition of $g$ is a convex function of $q$, we may differentiate through the integral sign and conclude:
\begin{align*}
g'(q) = \int_{\R} 1 - \frac{1}{1 - \tau} \mathbb{I} \{ y > q \} f(y \mid x, 1) \, \d y \quad \text{and} \quad 
g''(q) = \frac{f(q \mid x, 1)}{1 - \tau}
\end{align*}
The first derivative satisfies $g'(Q_+(x, 1)) = 0$ and the second derivative is uniformly bounded by $(\Lambda + 1) M \leq 2 \Lambda M$.  Hence, $| g(Q_+(x, 1)) - g(\hat{Q}_+(x, 1))| \leq \Lambda M | Q_+(x, 1) - \hat{Q}_+(x, 1)|^2$ follows by controlling the second-order remainder in a first-order Taylor expansion of $g$ around $Q_+(x, 1)$.
\end{proof}

\begin{lemma}[Nuisance parameter convergence results]
\label{lemma:nuisance_facts}
Under Conditions \ref{condition:primitives} and \ref{condition:first_stage}, the following hold for every $k \in \{ 1, \ldots, K \}$:
\begin{enumerate}[label = (\arabic*), itemsep=-1ex]
    \item $|| \hat{Q}_{+}^{(-k)}(x, 1) ||_{2,k} = O_P(1)$ and $|| \hat{\rho}_{+}^{(-k)}(x, 1) ||_{2,k} = O_P(1)$. \label{item:Q_kappa_stochastically_bounded}
    \item If $|| \hat{Q}_+^{(-k)} - \bar{Q}_{+} ||_2 = o_P(1)$, then $|| \hat{Q}_{+}^{(-k)}(x, 1) - \bar{Q}_+(x, 1) ||_{2,k} = o_P(1)$ as well. \label{item:Q1_convergence}
    \item If $|| \hat{\rho}_+^{(-k)} - \varrho_+(\cdot, \cdot; \hat{Q}_+^{(-k)}) ||_2 = o_P(1)$, then $|| \hat{\rho}_+^{(-k)}(x, 1) - \varrho_+(x, 1; \hat{Q}_+^{(-k)}) ||_{2,k} = o_P(1)$ as well. \label{item:kappa1_convergence}
    \item If $|| \hat{e}^{(-k)} - e ||_2 = o_P(1)$, then $|| 1/ \hat{e}^{(-k)} - 1/e ||_{2,k} = o_P(1)$ as well.
    \item If $Y$ is binary, then $|| \hat{\rho}_+ - \rho_+ ||_2 \leq \Lambda || \hat{\mu}^{(-k)} - \mu ||_2$. \label{item:mu_to_kappa}
    \item Let $\bar{\varrho}_+$ be as in Condition \ref{condition:erm}.  Then $|| \varrho_+(x, 1; \hat{Q}_+^{(-k)}) - \bar{\varrho}_+(x, 1) ||_2 \leq \Lambda || \hat{Q}_+(x, 1) - \bar{Q}_+(x, 1) ||_2$. \label{item:kappa_lipschitz}
    \item If $|| f_n ||_{q} = O_P(1)$ for some $q > 2$ and $|| \hat{e}^{(-k)} - e ||_2 = o_P(1)$, then $|| f_n/\hat{e}^{(-k)} - f_n/e ||_2 = o_P(1)$ as well. \label{item:f_over_e}
    \item If $Y$ is binary and $P( \mu(X, Z) = 1 - \tau) = 0$, then $|| \hat{\mu}^{(-k)} - \mu ||_2 = o_P(1)$ implies $|| \hat{Q}_+^{(-k)} - \hat{Q}_+ ||_2 = o_P(1)$ as well. \label{item:mu_to_Q}
\end{enumerate}
\end{lemma}

\begin{proof}
We prove these one at a time:
\begin{enumerate}[label = (\arabic*), itemsep=-1ex]
    \item Condition \ref{condition:primitives} requires $e(X) \geq \epsilon > 0$, which allows us to write: 
    \begin{align*}
    || \hat{Q}_+^{(-k)}(x, 1) ||_{2,k}^2 &= \E[ \{\hat{Q}_+^{(-k)}(X, 1) \}^2 \mid \F_{-k} ]\\
    &= \E[ \{ \hat{Q}_+^{(-k)}(X, Z) \}^2 Z/ e(X) \mid \F_{-k}]\\
    &\leq \E[ \{ \hat{Q}_+^{(-k)}(X, Z) \}^2 \mid \F_{-k}] / \epsilon\\
    &= || \hat{Q}_+^{(-k)} ||_2 / \epsilon.
    \end{align*}
    Condition \ref{condition:first_stage} implies that this upper bound is $O_P(1)$.  The result for $\hat{Q}_+^{(-k)}$ then follows by Markov's inequality.  The result for $\hat{\rho}_+^{(-k)}$ is analogous.
    
    \item By the same argument from above, $|| \hat{Q}_+^{(-k)}(x, 1) - \bar{Q}_+(x, 1) ||_2 \leq || \hat{Q}_+^{(-k)} - \bar{Q}_+ ||_2 / \epsilon = o_P(1)$.  Then, the conclusion follows by Markov's inequality.
    
    \item The argument is the same as above.
    
    \item Condition \ref{condition:first_stage} requires that $\P( \inf_x \hat{e}^{(-k)}(x) \geq \epsilon) \rightarrow 1$.  Thus, with probability approaching one we have $|| 1/\hat{e}^{(-k)} - 1/e ||_{2,k} \leq || \hat{e}^{(-k)} - e ||_{2,k} / \epsilon^2$.  Markov's inequality implies this upper bound tends to zero.   
    
    \item Write $\hat{\rho}_+(x, z) = g( \hat{\mu}(x, z))$ where $g(u) = \max \{ u \Lambda, 1 - 1/\Lambda + u/\Lambda \}$.  Since $u \mapsto g(u)$ is $\Lambda$-Lipschitz and $\rho_+(x, z) = g(\mu(x, z))$, the conclusion follows.
    \item  The map $q \mapsto q + \tfrac{1}{1 - \tau}(y - q)_+$ is $\Lambda$-Lipschitz, so \Cref{lemma:transformed_outcome_identity} allows us to write:
    \begin{align*}
    &|\varrho_+(x, 1; \hat{Q}_+^{(-k)}) - \bar{\varrho}_+(x, 1)|\\
    &= \left| (1 - \Lambda^{-1}) \int [\hat{Q}_+^{(-k)}(x, 1) + \tfrac{1}{1 - \tau} \{ y - \hat{Q}_+^{(-k)}(x, 1) \}_+] - [ \bar{Q}_{+}(x, 1) + \tfrac{1}{1 - \tau} \{ y - \bar{Q}_+(x, 1) \}_+] \, \d F(y \mid x, 1) \right|\\
    &\leq (1 - \Lambda^{-1}) \int \Lambda | \hat{Q}_+^{(-k)}(x, 1) - \bar{Q}_+(x, 1)| \, \d F(y \mid x, 1)\\
    &= (\Lambda - 1) | \hat{Q}_+^{(-k)}(x, 1) - \bar{Q}_+(x, 1)|.
    \end{align*}
    Squaring both sides and integrating gives the claimed result.
    
    \item This follows from H\"older's inequality and our overlap assumptions on $\hat{e}^{(-k)}$.
    
    \item Let $\gamma > 0$ be arbitrary.  Continuity of measure implies there exists $\delta > 0$ so small that $\P( \mu(X, Z) \in [1 - \tau \pm \delta]) < \gamma$.  Since $|| \hat{\mu}^{(-k)} - \mu ||_2 = o_P(1)$, Markov's inequality implies $\P( | \hat{\mu}^{(-k)}(X, Z) - \mu(X, Z)| > \gamma/2 \mid \F_{-k}) = o_P(1)$.  Thus, we have:
    \begin{align*}
    || \hat{Q}_+^{(-k)} - \hat{Q}_+ ||_2^2 &= \E[ | \mathbb{I} \{ \hat{\mu}^{(-k)}(X, Z) > 1 - \tau \} - \mathbb{I} \{ \mu(X, Z) > 1 - \tau \}| \mid \F_{-k}]\\
    &\leq \E[ \mathbb{I} \{ \mu(X, Z) \in [1 - \tau \pm \gamma] \} \mid \F_{-k}] + \E[ \mathbb{I} \{ | \hat{\mu}^{(-k)}(X, Z) - \mu(X, Z)| > \gamma/2 \} \mid \F_{-k}]\\
    &\leq \gamma + o_P(1)
    \end{align*}
    Since $\gamma$ is arbitrary, this proves $|| \hat{Q}_+^{(-k)} - \hat{Q}_+ ||_2 \xrightarrow{P} 0$.
\end{enumerate}
\end{proof}

\begin{lemma}[Convergence of the estimated influence function]
\label{lemma:influence_function_convergence}
Assume Conditions \ref{condition:primitives} and \ref{condition:first_stage}, and also that $|| \hat{e}^{(-k)} - e ||_2, || \hat{\rho}_{\pm}^{(-k)} - \varrho_{\pm}(\cdot, \cdot; \hat{Q}_{\pm}^{(-k)}) ||_2, || \hat{Q}_{\pm}^{(-k)} - \bar{Q}_{\pm} ||_2$ converge to zero in probability.  Then $|| \bar{Q}_{\pm} ||_2 < \infty$ and the following convergence holds:
\begin{align*}
|| \phi_{\ATE}^+(X, Y, Z; \hat{\eta}^{(-k)}) - \phi_{\ATE}^+(X, Y, Z; \bar{\eta} ) ||_{2,k} \xrightarrow{P} 0.
\end{align*}
Here, $\bar{\eta} = (e, \bar{Q}_+, \bar{Q}_-, \bar{\varrho}_+, \bar{\varrho}_-)$ and $\bar{\varrho}_{\pm} = \varrho_{\pm}(\cdot, \cdot; \bar{Q}_{\pm})$.
\end{lemma}

\begin{proof}
By the triangle inequality, $|| \bar{Q}_{\pm} ||_2 \leq || \hat{Q}_{\pm}^{(-k)} - \bar{Q}_{\pm} ||_2 + || \hat{Q}_{\pm}^{(-k)} ||_2 \leq o_P(1) + O_P(1)$.  Since the left-hand side is deterministic and does not change with $n$, it must be finite.  In addition,  \Cref{lemma:nuisance_facts}.\ref{item:kappa_lipschitz} shows that $|| \hat{\rho}_+^{(-k)} - \varrho_+(\cdot, \cdot; \hat{Q}_+^{(-k)}) ||_2 = o_P(1)$ under the given assumptions.  Therefore, we may write:
\begin{align*}
&|| \phi_1^+(X, Y, Z;\hat{\eta}^{(-k)}) - \phi_1^+(X, Y, Z; \bar{\eta}) ||_{2,k}\\
&\leq \underbrace{|| (1 - Z) \hat{\rho}_+^{(-k)}(x, 1) - (1 - z) \bar{\varrho}_+(x, 1) ||_{2,k}}_{\precsim_P || \hat{\rho}_+^{(-k)}(x, 1) - \bar{\varrho}_+(x, 1) ||_2 = o_P(1) \text{ by Markov + \Cref{lemma:nuisance_facts}.\ref{item:kappa1_convergence}}}\\
&+ \underbrace{\left| \left| \frac{1 - \hat{e}^{(-k)}(x)}{\hat{e}^{(-k)}(x)} z (1 - \Lambda^{-1}) ( \hat{Q}_+^{(-k)}(x, 1) + \tfrac{1}{1 - \tau} \{ y - \hat{Q}_+^{(-k)}(x, 1) \}_+ - \bar{Q}_+(x, 1) - \tfrac{1}{1 - \tau} \{ y - \bar{Q}_+(x, 1) \}_+ \right| \right|_{2,k}}_{\precsim_P (\Lambda/\epsilon) || \hat{Q}_+^{(-k)}(x, 1) - \bar{Q}_+(x, 1) ||_2 = o_P(1) \text{ since $q \mapsto q + \tfrac{1}{1 - \tau} \{ y - q \}_+$ is $\Lambda$-Lipschitz}}\\
& + \underbrace{\left| \left| \frac{1 - \hat{e}^{(-k)}(x)}{\hat{e}^{(-k)}(x)} z [ \hat{\rho}_+^{(-k)}(x, 1) - \bar{\varrho}_+(x, 1)] \right| \right|_{2,k}}_{\precsim_P || \hat{\rho}_+^{(-k)}(x, 1) - \varrho_+(\cdot, \cdot; \hat{Q}_+^{(-k)}) ||_2/\epsilon = o_P(1) \text{ by Markov + \Cref{lemma:nuisance_facts}.\ref{item:kappa1_convergence}}}\\
&+ \underbrace{\left| \left| \left( \frac{1 - \hat{e}^{(-k)}(x)}{\hat{e}^{(-k)}(x)} - \frac{1 - e(x)}{e(x)} \right) z [ \Lambda^{1} y + (1 - \Lambda^{-1}) ( \bar{Q}_+(x, 1) + \tfrac{1}{1 - \tau} \{ y - \bar{Q}_+(x, 1) \}_+) - \bar{\varrho}_+(x, 1)] \right| \right|_{2,k}}_{\precsim_P || z [ \Lambda^{-1} y + (1 - \Lambda^{-1}) \bar{Q}_+(x, 1) + \frac{1}{1 - \tau} \{ y - \bar{Q}_+(x, 1) \}_+ - \bar{\varrho}_+(x, 1)] (1/\hat{e}^{(-k)}(x) - 1/e(x)) ||_2 = o_P(1) \text{ by Markov + \Cref{lemma:nuisance_facts}.\ref{item:f_over_e}}}\\
&\xrightarrow{P} 0.
\end{align*}
Symmetric arguments show that $|| \phi_1^-(X, Y, Z; \hat{\eta}^{(-k)} - \phi_1^-(X, Y, Z; \bar{\eta}) ||_{2,k} \xrightarrow{P} 0$ as well.  Combining these two results with the triangle inequality gives the conclusion.
\end{proof}

\subsection{Proof of Proposition \ref{prop:ipw}}

\begin{proof}
This is a restatement of \citet[Proposition 2]{dornGuo2021sharp}.
\end{proof}

\subsection{Proof of Proposition \ref{prop:ipw_valid}}

\begin{proof}
By \Cref{eq:robust_ipw_formula}, it suffices to prove the following inequality:
\begin{align*}
\E \left[ \frac{\hat{Q}_{\tau}(X, 1)}{e(X)} + \frac{Z\{ Y - \hat{Q}_{\tau}(X, 1) \}}{e_+(X, Y)} \right] \leq \E \left[ \frac{\hat{Q}_{\tau}(X, 1)}{e(X)} + \{ Y - \hat{Q}_{\tau}(X, 1) \} Z \left( 1 + \Lambda^{\sign \{ Y - \hat{Q}_{\tau}(X, 1) \}} \frac{1 - e(X)}{e(X)} \right) \right].
\end{align*}
In fact, the inequality holds even without taking the expectation on both sides.  By  \citet[Proposition 2]{dornGuo2021sharp}, the function $e_+$ may be chosen to satisfy:
\begin{align*}
1 + \Lambda^{-1} \frac{1 - e(x)}{e(x)} \leq \frac{1}{e_+(x, y)} \leq 1 + \Lambda \frac{1 - e(x)}{e(x)}
\end{align*}
for every $(x, y)$.  Thus, whenever $y \geq \hat{Q}_{\tau}(x, 1)$, we have:
\begin{align*}
\frac{z \{ y - \hat{Q}_{\tau}(x, 1) \}}{e_+(x, y)} &\leq z \{ y - \hat{Q}_{\tau}(x, 1) \} \left( 1 + \Lambda \frac{1 - e(x)}{e(x)} \right)\\
&= z \{ y - \hat{Q}_{\tau}(x, 1) \} \left( 1 + \Lambda^{\sign \{ y - \hat{Q}_{\tau}(x, 1) \}} \frac{1 - e(x)}{e(x)} \right).
\end{align*}
On the other hand, whenever $y < \hat{Q}_{\tau}(x, 1)$, we instead have:
\begin{align*}
\frac{z \{ y - \hat{Q}_{\tau}(x, 1) \}}{e_+(x, y)} &\leq z \{ y - \hat{Q}_{\tau}(x, 1) \} \left( 1 + \Lambda^{-1} \frac{1 - e(x)}{e(x)} \right)\\
&= z \{ y - \hat{Q}_{\tau}(x, 1) \} \left( 1 + \Lambda^{\sign \{ y - \hat{Q}_{\tau}(x, 1) \}} \frac{1 - e(x)}{e(x)} \right).
\end{align*}
Thus, $z \{ y - \hat{Q}_{\tau}(x, 1) \}/e_+(x, y) \leq z \{ y - \hat{Q}_{\tau}(x, 1) \} (1 + \Lambda^{\sign \{ y - \hat{Q}_{\tau}(x, 1) \}} \frac{1 - e(x)}{e(x)})$ holds for every $(x, y)$ and the conclusion follows by taking expectations on both sides.  Evidently, all of these inequalities are equalities when $\hat{Q}_{\tau} = Q_{\tau}$.
\end{proof}

\subsection{Proof of Proposition \ref{prop:dm}}

\begin{proof}
The derivation in terms of robust optimization was already given in the main text.  For completeness, we present here an alternative proof based on \Cref{prop:ipw_valid}. \Cref{lemma:transformed_outcome_identity} implies $\E[ Q_{\tau}(X, 1) + \{ Y - Q_{\tau}(X, 1) \} \Lambda^{\sign \{ Y - Q_{\tau}(X, 1) \}} \mid X, Z = 1] = \rho_+(X, 1)$, which allows us to write:
\begin{align*}
\psi_1^+ &= \E \left[ \frac{Q_{\tau}(X, 1) Z}{e(X)} + Z \{ Y - Q_{\tau}(X, 1) \} \left( 1 + \Lambda^{\sign \{ Y - Q_{\tau}(X, 1) \}} \frac{1 - e(X)}{e(X)} \right) \right]\\
&= \E[ZY] + \E[ (1 - Z) Q_{\tau}(X, 1)] + \E \left[ \frac{1 - e(X)}{e(X)} Z \{ Y - Q_{\tau}(X, 1) \} \Lambda^{\sign \{ Y - Q_{\tau}(X, 1) \}} \right]\\
&= \E[ ZY] + \E[(1 - Z) Q_{\tau}(X, 1)] + \E [ \{ 1 - e(X) \} \E[ \{ Y - Q_{\tau}(X, 1) \} \Lambda^{\sign \{ Y - Q_{\tau}(X, 1) \}} \mid X, Z = 1]]\\
&= \E[ZY] + \E[ \{ 1 - e(X) \} \E[ Q_{\tau}(X, 1) + \{ Y - Q_{\tau}(X, 1) \} \Lambda^{\sign \{ Y - Q_{\tau}(X, 1) \}} \mid X, Z = 1]]\\
&= \E[ZY] + \E[ \{ 1 - e(X) \} \rho_+(X, 1)]\\
&= \E[ZY + (1 - Z) \rho_+(X, 1)].
\end{align*}
\end{proof}

\subsection{Proof of Proposition \ref{prop:dm_valid}}

\begin{proof}
\cite[Theorem 1]{rockafellar2000} shows that $Q_{\tau}(x, 1)$ is a minimizer of the map $q \mapsto \int q + \tfrac{1}{1 - \tau} \{ y - q \}_+ \, \d F(y \mid x, 1)$.  Thus, we have:
\begin{align*}
\varrho_+(x, 1; \hat{Q}_+) &= \Lambda^{-1} \mu(x, 1) + (1 - \Lambda^{-1}) \int \hat{Q}_{\tau}(x, 1) + \tfrac{1}{1 - \tau} \{ y - \hat{Q}_{\tau}(x, 1) \}_+ \, \d F(y \mid x, 1)\\
&\geq \Lambda^{-1} \mu(x, 1) + (1 - \Lambda^{-1}) \int \hat{Q}_{\tau}(x, 1) + \tfrac{1}{1 - \tau} \{ y - \hat{Q}_{\tau}(x, 1) \}_+ \, \d F(y \mid x, 1)\\
&= \rho_+(x, 1).
\end{align*}
Hence, $\E[ ZY + (1 - Z) \varrho_+(X, 1; \hat{Q}_{\tau})] \geq \E[ZY + (1 - Z) \rho_+(X, 1)] = \psi_1^+$.
\end{proof}

\subsection{Proof of Theorem \ref{theorem:double_valid_double_sharp}}

For simplicity, we prove the result only for the estimator $\hat{\psi}_{1,k}^+ = \E_{n,k}[ \phi_1^+(X, Y, Z; \hat{\eta}^{(-k)})]$.  Combining over folds gives the result for the fully aggregated DVDS estimator.  The result for other estimands follows by symmetric arguments.  Since we consider only one fold, we drop the $(\cdot)^{(-k)}$ superscript on nuisance estimators.

The proof proceeds in two cases, depending on whether $|| \hat{e} - e ||_2 = o_P(1)$ or $|| \hat{\rho}_+ - \varrho_+(\cdot, \cdot; \hat{Q}_+) ||_2 = o_P(1)$.  In both cases, we freely use results from \Cref{lemma:nuisance_facts}.
\bigskip\\
\noindent \textbf{First case}. Suppose first that $|| \hat{e}  - e ||_2 = o_P(1)$.  Then, using \Cref{eq:firstDecomposition} in \Cref{sec:estimator} gives:
\begin{align*}
\hat{\psi}_{1,k} &= \E_{n,k} \left[ \left( 1 - \frac{Z}{\hat{e}(X)} \right) \hat{Q}_+(X, 1) \right]\\
&+ \E_{n,k} \left[ \frac{Z \hat{Q}_+(X, 1)}{\hat{e}(X)} + \{ Y - \hat{Q}_+(X, 1) \} Z \left( 1 + \frac{1 - \hat{e}(X)}{\hat{e}(X)} \Lambda^{\sign \{ Y - \hat{Q}_+(X, 1) \}} \right) \right]\\
&= \E_{n,k} \left[ \frac{Z \hat{Q}_+(X, 1)}{e(X)} + \{ Y - \hat{Q}_+(X, 1) \} Z \left( 1 + \frac{1 - e(X)}{e(X)} \Lambda^{\sign\{ Y - \hat{Q}_+(X, 1) \}} \right) \right] \\ 
&+ \underbrace{\E_{n,k} \left[ \left( 1 - \frac{Z}{e(X)} \right) \hat{Q}_+(X, 1) \right]}_{= o_P(1) \text{ by Chebyshev's conditional on $\F_{-k}$}} + \underbrace{\E_{n,k} \left[ \left( \frac{1}{e(X)} - \frac{1}{\hat{e}(X)} \right) Z \hat{Q}_+(X, 1) \right]}_{\precsim || 1/\hat{e} - 1/e ||_{2,k} \times || \hat{Q}_+(x, 1) ||_{2,k} = o_P(1) O_P(1)}\\
&+ \underbrace{\E_{n,k} \left[ \{ Y - \hat{Q}_+(X, 1) \} Z \Lambda^{\sign \{ Y - \hat{Q}_+(X, 1) \}} \left( \frac{1 - \hat{e}(X)}{\hat{e}(X)} - \frac{1 - e(X)}{e(X)} \right)  \right]}_{\precsim \Lambda || y - \hat{Q}_+(x, 1) ||_{2,k} \times || 1/\hat{e} - 1/e ||_{2,k} = O_P(1) o_P(1)}\\
&= \E_{n,k} \left[ \frac{Z \hat{Q}_+(X, 1)}{e(X)} + \{ Y - \hat{Q}_+(X, 1) \} Z \left( 1 + \frac{1 - e(X)}{e(X)} \Lambda^{\sign\{ Y - \hat{Q}_+(X, 1) \}} \right) \right] + o_P(1)\\
&= \E \left[ \frac{Z \hat{Q}_+(X, 1)}{e(X)} + \{ Y - \hat{Q}_+(X, 1) \} Z \left( 1 + \frac{1 - e(X)}{e(X)} \Lambda^{\sign \{ Y - \hat{Q}_+(X, 1) \}} \right) \, \bigg| \, \F_{-k} \right] + o_P(1)
\end{align*}
Here, the last line follows by Chebyshev's inequality conditional on $\F_{-k}$.  By \Cref{prop:ipw_valid}, the conditional expectation in the final line of the preceding display is always at least $\psi_1^+$, so we have shown $\hat{\psi}_{1,k} \geq \psi_1^+ - o_P(1)$.  This proves validity when $\hat{e}$ is consistent.

To prove sharpness when $\hat{Q}_+$ is also consistent, we further analyze the conditional expectation in the final line of the preceding display:
\begin{align*}
&\E \left[ \frac{Z \hat{Q}_+(X, 1)}{e(X)} + \{ Y - \hat{Q}_+(X, 1) \} Z \left( 1 + \frac{1 - e(X)}{e(X)} \Lambda^{\sign \{ Y - \hat{Q}_+(X, 1) \}} \right) \, \bigg| \, \F_{-k} \right]\\
&= \E[ZY] + \E \left[ (1 - Z) \E [ \hat{Q}_+(X, 1) + \Lambda^{\sign \{ Y - \hat{Q}_+(X, 1) \}} \{ Y - \hat{Q}_+(X, 1) \} \mid X, Z = 1, \F_{-k}] \mid \F_{-k} \right]\\
& = \E[ ZY] + \E[(1 - Z) \E[ \Lambda^{-1} Y + (1 - \Lambda^{-1}) [ \hat{Q}_+(X, 1) + \tfrac{1}{1 - \tau} \{ Y - \hat{Q}_+(X, 1) \}_+ \mid X, Z = 1, \F_{-k}] \mid \F_{-k}] &\text{(\Cref{lemma:transformed_outcome_identity})}\\
& = \E[ZY] + \E[(1 - Z) \varrho_+(X, 1; \hat{Q}_+)]\\
&= \E[ZY + (1 - Z) \rho_+(X, 1)] + \underbrace{\E[(1 - Z) \{ \varrho_+(X, 1; \hat{Q}_+) - \rho_+(X, 1) \} \mid \F_{-k} ]}_{\leq \Lambda || \hat{Q}_+(x, 1) - Q_+(x, 1) ||_2 = o_P(1) \text{ by \Cref{lemma:nuisance_facts}.\ref{item:kappa_lipschitz}}}\\
&= \psi_1^+ + o_P(1).
\end{align*}
Here, the last line follows from \Cref{prop:dm}.  Combining with the previous calculation shows $\hat{\psi}_{1,k}^+ = \psi_1^+ + o_P(1)$.
\bigskip\\
\noindent \textbf{Second case}.  Suppose next that $|| \hat{\rho}_+ - \varrho_+(\cdot, \cdot; \hat{Q}_+) ||_2 = o_P(1)$.  Then, \Cref{eq:secondDecomposition} in \Cref{sec:estimator} gives:
\begin{align*}
\hat{\psi}_{1,k}^+ &= \E_{n,k}[ZY + (1 - Z) \hat{\rho}_+(X, 1)]\\
&+ \E_{n,k} \left[ \frac{1 - \hat{e}(X)}{\hat{e}(X)} Z ( \Lambda^{-1} Y + (1 - \Lambda^{-1})[ \hat{Q}_+(X, 1) + \tfrac{1}{1 - \tau} \{ Y - \hat{Q}_+(X, 1) \}_+] - \hat{\rho}_+(X, 1)) \right]\\
&= \E[ ZY + (1 - Z) \varrho_+(X, 1; \hat{Q}_+) \mid \F_{-k}]\\
&+ \underbrace{\E_{n,k}[ZY + (1 - Z) \varrho_+(X, 1; \hat{Q}_+)] - \E[ZY + (1 - Z) \varrho_+(X, 1; \hat{Q}_+) \mid \F_{-k} ]}_{= o_P(1) \text{ by Chebyshev's inequality conditional on $\F_{-k}$}}\\
&+ \underbrace{\E_{n,k}[(1 - Z) ( \hat{\rho}_+(X, 1) - \varrho_+(X, 1; \hat{Q}_+) )]}_{\leq || \hat{\rho}_+(x, 1) - \rho_+(x, 1) ||_2 = o_P(1) \text{ by Markov conditional on $\F_{-k}$}}\\
&+ \underbrace{\E_{n,k} \left[ \frac{1 - \hat{e}(X)}{\hat{e}(X)} Z ( \Lambda^{-1} Y + (1 - \Lambda^{-1})  [ \hat{Q}_+(X, 1) + \tfrac{1}{1 - \tau} \{ Y - \hat{Q}_+(X, 1) \}_+] - \varrho_+(X, 1; \hat{Q}_+)) \right]}_{= o_P(1) \text{ by Chebyshev's inequality conditional on $\F_{-k}$ + def. of $\varrho$}}\\
&+ \underbrace{\E_{n,k} \left[ \frac{1 - \hat{e}(X)}{\hat{e}(X)} Z ( \varrho_+(X, 1; \hat{Q}_+) - \hat{\rho}_+(X, 1) ) \right]}_{\precsim_P || \varrho_+(\cdot, \cdot; \hat{Q}_+) - \hat{\rho}_+ ||_2 / \epsilon = o_P(1) \text{ by Markov}}\\
&= \E[ ZY + (1 - Z) \hat{\rho}_+(X, 1) \mid \F_{-k}] + o_P(1)
\end{align*}
\Cref{prop:dm_valid} implies $\E[ZY + (1 - Z) \hat{\rho}_+(X, 1) \mid \F_{-k}] \geq \psi_1^+$, so $\hat{\psi}_{1,k}^+ \geq \psi_1^+ - o_P(1)$ even if $\hat{Q}_+$ is misspecified.  

If we also have $|| \hat{Q}_+(x, 1) - Q_+(x, 1) ||_2 = o_P(1)$, then \Cref{lemma:nuisance_facts}.\ref{item:kappa_lipschitz} gives $|| \varrho_+(x, 1; \hat{Q}_+) - \rho_+(x, 1) ||_2 = o_P(1)$.  As a consequence, $\E[ZY + (1 - Z) \varrho_+(X, 1; \hat{Q}_+) \mid \F_{-k}] = \E[ ZY + (1 - Z) \rho_+(X, 1) \mid \F_{-k}] + o_P(1) = \psi_1^+ + o_P(1)$.

\subsection{Proof of Corollary \ref{corollary:binary_double_valid_double_sharp}}

\begin{proof}
Throughout this proof, we use the same notation from the proof of \Cref{theorem:double_valid_double_sharp}.

If $\hat{e}$ is consistent, then validity follows immediately from \Cref{theorem:double_valid_double_sharp}.

If $\hat{\mu}$ is consistent, then sharpness requires a bit more argument. We cannot recover $\hat{Q}_+$ consistency if $\mu(X, Z)$ has an atom at $1 - \tau$. Still, the proof of \Cref{theorem:double_valid_double_sharp}'s second case uses $\varrho_+(\cdot, \cdot; \hat{Q}_+)$'s consistency rather than $\hat{Q}_+$'s consistency. Its argument will go through if we have $|| \hat{\rho}_+ - \rho_+ ||_2 = o_P(1)$ and $|| \varrho_+(\cdot, \cdot; \hat{Q}_+) - \hat{\rho}_+ ||_2 = o_P(1)$. The first of these follows from \Cref{lemma:nuisance_facts}.\ref{item:mu_to_kappa}, so we only need to verify the second.

The result $|| \varrho_+(\cdot, \cdot; \hat{Q}_+) - \hat{\rho}_+ ||_2 = o_P(1)$  follows by writing out the explicit formulae for $\hat{\rho}_+(x, 1)$, $\varrho_+(x, 1; \hat{Q}_+)$ and their difference:
\begin{align*}
\varrho_+(x, 1; \hat{Q}_+) &= \int \Lambda^{-1} y + (1 - \Lambda^{-1})[\hat{Q}_+(x, 1) + \tfrac{1}{1 - \tau} \{ y - \hat{Q}_+(x, 1) \}_+] \, \d F(y \mid x, z)\\
&= \left\{ 
\begin{array}{ll}
\Lambda^{-1} \mu(x, 1) + (1 - \Lambda^{-1}) &\text{if } \hat{\mu}(x, 1) > 1 - \tau\\
\Lambda^{-1} \mu(x, 1) + (\Lambda - \Lambda^{-1}) \mu(x, 1) &\text{if } \hat{\mu}(x, 1) \leq 1 - \tau
\end{array}
\right.\\
\hat{\rho}_+(x, 1) &= \left\{ 
\begin{array}{ll}
\Lambda^{-1} \hat{\mu}(x, 1) + (1 - \Lambda^{-1}) &\text{if } \hat{\mu}(x, 1) > 1 - \tau\\
\Lambda^{-1} \hat{\mu}(x, 1) + (\Lambda - \Lambda^{-1}) \hat{\mu}(x, 1) &\text{if } \hat{\mu}(x, 1) \leq 1 - \tau
\end{array}
\right.\\
\implies \hat{\rho}_+(x, 1) - \varrho_+(x, 1; \hat{Q}_+) &= \left\{ 
\begin{array}{ll}
\Lambda^{-1} [ \hat{\mu}(x, 1) - \mu(x, 1)] &\text{if } \hat{\mu}(x, 1) > 1 - \tau\\
\Lambda [ \hat{\mu}(x, 1) - \mu(x, 1)] &\text{if } \hat{\mu}(x, 1) \leq 1 - \tau
\end{array}
\right.
\end{align*}
Thus, $| \hat{\rho}_+(x, 1) - \varrho_+(x, 1; \hat{Q}_+)|^2 \leq \Lambda^2 | \hat{\mu}(x, 1) - \mu(x, 1)|^2$, meaning $|| \hat{\rho}_+(x, 1) - \varrho_+(x, 1; \hat{Q}_+) ||_2 \leq \Lambda || \hat{\mu}(x, 1) - \mu(x, 1) ||_2 = o_P(1)$. 
\end{proof}

\subsection{Proof of Theorem \ref{theorem:efficiency}}

Since the sum or difference of two efficient and asymptotically linear estimators is again efficient and asymptotically linear, we only need to prove the result for $\hat{\psi}_1^+$.  Results for the other bounds follow by symmetric arguments.

We break the proof into four parts:  asymptotic linearity for continuous outcomes (\ref{sec:asymptotic_linearity_continuous}), asymptotic linearity for binary outcomes (\ref{sec:asymptotic_linearity_binary}), efficiency for continuous outcomes (\ref{sec:efficiency_continuous}) and efficiency for binary outcomes (\ref{sec:efficiency_binary}).

\subsubsection{Asymptotic linearity for continuous outcomes} \label{sec:asymptotic_linearity_continuous}

\begin{proof}
For the first part of this proof, we fix a single fold $k \in \{ 1, \ldots, K \}$ and drop the $(\cdot)^{(-k)}$ superscripts on nuisance estimates.  Write the difference between the $\E_{n,k}[ \phi_1^+(X, Y, Z; \hat{\eta})]$ and the ``oracle" DVDS estimator $\E_{n,k}[ \phi_1^+(X, Y, Z; \eta)]$ as follows:
\begin{align}
&\E_{n,k}[ \phi_1^+(X, Y, Z;\hat{\eta})] - \E_{n,k}[ \phi_1^+(X, Y, Z;\eta)] \nonumber \\
&= \E_{n,k} \left[ \frac{1 - \hat{e}(X)}{\hat{e}(X)} Z (1 - \Lambda^{-1}) \left\{ \hat{Q}_+(X, 1) + \frac{\{ Y - \hat{Q}_+(X, 1) \}_+}{1 - \tau} - Q_+(X, 1) -  \frac{\{ Y - Q_+(X, 1) \}_+}{1 - \tau} \right\} \right] \label{eq:change_q}\\
&+ \E_{n,k} \left[ \left( \frac{1 - \hat{e}(X)}{\hat{e}(X)} - \frac{1 - e(X)}{e(X)} \right) Z \left\{ \Lambda^{-1} Y + (1 - \Lambda^{-1}) \left[ Q_+(X, 1) + \frac{\{ Y - Q_+(X, 1) \}_+}{1 - \tau} \right] - \rho_+(X, 1) \right\} \right] \label{eq:change_e}\\
& + \E_{n,k} \left[ \left( \frac{1 - \hat{e}(X)}{\hat{e}(X)} - \frac{1 - e(X)}{e(X)} \right) Z \{ \rho_+(X, 1)  - \hat{\rho}_+(X, 1) \} \right] \label{eq:error_product} \\
&+ \E_{n,k} \left[ \left( \frac{1 - e(X)}{e(X)} Z - [1 - Z] \right) \{ \rho_+(X, 1) - \hat{\rho}_+(X, 1) \} \right] \label{eq:change_kappa} 
\end{align}

Each of the terms in this decomposition will be shown to be $o_P(n^{-1/2})$.  For  (\ref{eq:change_q}), (\ref{eq:change_e}) and (\ref{eq:change_kappa}), we start by applying Chebyshev's inequality conditional on $\F_{-k}$:
\begin{align*}
|(\ref{eq:change_q})| &\precsim_P \left| \E \left[ \frac{1 - \hat{e}(X)}{\hat{e}(X)} Z \left\{ \hat{Q}_+(X, 1) + \frac{\{ Y - \hat{Q}_+(X, 1) \}_+}{1 - \tau} - Q_+(X, 1) - \frac{\{ Y - Q_+(X, 1) \}_+}{1 - \tau} \right\} \, \bigg| \, \F_{-k} \right] \right|\\
&+ \frac{1}{\sqrt{| \F_k|}} \underbrace{\E \left[ \left( \frac{1 - \hat{e}(X)}{\hat{e}(X)} Z  \left\{ \hat{Q}_+(X, 1) + \frac{\{ Y - \hat{Q}_+(X, 1) \}_+}{1 - \tau} - Q_+(X, 1) - \frac{\{ Y - Q_+(X, 1) \}_+}{1 - \tau} \right\} \right)^2 \, \bigg| \, \F_{-k} \right]^{1/2}}_{\leq \Lambda || \hat{Q}_+(x, 1) - Q_+(x, 1) ||_2 / \epsilon \text{ since $q \mapsto q + \tfrac{1}{1 - \tau} \{ y - q \}_+$ is $\Lambda$-Lipschitz}}\\
&\leq  \int_{\R^d} \frac{1}{\epsilon} \underbrace{\left| \int_{\R} \hat{Q}_+(x, 1) + \frac{\{ y - \hat{Q}_+(x, 1) \}_+}{1 - \tau} - Q_+(x, 1) - \frac{\{ y - Q_+(x, 1) \}_+}{1 - \tau} \, f(y \mid x, 1) \, \d y \right|}_{\leq \Lambda M | \hat{Q}_+(x, 1) - Q_+(x, 1)|^2 \text{ by \Cref{lemma:cvar_neyman_orthogonal}}} \, \d P_X(x) + o_P(1/\sqrt{|\F_k|})\\
&\leq \int_{\R^d} \frac{\Lambda M}{\epsilon} | \hat{Q}_+(x, 1) - Q_+(x, 1)|^2 \, \d P_X(x) + o_P(1/\sqrt{|\F_k|})\\
&\precsim || Q_+ - \hat{Q}_+ ||_2^2 + o_P(1/\sqrt{|\F_k|})\\
&= o_P(n^{-1/2})\\
|(\ref{eq:change_e})| &\precsim_P \bigg| \E \bigg[ \left( \frac{e(X) - \hat{e}(X)}{\hat{e}(X)} \right) \underbrace{\E \left[ \Lambda^{-1} Y + (1 - \Lambda^{-1}) \left[ Q_+(X, 1) + \frac{\{ Y - Q_+(X, 1) \}_+}{1 - \tau} \right] - \rho_+(X, 1) \, \bigg| \, X, Z = 1 \right]}_{= 0 \text{ by def. of $\rho_+$}} \, \bigg| \, \F_{-k} \bigg] \bigg|\\
&+ \frac{1}{\sqrt{| \F_k|}} \underbrace{\E \left[ \left( \left[ \frac{1}{\hat{e}(X)} - \frac{1}{e(X)} \right] \left\{ \Lambda^{-1} Y + (1 - \Lambda^{-1}) \left[ Q_+(X, 1) + \frac{\{ Y - Q_+(X, 1) \}_+}{1 - \tau}  \right] - \rho_+(X, 1) \right\} \right)^2 \, \bigg| \, \F_{-k} \right]^{1/2}}_{= o_P(1) \text{ by \Cref{lemma:nuisance_facts}.\ref{item:f_over_e}}}\\
&= o_P(n^{-1/2})\\
|(\ref{eq:change_kappa})| &\precsim_P \bigg| \E \bigg[ \{ \rho_+(X, 1) - \hat{\rho}_+(X, 1) \} \underbrace{\E \left[ \frac{1 - e(X)}{e(X)} Z - [1 - Z] \, \bigg| \, X \right]}_{= 0} \, \bigg| \, \F_{-k} \bigg] \bigg| \\
&+ \frac{1}{\sqrt{| \F_k|}} \underbrace{\E \left[ \{ \rho_+(X, 1) - \hat{\rho}_+(X, 1) \}^2 \left( \frac{1 - e(X)}{e(X)} Z - [1 - Z] \right)^2 \, \bigg| \, \right]^{1/2}}_{\precsim || \rho_+(x, 1) - \hat{\rho}_+(x, 1) ||_2^2 / \epsilon^2 = o_P(1)}\\
&= o_P(n^{-1/2}).
\end{align*}

Meanwhile, for the term (\ref{eq:error_product}), we simply use the Cauchy-Schwarz inequality:  $|(\ref{eq:error_product})| \leq || 1/\hat{e} - 1/e ||_{2,k} \times || \hat{\rho}_+(x, 1) - \rho_+(x, 1) ||_{2,k} \precsim_P || 1/\hat{e} - 1 / e||_2 \times || \hat{\rho}_+ - \rho_+ ||_2 = o_P(n^{-1/2})$.  Therefore, we have shown that $\E_{n,k}[ \phi_1^+(X, Y, Z; \hat{\eta})] = \E_{n,k}[\phi_1^+(X, Y, Z; \eta)] + o_P(n^{-1/2})$.  

Finally, combining over folds gives our desired asymptotic linearity:
\begin{align*}
\sqrt{n}( \hat{\psi}_1^+ - \psi_1^+) &= \frac{1}{\sqrt{n}} \sum_{i = 1}^n \{ \phi_1^+(X_i, Y_i, Z_i; \eta) - \psi_1^+ \} + o_P(n^{-1/2}).
\end{align*}
\end{proof}

\subsubsection{Asymptotic linearity for binary outcomes} \label{sec:asymptotic_linearity_binary}

\begin{proof}
We use the same decomposition as in the continuous case.  The arguments showing that terms (\ref{eq:change_e}), (\ref{eq:change_kappa}) are $o_P(n^{-1/2})$ do not rely on continuity or rate conditions.  Therefore, they continue to work in the binary case since \Cref{lemma:nuisance_facts}.\ref{item:mu_to_Q} shows that $\hat{Q}_+$ is consistent under the binary assumptions.  Meanwhile, \Cref{lemma:nuisance_facts}.\ref{item:mu_to_kappa} allows us to conclude $|| \hat{e} - e ||_2 \times || \hat{\rho}_+ - \rho_+ ||_2 \precsim || \hat{e} - e ||_2 \times || \hat{\mu} - \mu ||_2 = o_P(n^{-1/2})$, so $(\ref{eq:error_product}) = o_P(n^{-1/2})$ as well.  

The remaining term is (\ref{eq:change_q}), which cannot be directly handled using the arguments from the continuous case as \Cref{lemma:cvar_neyman_orthogonal} is no longer applicable.  Nevertheless, the continuous arguments still allow us to write:
\begin{align*}
| ( \ref{eq:change_q})| &\precsim_P \int \frac{1}{\epsilon} \left| \int \hat{Q}_+(x, 1) + \frac{\{ y - \hat{Q}_+(x, 1) \}_+}{1 - \tau} - Q_+(x, 1) - \frac{\{ y - Q_+(x, 1) \}_+}{1 - \tau} \, \d F(y \mid x, 1) \right| \, \d P_X(x) + o_P(n^{-1/2})
\end{align*}
The interior integral in the above display is zero if $\hat{\mu}(x, 1), \mu(x, 1)$ are both larger than $1 - \tau$ or both less than or equal to $1 - \tau$, since $\hat{Q}_+(x, 1) = Q_+(x, 1)$ on these events.  Meanwhile, if $\hat{\mu}(x, 1)$ and $\mu(x, 1)$ are on opposite sides of $1 - \tau$, the absolute value of the interior integral takes the value $|1 - \mu(x, 1)/(1 - \tau)|$.  Thus, the interior integral can be simplified as:
\begin{align*}
&\left| \int \hat{Q}_+(x, 1) + \frac{\{ y - \hat{Q}_+(x, 1) \}}{1 - \tau} - Q_+(x, 1) - \frac{\{ y - Q_+(x, 1) \}_+}{1 - \tau} \right|\\
&\quad = |1 - \mu(x, 1)/(1 - \tau)| \mathbb{I} \{ \mu(x, 1) \leq 1 - \tau < \hat{\mu}(x, 1) \text{ or } \hat{\mu}(x, 1) \leq 1 - \tau < \mu(x, 1) \}
\end{align*}
Since $|| \hat{\mu} - \mu ||_{\infty} = o_P(n^{-1/4})$, there exists a sequence $\delta_n = o(n^{-1/4})$ such that $\P( || \hat{\mu} - \mu ||_{\infty} \leq \delta_n) \rightarrow 1$.  On the event $|| \hat{\mu} - \mu ||_{\infty} \leq \delta_n$, the events $\mu(x, 1) \leq 1 - \tau < \hat{\mu}(x, 1)$ and $\hat{\mu}(x, 1) \leq 1 - \tau < \mu(x, 1)$ only occur if $\mu(x, 1)$ is within $\delta_n$ of $1 - \tau$.  Therefore, we may write:
\begin{align*}
(\ref{eq:change_q}) &\precsim_P \int_{\R^d} \frac{1}{\epsilon} | 1 - \mu(x, 1)/(1 - \tau)| \mathbb{I} \{ \mu(x, 1) \in [1 - \tau \pm \delta_n] \} \, \d P_X(x) + o_P(n^{-1/2})\\
&= \frac{1}{\epsilon} \delta_n \Pobs( \mu(X, 1) \in [1 - \tau \pm \delta_n]) + o_P(n^{-1/2})\\
&\precsim \delta_n^2/\epsilon + o_P(n^{-1/2}) &\text{Bounded density}\\
&= o_P(n^{-1/2}).
\end{align*}
In summary, $\E_{n,k}[ \phi_1^+(X, Y, Z; \hat{\eta}^{(-k)})] = \E_{n,k}[ \phi_1^+(X, Y, Z; \eta) + o_P(n^{-1/2})$ also in the binary case.  From here, the rest of the proof goes through as in the continuous case.
\end{proof}

\subsubsection{Efficiency for continuous outcomes} \label{sec:efficiency_continuous}

\begin{proof}
We start by defining the model.  Let $\nu$ be some $\sigma$-finite measure on $\R^d$ which dominates $P_X$ (e.g., $P_X$ itself), and let $\mathcal{P}$ be the set of distributions on $\mathcal{X} \times \R \times \{ 0, 1 \}$ with density of the form (\ref{eq:model_density}) with respect to $\nu \times (\textup{Lebesgue measure}) \times \textup{(Counting measure)}$. 
\begin{align}
f_{X, Y, Z}(x, y, z) = f_X(x) e(x)^z[1 - e(x)]^{1 - z} f_{Y | X, Z}(y | x, z) \label{eq:model_density}
\end{align}
Condition \ref{condition:density} (continuous case) implies that the observed-data distribution $P$ belongs to the model $\mathcal{P}$.

Now, we can begin computing the efficient influence function.  We write the functional $P \mapsto \psi_1^+(P)$ as a linear combination of three simpler functionals:
\begin{align*}
\psi_1^+(P) &= \E_P[ZY + (1 - Z) \rho_+(X, 1)]\\
&= \underbrace{\E_P[ZY]}_{\psi_a(P)} + \Lambda^{-1} \underbrace{\E_P[(1 - Z) \mu(X, 1)]}_{\psi_b(P)} + (1 - \Lambda^{-1}) \underbrace{\E_P[(1 - Z) \textup{CVaR}_{\tau}(X, 1)]}_{\psi_c(P)}.
\end{align*}
The functionals $\psi_a$ and $\psi_b$ are well-understood.  The results in \citet[Example 25.24]{asymptotic_statistics}, \citet[Theorem 1]{hahn1998}, and the elementary calculus of influence functions shows that they have the following efficient influence functions:
\begin{align*}
\phi_b(x, y, z) &= zy - \E_P[ZY]\\
\phi_b(x, y, z) &= (1 - z) \mu(x, 1) + \frac{1 - e(x)}{e(x)} z \{ y - \mu(x, 1) \} - \E[(1 - Z) \mu(X, 1)]
\end{align*}
However, $\psi_c$ may be less standard.  We claim its efficient influence function is given by the following expression:
\begin{align*}
\phi_c(x, y, z) &= (1 - z) \textup{CVaR}_{\tau}(x, 1) + \frac{1 - e(x)}{e(x)} z [ Q_{\tau}(x, 1) + \tfrac{1}{1 - \tau} \{ y - Q_{\tau}(x, 1) \}_+ - \textup{CVaR}_{\tau}(x, 1 ] - \psi_c(P)
\end{align*}
To verify this, we use the general strategy outlined in \cite{bkrw1998, newey1990, asymptotic_statistics}.  Let $\{ \d P_t = f_X(x; t) e(x;t)^z \{ 1 - e(x;t) \}^{1 - z} f_{Y|X,Z}(y | x, z;t) \}$ be a smoothly-parameterized submodel of $\mathcal{P}$ which passes through $P$ at $t = 0$.  Let $\ell_{X,Y,Z}(x, y, z)$ be the log likelihood function at $t = 0$.  In what follows, for any function $g(w;t)$, we denote by $\dot{g}$ the function $(x, y, z) \mapsto \tfrac{\d}{\d t} g(x, y, z; t) |_{t = 0}$.  In this notation, the score function of the submodel at $t = 0$ is:
\begin{align}
\dot{\ell}_{X,Y,Z}(x, y, z) = \dot{\ell}_X(x) + \frac{\dot{e}(x) \{ z - e(x) \}}{e(x) \{ 1 - e(x) \}}+ \dot{\ell}_{Y|X,Z}(y | x, z). \label{score_function}
\end{align}
The tangent space of the model $\mathcal{P}$ at $P$ consists of functions of the form (\ref{score_function}) as $\dot{\ell}_X, \dot{e}$ and $\dot{\ell}_{Y | X, Z}$ range over square-integrable functions satisfying $\E_P[\dot{\ell}_X(X)] = \E_P[\dot{\ell}_{Y | X, Z}(Y | X, Z) | X, Z] = 0$.  

Now, we compute the pathwise derivative of the map $t \mapsto \psi_c(P_t)$:
\begin{align*}
\frac{\d}{\d t} \psi_c(P_t) \, \bigg|_{t = 0} &= \frac{\d}{\d t} \E_t[ (1 - Z) \textup{CVaR}_{\tau}(X, 1)\, \bigg|_{t = 0}\\
&= \frac{\d}{\d t} \int \{ 1 - e(x;t) \} \textup{CVaR}_{\tau}(x, 1 ; t) f_X(x; t) \d \nu(x) \, \bigg|_{t = 0} \\
&= \int -\dot{e}(x) \textup{CVaR}_{\tau}(x, 1) f_X(x) \, \d \nu(x) \bigg|_{t = 0} + \int \{ 1 - e(x) \} \textup{CVaR}_{\tau}(x, 1) \dot{\ell}_X(x) f_X(x) \, \d \nu(x) \bigg|_{t = 0}\\
&+ \int \{ 1 - e(x) \} \dot{\textup{CVaR}}_{1 - \tau}(x, 1) f_X(x) \, \d \nu(x) \, \bigg|_{t = 0}\\
&= \E \left[ -\textup{CVaR}_{\tau}(X, 1) \{ Z - e(X) \} \dot{\ell}_{X,Y,Z}(X,Y,Z) \right]\\
&+ \E \left[ \left( \{1 - e(X) \} \textup{CVaR}_{\tau}(X, 1) - \psi_c(P) \right) \dot{\ell}_{X,Y,Z}(X,Y,Z) \right] \\
&+ \E[\{ 1 - e(X) \} \dot{\textup{CVaR}}_{1 - \tau}(X, 1) ]
\end{align*}
To evaluate the quantity $\E[ \{ 1 - e(X) \} \dot{\textup{CVaR}}_{\tau}(X, 1)]$ more explicitly, we compute $\dot{\textup{CVaR}}_{\tau}(x, 1)$ using the Leibniz rule and implicit differentiation:
\begin{align*}
\dot{\textup{CVaR}}_{\tau}(x, 1) &= \frac{\d}{\d t} \int Q_{\tau}(x, 1; t) + \tfrac{1}{1 - \tau} \{ y - Q_{\tau}(x, 1; t) \}_+ f_{Y|X,Z}(y \mid x, 1; t) \d y \bigg|_{t = 0}\\
&= \frac{\d}{\d t} \frac{1}{1 - \tau} \int_{Q_{\tau}(x, 1; t)}^{\infty} y f_{Y|X,Z}(y | x, 1;t) \d y \, \bigg|_{t = 0}\\
&= -\frac{1}{1 - \tau} \left\{ Q_{\tau}(x, 1) f_{Y|X,Z}(Q_{\tau}(x,1) \mid x, 1) \dot{Q}_{\tau}(x, 1) + \int_{Q_{\tau}(x, 1)}^{\infty} y f_{Y|X,Z}(y | x, 1) \dot{\ell}_{Y|X,Z}(y|x,1)\, \d y \right\}\\
&= \frac{1}{1 - \tau} \int_{\R} \{ y - Q_{\tau}(x, 1) \} \mathbb{I} \{ y > Q_{\tau}(x, 1) \} \dot{\ell}_{Y|X,Z}(y|x,1) f_{Y|X,Z}(y|x,1) \, \d y\\
&= \E[ \tfrac{1}{1 - \tau} \{ Y - Q_{\tau}(X, 1) \}_+ \dot{\ell}_{Y|X,Z}(Y|X,1) \mid X = x, Z = 1]
\end{align*}
Thus, we have:
\begin{align*}
\E[ \{ 1 - e(X) \} \dot{\textup{CVaR}}_{\tau}(X, 1)] &= \E [ \{ 1 - e(X) \} \E[ \tfrac{1}{1 - \tau} \{ Y - Q_{\tau}(X, 1) \}_+ \dot{\ell}_{Y | X, Z}(Y | X, 1) \mid X, Z = 1] ]\\
&= \E \left[ \frac{1 - e(X)}{e(X)} Z \tfrac{1}{1 - \tau} \{ Y - Q_{\tau}(X, 1) \}_+ \dot{\ell}_{Y|X,Z}(Y|X,Z) \right]\\
&= \E \left[ \frac{1 - e(X)}{e(X)} Z[ Q_{\tau}(X, 1) + \tfrac{1}{1 - \tau} \{ Y - Q_{\tau}(X, 1) \}_+ - \textup{CVaR}_{\tau}(X, 1) ] \dot{\ell}_{Y|X,Z}(Y | X, Z) \right]\\
&= \E \left[ \frac{1 - e(X)}{e(X)} Z [ Q_{\tau}(X, 1) + \tfrac{1}{1 - \tau} \{ Y - Q_{\tau}(X, 1) \}_+ - \textup{CVaR}_{\tau}(X, 1)] \dot{\ell}_{X,Y,Z}(X, Y, Z) \right].
\end{align*}
Plugging this into the above expression for $\frac{\d}{\d t} \psi_c(P_t) |_{t = 0}$ and simplifying gives:
\begin{align*}
\frac{\d}{\d t} \psi_c(P_t) &= \E[ \dot{\ell}_{X,Y,Z}(X, Y, Z) \phi_c(X, Y, Z)].
\end{align*}
This proves that the claimed function $\phi_c$ is an influence function for $\psi_c$.  Since the model $\mathcal{P}$ is locally saturated, $\phi_c$ must also be the efficient influence function.

By the algebra of influence functions, the efficient influence function for $\psi_1^+$ must be $\phi_a + \Lambda^{-1} \phi_b + (1 - \Lambda^{-1}) \phi_c$.  Some straightforward algebra shows that this is equal to $\phi_1^+(x, y, z; \eta) - \psi_1^+(P)$.  Thus, the influence function of the DVDS estimator $\hat{\psi}_1^+$ is the efficient one.
\end{proof}

\subsubsection{Efficiency for binary outcomes} \label{sec:efficiency_binary}

\begin{proof}
For binary outcomes, the model $\mathcal{P}$ is modified so that the $Y$-component of the likelihood is conditionally Bernoulli.  This means the new model consists of distributions with density of the form (\ref{eq:binary_model_density}) with respect to $\nu \times (\textup{Counting measure}) \times (\textup{Counting measure})$
\begin{align}
f_{X,Y,Z}(x, y, z) = f_X(x) \times e(x)^z \{ 1 - e(x) \}^{1 - z} \times \mu(x, z)^y \{ 1 - \mu(x, z) \}^{1 - y} \label{eq:binary_model_density}
\end{align}
We consider paths of the form $t \mapsto f_X(x; t) \times e(x;t)^z \{ 1 - e(x;t) \}^{1 - z} \times \mu(x, z; t)^y \{ 1 - \mu(x, z;t) \}^{1 - y}$, with the following score functions:
\begin{align*}
\dot{\ell}_{X, Y, Z}(x, y, z) &= \dot{\ell}_X(x) + \frac{\dot{e}(x) \{ z - e(x) \}}{e(x) \{ 1 - e(x) \}} + \frac{\dot{\mu}(x, z) \{ y - \mu(x, z) \}}{\mu(x, z) \{ 1 - \mu(x, z) \}}.
\end{align*}

Even in this new model, the decomposition $\psi_1^+ = \psi_a + \psi_b + \psi_c$ applies and the efficient influence functions for $\psi_a$ and $\psi_b$ remain the same.  Only the calculation of the influence function for $\psi_c$ needs to change.  To compute this influence function, we use the fact that the CVaR of a $\textup{Bernoulli}(\mu(x, z))$ distribution has the following formula:
\begin{align*}
\textup{CVaR}_{\tau}(x, 1) &= \E[ Q_{\tau}(x, 1) + \tfrac{1}{1 - \tau} \{ Y - Q_{\tau}(x, 1) \}_+ \mid X = x, Z = 1]\\
&= \E[ \mathbb{I} \{ \mu(x, 1) > 1 - \tau \} + \tfrac{1}{1 - \tau} \{ Y - \mathbb{I} \{ \mu(x, 1) > 1 - \tau \} \}_+ \mid X = x, Z = 1]\\
&= \min \{ \tfrac{1}{1 - \tau} \mu(x, 1), 1 \}
\end{align*}
Therefore, the derivative of the map $t \mapsto \psi_c(P_t)$ is given by:
\begin{align*}
\frac{\d}{\d t} \psi_c(P_t) \, \bigg|_{t = 0} &= \frac{\d}{\d t} \int_{\R^d} \{ 1 - e(x; t) \} \min \{ \tfrac{1}{1 - \tau} \mu(x, 1; t), 1 \} f_X(x; t) \, \d \nu(x)\\
&= \int_{\R^d} -\dot{e}(x) \textup{CVaR}_{\tau}(x, 1) f_X(x) \, \d \nu(x) + \int_{\R^d} \{ 1 - e(x) \} \textup{CVaR}_{\tau}(x, 1) \dot{\ell}_X(x) f_X(x) \, \d \nu(x)\\
&+ \int_{\R^d} \{1 - e(x) \} \tfrac{1}{1 - \tau} \dot{\mu}(x, 1) \mathbb{I} \{ \mu(x, 1) \leq 1 - \tau \} f_X(x) \, \d \nu(x)\\
&= \E\left[ \textup{CVaR}_{\tau}(X, 1) \{ Z - e(X) \} \dot{\ell}_{X,Y,Z}(X, Y, Z) \right] + \E \left[ \{ 1 - e(X) \} \textup{CVaR}_{\tau}(X, 1) \dot{\ell}_{X,Y,Z}(X, Y, Z) \right]\\
&+ \E \left[ \{ 1 - e(X) \} \E \left[ \tfrac{1}{1 - \tau} \dot{\mu}(X, Z) \mathbb{I} \{ \mu(X, 1) \leq 1 - \tau \} \mid X, Z = 1 \right] \right]
\end{align*}
Here, the assumption that $\mu(X, Z)$ has a density ensures that the map $t \mapsto \min \{ \mu(x, 1) / (1 - \tau), 1 \}$ is differentiable almost everywhere.  The expectation in the final line of the preceding display can be rewritten as follows: 
\begin{align*}
&\E \left[ \{ 1 - e(X) \} \E \left[ \tfrac{1}{1 - \tau} \dot{\mu}(X, Z) \mathbb{I} \{ \mu(X, 1) \leq 1 - \tau \} \mid X, Z = 1 \right] \right] \\
&\quad = \E \left[ \frac{1 - e(X)}{e(X)} Z \frac{1}{1 - \tau} \{ Y - \mu(X, 1) \} \mathbb{I} \{ \mu(X, 1) \leq 1 - \tau \} \frac{\dot{\mu}(X, Z) \{ Y - \mu(X, Z) \}}{\mu(X, Z) \{ 1 - \mu(X, Z) \}} \right]\\
&\quad = \E \left[ \frac{1 - e(X)}{e(X)} Z \left( Q_{\tau}(X, 1) + \tfrac{1}{1 - \tau} \{ Y - Q_{\tau}(X, 1) \}_+ - \textup{CVaR}_{\tau}(X, 1) \right) \frac{\dot{\mu}(X, Z) \{ Y - \mu(X, Z) \}}{\mu(X, Z) \{ 1 - \mu(X, Z) \}} \right]\\
&\quad = \E \left[ \frac{1 - e(X)}{e(X)} Z \left( Q_{\tau}(X, 1) + \tfrac{1}{1 - \tau} \{ Y - Q_{\tau}(X, 1) \}_+ - \textup{CVaR}_{\tau}(X, 1) \right) \dot{\ell}_{X,Y,Z}(X, Y, Z) \right]
\end{align*}
where in the second-to-last last line, we used the identity $\tfrac{1}{1 - \tau} \{ y - \mu(x, 1) \} \mathbb{I} \{ \mu(x, 1) \leq 1 - \tau \} = Q_{\tau}(x, 1) + \tfrac{1}{1 - \tau} \{ y - Q_{\tau}(x, 1) \}_+ - \textup{CVaR}_{\tau}(x, 1)$, which may be verified by separately considering the cases $\mu(x, 1) > 1 - \tau$ and $\mu(x, 1) \leq 1 - \tau$.

Finally, combining terms gives $\tfrac{\d}{\d t} \psi_t(P_t) |_{t = 0} = \E[ \phi_c(X, Y, Z) \dot{\ell}_{X,Y,Z}(X, Y, Z)]$.  The rest of the proof now goes through just as in the continuous case.
\end{proof}

\subsection{Proof of Theorem \ref{theorem:sharp_inference}}

\begin{proof}
We begin by proving that the proposed standard error $\hat{\sigma}_+$ satisfies $n  \hat{\sigma}_+^2 \xrightarrow{P} \Var( \phi_{\ATE}^+(X, Y, Z; \eta))$.  Start with the following identity:
\begin{align*}
n \hat{\sigma}_+^2 &= \frac{n}{n - 1} \sum_{k = 1}^K \frac{| \F_k|}{n} || \phi_{\ATE}^+(X, Y, Z ; \hat{\eta}^{(-k)}) - \hat{\psi}_{\ATE}^+ ||_{2,k}^2.
\end{align*}
By \Cref{lemma:influence_function_convergence} and \Cref{theorem:double_valid_double_sharp}, $|| \phi_{\ATE}^+(X, Y, Z; \hat{\eta}^{(-k)}) - \hat{\psi}_{\ATE}^+ ||_{2,k}^2 = || \phi_{\ATE}^+(X, Y, Z; \eta) - \psi_{\ATE}^+ ||_{2,k}^2 + o_P(1)$, and the right-hand side of this expression converges to $\Var( \phi_{\ATE}^+(X, Y, Z; \eta))$ by the law of large numbers.  Averaging over the $K$ folds gives $n \hat{\sigma}_+^2 = \Var( \phi_{\ATE}^+(X, Y, Z; \eta)) + o_P(1)$ as well.

Next, we observe that the conditions of \Cref{theorem:sharp_inference} imply that $\Var( \phi_{\ATE}^+(X, Y, Z; \eta))$ is positive.  Specifically, $\Var(Y \mid X, Z) > 0$ if $Y$ has a bounded conditional density or $Y$ is binary and $P(\mu(X, Z) \in \{ 0, 1 \}) = 0$.

Therefore, valid coverage follows from \Cref{theorem:efficiency} and Slutsky's theorem.
\end{proof}

\subsection{Proof of Theorem \ref{theorem:valid_inference}}

We analyze the Wald confidence interval based on $\hat{\psi}_1^+$, which takes the form $\hat{\psi}_1^+ + z_{1 - \alpha} \hat{\sigma}_{1,+}$ for $\hat{\sigma}_{1,+}$ defined as follows:
\begin{align*}
\hat{\sigma}_{1,+}^2 &= \frac{1}{n(n - 1)} \sum_{k = 1}^K \sum_{i \in \F_k} \{ \phi_1^+(X_i, Y_i, Z_i; \hat{\eta}^{(-k)}) - \hat{\psi}_1^+ \}^2.
\end{align*}
The argument for the ATE case is conceptually identical but is notationally more complex. 

We divide the proof into five parts.
\begin{enumerate}[itemsep=-1ex]
    \item In \ref{sec:bootstrap_wald_equivalence}, we show that the Wald upper confidence bound based on $\hat{\psi}_1^+$ is first-order equivalent to the (linearized, percentile) bootstrap upper confidence bound based on this estimator.
    \item In \ref{sec:stochastically_dominated_estimator}, we construct an estimator $\bar{\psi}_1^+$ whose (linearized, percentile) bootstrap distribution is stochastically dominated by that of $\hat{\psi}_1^+$.
    \item In \ref{sec:asymptotic_bootstrap_validity}, we study the asymptotics of $\bar{\psi}_1^+$ when $Y$ is continuous.
    \item In \ref{sec:valid_inference_summary}, we combine the preceding results to complete the proof in the continuous case.
    \item Finally, \ref{sec:valid_inference_binary_proof} proves the result in the binary case, which does not require the machinery from the earlier steps.
\end{enumerate}

\subsubsection{Bootstrap and Wald equivalence} \label{sec:bootstrap_wald_equivalence}

The first step is to relate the Wald upper confidence bound to a certain bootstrap upper confidence bound.  To define this bound, let $(M_{1,n}, \ldots, M_{n,n}) \sim \textup{Multinomial} \{ n, (1/n, \ldots, 1/n) \}$ independently of the original data.  For any function $f^* \equiv f^*(\{ (X_i, Y_i, Z_i, M_{i,n}) \}_{i \leq n})$, let $q_{1 - \alpha}(f)$ denote the $100(1 - \alpha)\%$ quantile of the distribution of $f$ given the original data:
\begin{align*}
q_{1 - \alpha}(f) &= \inf \{ t \in \R \, : \, \P( f( \{ (X_i, Y_i, Z_i, M_{i,n}) \}_{i \leq n}) \leq t \mid \{ (X_i, Y_i, Z_i) \}_{i \leq n}) \geq 1 - \alpha \}.
\end{align*}

The main result of this section shows that $\hat{\psi}_1^+ + z_{1 - \alpha} \hat{\sigma}_{1,+}$ is asymptotically equivalent to $q_{1 - \alpha}( \E_n[ M_{i,n} \phi_{i,n}])$ where, for each $k \in \{ 1, \ldots, K \}$ and each $i \in \F_k$, we set $\phi_{i,n} = \phi_1^+(X_i, Y_i, Z_i; \hat{\eta}^{(-k)})$.  This is the upper confidence bound that would be obtained by pretending that the nuisances $\hat{\eta}^{(-k)}$ were fixed rather than estimated and then applying the percentile bootstrap to the estimator $\hat{\psi}_1^+$.  For this reason, we call the distribution of $\E_n[ M_{i,n} \hat{\phi}_{i,n}]$ given $\{ (X_i, Y_i, Z_i) \}_{i \leq n}$ the ``linearized percentile bootstrap distribution" of $\hat{\psi}_1^+$.  This equivalence follows from a general fact on the bootstrap distribution of an average.

\begin{lemma}
\label{lemma:bootstrap_critical_value}
Let $h_n \equiv h_n( \{ (X_i, Y_i, Z_i) \}_{i \leq n}) \in \R^n$ be a random vector satisfying $\E_n[\{ h_{i,n} - h_{\infty}(X_i, Y_i, Z_i) \}^2] = o_P(1)$ for some square-integrable function $h_{\infty}$.  Then for any $\alpha \in (0, 1)$, we have the expansion
\begin{align*}
    q_{1 - \alpha}( \E_n[ M_{i,n} h_{i,n}] ) = \E_n[ h_{i,n}]  + z_{1 - \alpha} \sqrt{\frac{\Var(h_{\infty}(X, Y, Z))}{n}} + o_P(1/\sqrt{n})
\end{align*}
\end{lemma}

\begin{proof}
Let $\hat{F}_n(t) = \E_n[ \mathbb{I} \{ h_{i,n} \leq t \}]$ and $F(t) = P( h_{\infty}(X, Y, Z) \leq t)$.  Since $q_{1 - \alpha}( \E_n[ M_{i,n} h_{i,n}]) = \E_n[ h_{i,n}] + q_{1 - \alpha}( \sqrt{n}(\E_n[ (M_{i,n} - 1) h_{i,n}]))/\sqrt{n}$, it suffices to prove that $q_{1 - \alpha}( \sqrt{n}( \E_n [ (M_{i,n} - 1)] h_{i,n}]))$ converges to $z_{1 - \alpha} \sqrt{\Var(h_{\infty}(X, Y, Z))}$.  This will follow from the in-probability version of Theorem 15.4.3 in \cite{tsh} as long as we can verify two conditions.

First, we must check that the first two moments of $\hat{F}_n$ converge to those of $F$.  The following calculation shows that this is the case:
    \begin{align*}
    \E_n[ h_{i,n} ] - \E[ h_{\infty}(X, Y, Z)] &= \underbrace{\E_n[ h_{i,n} - h_{\infty}(X_i, Y_i, Z_i)]}_{\leq \E_n[ \{ h_{i,n} - h_{\infty}(X_i, Y_i, Z_i) \}^2]^{1/2} = o_P(1)} + \underbrace{\E_n[ h_{\infty}(X_i, Y_i, Z_i)] - \E[ h_{\infty}(X_i, Y_i, Z_i)]}_{= o_P(1) \text{ by law of large numbers}}\\
    \E_n[h_{i,n}^2] - \E[ h_{\infty}(X, Y, Z)^2] &= \underbrace{\E_n[ \{ h_{i,n} - h_{\infty}(X_i, Y_i, Z_i) \} \{ h_{i,n} + h_{\infty}(X_i, Y_i, Z_i) \}]}_{\leq \E_n[\{ h_{i,n} - h_{\infty}(X_i, Y_i, Z_i) \}^2]^{1/2} ( \E_n[h_{i,n}^2] + \E_n[h_{\infty}(X_i, Y_i, Z_i)^2])^{1/2} = o_P(1) O_P(1)}\\
    &- \underbrace{\E_n[ h_{\infty}(X_i, Y_i, Z_i)^2] - \E[ h_{\infty}(X, Y, Z)^2]}_{= o_P(1) \text{ by law of large numbers}}
    \end{align*}
    
Next, we must check that $\rho( \hat{F}_n, F) \xrightarrow{P} 0$ for some distance $\rho$ metrizing weak convergence.  To prove this, let $t \in \R$ be any continuity point of $F$.  Let $\epsilon > 0$ be arbitrary.  Since $F$ is continuous at $t$, there exists $\delta > 0$ such that $|t' - t| < \delta$ implies $|F(t) - F(t')| < \epsilon$.  Thus, we may write:
    \begin{align*}
    \hat{F}_n(t) - F(t) &= \underbrace{\E_n[ \mathbb{I} \{ h_{i,n} \leq t \} - \mathbb{I} \{ h_{\infty}(X, Y, Z) \leq t + \delta \}]}_{\leq \E_n[ \mathbb{I} \{ | h_{i,n} - h_{\infty}(X_i, Y_i, Z_i)| > \delta \}] \leq \E_n[ \{ h_{i,n} - h_{\infty}(X_i, Y_i, Z_i) \}^2]/\delta^2 = o_P(1)} + \underbrace{\E_n[ \mathbb{I} \{ h_{\infty}(X_i, Y_i, Z_i) \leq t + \delta \}] - F(t + \delta)}_{= o_P(1) \text{ by law of large numbers}}\\
    &+ F(t + \delta) - F(t)\\
    &\leq o_P(1) + \epsilon.
    \end{align*}
    A completely analogous argument shows that $\hat{F}_n(t) - F(t) \geq o_P(1) - \epsilon$.  Since $\epsilon$ is arbitrary, we conclude that $| \hat{F}_n(t) - F(t)| = o_P(1)$.  In particular, for any subsequence $\{ n_k \}$ we may find a further subsequence $\{ n_{k_j} \}$ along which $\hat{F}_n(t) \xrightarrow{\text{a.s.}} F(t)$.  By a diagonalization argument, we may choose this sub-subsequence so that $\hat{F}_n(t) \xrightarrow{\text{a.s.}} F(t)$ for all $t$ in a countable, dense collection of $F$-continuity points $\{ t_j \}_{j \geq 1}$.  This implies $\hat{F}_n \overset{\text{a.s.}}{\rightsquigarrow} F$ along the sub-subsequence.  By \cite[Theorem 2.3.2]{durrett2010probability}, we may conclude $\rho( \hat{F}_n, F) = o_P(1)$ along the entire sequence $n = 1, 2, 3, \ldots$.  Since $q_{1 - \alpha}( \E_n[ M_{i,n} h_{i,n}]) = \E_n[ h_{i,n}] + q_{1 - \alpha}( \sqrt{n}(\E_n[ (M_{i,n} - 1) h_{i,n}]))/\sqrt{n}$, this completes the proof. 
\end{proof}

\begin{lemma}
\label{lemma:dvds_bootstrap_equivalence}
The assumptions of \Cref{theorem:valid_inference} (continuous case) imply $\hat{\psi}_1^+ + z_{1 - \alpha} \hat{\sigma}_{1,+} = q_{1 - \alpha}( \E_n[ M_{i,n} \hat{\phi}_{i,n}]) + o_P(1/\sqrt{n})$.
\end{lemma}

\begin{proof}
In the continuous case, the assumptions of \Cref{theorem:valid_inference} combined with (the proof of) \Cref{lemma:influence_function_convergence} imply that $\E_n[ \{ \hat{\phi}_{i,n} - \phi_1^+(X_i, Y_i, Z_i; \bar{\eta}) \}^2] = o_P(1)$.  Thus, we may apply \Cref{lemma:bootstrap_critical_value} with $h_{\infty}(x, y, z) = \phi_1^+(x, y, z; \bar{\eta})$ to conclude:
\begin{align*}
q_{1 - \alpha}( \E_n[ M_{i,n} \hat{\phi}_{i,n}]) &= \E_n[ \hat{\phi}_{i,n} ] + z_{1 - \alpha} \sqrt{\frac{\Var( \phi_1^+(X, Y, Z; \bar{\eta}))}{n}} + o_P(1/\sqrt{n})
\end{align*}
At the same time, (the proof of) \Cref{theorem:sharp_inference} implies that $n \hat{\sigma}_{1,+}^2 \xrightarrow{P} \Var( \phi(X, Y, Z; \bar{\eta}))$, so the Wald upper confidence limit has the same expansion as above.
\end{proof}

\subsubsection{Constructing a stochastically dominated oracle} \label{sec:stochastically_dominated_estimator}

The second step is to construct an ``oracle" estimator $\bar{\psi}_1^+$ whose linearized percentile bootstrap distribution is stochastically dominated by that of the DVDS estimator $\hat{\psi}_1^+$.  For each $i \in \F_k$, define $\hat{\ell}_{i,n}$ by:
\begin{align*}
\hat{\ell}_{i,n} &= \hat{\phi}_{i,n} - \frac{1 - e(X_i)}{e(X_i)} Z_i \{ Y_i - \hat{Q}_+(X_i, 1) \} ( \Lambda^{\sign \{ Y_i - \hat{Q}_+(X_i, 1) \}} - \Lambda^{\sign \{ Y_i - Q_+(X_i, 1) \}}).
\end{align*}
and set $\bar{\psi}_1^+ = \E_n[ \hat{\ell}_{i,n}]$.  This estimator depends on oracle knowledge of the true quantile $Q_+$ and thus cannot be computed by the statistician.  Nevertheless, it is useful to consider because of the following result:

\begin{lemma}
\label{lemma:bootstrap_stochastic_dominance}
The inequality $q_{1 - \alpha}( \E_n[ M_{i,n} \hat{\ell}_{i,n}]) \leq q_{1 - \alpha}( \E_n[ M_{i,n} \hat{\phi}_{i,n}])$ holds surely.
\end{lemma}

\begin{proof}
The sign-matching argument from \Cref{sec:eform} implies that $\{ Y_i - \hat{Q}_+(X_i, 1) \} \Lambda^{\sign \{ Y_i - \hat{Q}_+(X_i, 1) \}}$ is always larger than $\{ Y_i - \hat{Q}_+(X_i, 1) \} \Lambda^{\sign \{ Y_i - Q_+(X_i, 1) \}}$.  Thus, the quantity subtracted from $\hat{\phi}_{i,n}$ in the definition of $\hat{\ell}_{i,n}$ is nonnegative and so $\E_n [M_{i,n} \hat{\ell}_{i,n}] \leq \E_n[ M_{i,n} \hat{\phi}_{i,n}]$ with probability one.  This implies that the distribution of $\E_n[ M_{i,n} \hat{\ell}_{i,n}]$ conditional on $\{ (X_i, Y_i, Z_i) \}_{i \leq n}$ is stochastically dominated by that of $\E_n [ M_{i,n} \hat{\phi}_{i,n}]$.  Stochastic dominance implies domination of all quantiles, so the conclusion follows.
\end{proof}

\subsubsection{Asymptotics of the oracle estimator in the continuous case} \label{sec:asymptotic_bootstrap_validity}

The third step is to study the behavior of the oracle estimator $\bar{\psi}_1^+$ and its linearized percentile bootstrap critical value.  In both of these results, the following function plays an important role:
\begin{align*}
\ell_{\infty}(x, y, z) &= zy + (1 - z) \bar{\varrho}_+ + \frac{1 - e(x)}{e(x)} z [ \bar{Q}_+(x, 1) + \Lambda^{\sign \{y - Q_+(x, 1) \}} \{ y - \bar{Q}_+(x, 1) \} - \bar{\varrho}_+(x, 1)]
\end{align*}
Here, we recall that $\bar{Q}_+$ is the limit of $\hat{Q}_+$ and $\bar{\varrho}_+$ is defined in \Cref{lemma:nuisance_facts}.  The quantity $\ell_{\infty}(x, y, z)$ is similar to $\phi_1^+(x, y, z; \bar{\eta})$ except the true quantile $Q_+$ is used in the ``sign" function.

\begin{lemma}
\label{lemma:psibar_bootstrap_equivalence}
Under the conditions of \Cref{theorem:valid_inference} (continuous case), the following holds for every $\alpha \in (0, 1)$:
\begin{align*}
q_{1 - \alpha}( \E_n[ M_{i,n} \hat{\ell}_{i,n}]) &= \bar{\psi}_1^+ + z_{1 - \alpha} \sqrt{\frac{\Var( \ell_{\infty}(X, Y, Z))}{n}} + o_P(n^{-1/2}).
\end{align*}
\end{lemma}

\begin{proof}
By \Cref{lemma:bootstrap_critical_value}, it suffices to prove that $\E_n[ \{ \hat{\ell}_{i,n} - \phi_1^+(X_i, Y_i, Z_i; \bar{\eta}) \}^2] = o_P(1)$.  We do this fold-by-fold, suppressing the $(\cdot)^{(-k)}$ superscripts for readability:
\begin{align*}
&|| \hat{\ell}_{i,n} - \ell_{\infty}(x, y, z) ||_{2,k}\\
&\leq
\underbrace{|| (1 - z) \{ \hat{\rho}_+(x, 1) - \bar{\varrho}_+(x, 1) \} ||_{2,k}}_{= o_P(1) \text{ by \Cref{lemma:nuisance_facts}.\ref{item:kappa_lipschitz}}} + \underbrace{\left| \left| \frac{1 - \hat{e}(x)}{\hat{e}(x)} z [\hat{\rho}_+(x, 1) - \bar{\varrho}_+(x, 1)] \right| \right|_{2,k}}_{\precsim_P || \hat{\rho}_+(x, 1) - \bar{\varrho}_+(x, 1) ||_2 / \epsilon = o_P(1) \text{ by \Cref{lemma:nuisance_facts}.\ref{item:kappa_lipschitz}}} \\
&+ \underbrace{\left| \left| \frac{1 - \hat{e}(x)}{\hat{e}(x)} [ \hat{Q}_+(x, 1) + \Lambda^{\sign \{ y - \hat{Q}_+(x, 1) \}} \{ y - \hat{Q}_+(x, 1) \} - \bar{Q}_+(x, 1) + \Lambda^{\sign \{ y - \hat{Q}_+(x, 1) \}} \{ y - \bar{Q}_+(x, 1) \} ] \right| \right|_{2,k}}_{\precsim_P \Lambda || \hat{Q}_+(x, 1) - \bar{Q}_+(x, 1) ||_2 / \epsilon = o_P(1) \text{ as $q \mapsto q + \Lambda^{\sign \{ y - \hat{Q}_+(x, 1) \}} \{ y - q \}$ is $\Lambda$-Lipschitz}}\\
&+ \underbrace{\left| \left| \left( \frac{1 - \hat{e}(x)}{\hat{e}(x)} - \frac{1 - e(x)}{e(x)} \right) z [ \bar{Q}_+(x, 1) + \Lambda^{\sign \{ y - \hat{Q}_+(x, 1) \}}\{ y - \bar{Q}_+(x, 1) \}] \right| \right|_{2,k}}_{\precsim || \{ 1/\hat{e}(x) - 1/e(x) \} \{ | \bar{Q}_+(x, 1)| + \Lambda |y - \bar{Q}_+(x, 1)| \} ||_{2,k} = o_P(1) \text{ by Markov + \Cref{lemma:nuisance_facts}.\ref{item:f_over_e}}}\\
&+ \underbrace{\left| \left| \frac{1 - e(x)}{e(x)} \Lambda^{\sign \{ y - Q_+(x, 1) \}} \{ \hat{Q}_+(x, 1) - \bar{Q}_+(x, 1) \} \right| \right|_{2,k}}_{\precsim_P || \hat{Q}_+(x, 1) - \bar{Q}_+(x, 1) ||_2 / \epsilon = o_P(1)}\\
&+ \underbrace{\left| \left| \left( \frac{1 - \hat{e}(x)}{\hat{e}(x)} - \frac{1 - e(x)}{e(x)} \right) \bar{\varrho}_+(x, 1) \right| \right|_{2,k}}_{\precsim_P || \bar{\varrho}_+(x, 1) \{ 1/\hat{e}(x) - 1/e(x) \} ||_2 = o_P(1) \text{ by \Cref{lemma:nuisance_facts}.\ref{item:f_over_e}}}\\
&\xrightarrow{P} 0.
\end{align*}
Combining over the $K$ folds and applying \Cref{lemma:bootstrap_critical_value} gives the conclusion.
\end{proof}

\begin{lemma}
\label{lemma:psibar_asymptotically_linear}
The conditions of \Cref{theorem:valid_inference} (continuous case) imply $\sqrt{n}( \bar{\psi}_1^+ - \psi_1^+) \rightsquigarrow N(0, \Var(\ell_{\infty}(X, Y, Z)))$.
\end{lemma}

\begin{proof}
We begin by proving that $\E_{n,k}[ \hat{\ell}_{i,n}] = \E_{n,k}[ \ell_{\infty}(X, Y, Z)] + o_P(n^{-1/2})$ for any $k \in \{ 1, \ldots, K \}$.  We prove this fold-by-fold, dropping the $(\cdot)^{(-k)}$ superscripts for readability.  We use the following decomposition:
\begin{align}
&\E_{n,k}[ \hat{\ell}_{i,n}] - \E_{n,k}[ \ell_{\infty}(X, Y, Z)] \nonumber \\
&= \E_{n,k} \left[ \left( \frac{1 - \hat{e}(X)}{\hat{e}(X)} - \frac{1 - e(X)}{e(X)} \right) Z [ \hat{Q}_+(X, 1) + \Lambda^{\sign \{ Y - \hat{Q}_+(X, 1) \}} \{ Y - \hat{Q}_+(X, 1) \} - \varrho_+(X, 1; \hat{Q}_+)] \right] \label{eq:lowerbound_change_e}\\
&+ \E_{n,k} \left[  \frac{1 - e(X)}{e(X)} [ \hat{Q}_+(X, 1) - \bar{Q}_+(X, 1)] Z (1 - \Lambda^{\sign \{ Y - Q_+(X, 1) \}}) \right] \label{eq:lowerbound_change_Q}\\
&+ \E_{n,k} \left[ \left( \frac{1 - e(X)}{e(X)} Z - [1 - Z] \right) \{ \bar{\varrho}_+(X, 1) - \hat{\rho}_+(X, 1) \} \right] \label{eq:lowerbound_change_kappa}\\
&+ \E_{n,k} \left[ \left( \frac{1 - \hat{e}(X)}{\hat{e}(X)} - \frac{1 - e(X)}{e(X)} \right) Z \{ \varrho_+(X, 1; \hat{Q}_+) - \hat{\rho}_+(X, 1) \} \right]  \label{eq:lowerbound_error_product}
\end{align}
The terms (\ref{eq:lowerbound_change_e}), (\ref{eq:lowerbound_change_Q}), (\ref{eq:lowerbound_change_kappa}) are all $o_P(n^{-1/2})$ by the same conditional Chebyshev argument used in the proof of \Cref{theorem:double_valid_double_sharp}:
\begin{align*}
|(\ref{eq:lowerbound_change_e})| &\precsim_P \bigg| \E \bigg[ \frac{e(X) - \hat{e}(X)}{e(X)} \underbrace{\E \bigg[ \hat{Q}_+(X, 1) + \Lambda^{\sign \{ Y - \hat{Q}_+(X, 1) \}} \{ Y - \hat{Q}_+(X, 1) \} - \varrho_+(X, 1; \hat{Q}_+) \, \bigg| \, X, Z = 1, \F_{-k} \bigg]}_{= 0 \text{ by \Cref{lemma:transformed_outcome_identity} + def. of $\varrho_+(\cdot, \cdot; \hat{Q}_+)$}} \, \bigg| \, \F_{-k} \bigg] \bigg| \\
&+ \frac{1}{\sqrt{| \F_k|}} \underbrace{\E \left[ \left( \frac{1}{\hat{e}(X)} - \frac{1}{e(X)} \right)^2 Z [\hat{Q}_+(X, 1) + \Lambda^{\sign \{ Y - \hat{Q}_+(X, 1) \}} \{ Y - \hat{Q}_+(X, 1) \} - \varrho_+(X, 1; \hat{Q}_+) ]^2 \, \bigg| \, \F_{-k} \right]^{1/2}}_{= o_P(1) \text{ by \Cref{lemma:nuisance_facts}.\ref{item:f_over_e}}}\\
&= o_P(n^{-1/2})\\
| (\ref{eq:lowerbound_change_Q})| &\precsim_P \bigg|  \E \bigg[ \frac{1 - e(X)}{e(X)} [ \hat{Q}_+(X, 1) - \bar{Q}_+(X, 1)] e(X) \underbrace{\E[ 1 - \Lambda^{\sign \{ Y - Q_+(X, 1) \}} \mid X, Z = 1]}_{= 1 - \Lambda (1 - \tau) - \Lambda^{-1} \tau = 0} \, \bigg| \, \F_{-k} \bigg] \bigg| \\
&+ \frac{1}{\sqrt{| \F_k|}} \underbrace{\E \left[ \left( \frac{1 - e(X)}{e(X)} [ \hat{Q}_+(X, 1) - \bar{Q}_+(X, 1)] Z (1 - \Lambda^{\sign \{ Y - Q_+(X, 1) \}}) \right)^2 \, \bigg| \, \F_{-k} \right]^{1/2}}_{\precsim_P (\Lambda/\epsilon) || \hat{Q}_+(x, 1) - \bar{Q}_+(x, 1) ||_2 = o_P(1)}\\
&= o_P(n^{-1/2})\\
|(\ref{eq:lowerbound_change_kappa})| &\precsim_P \bigg| \E \bigg[ \{ \bar{\varrho}_+(X, 1) - \hat{\rho}_+(X, 1) \} \underbrace{\E \bigg[ \frac{1 - e(X)}{e(X)} Z - [1 - Z] \, \bigg| \, X \bigg]}_{= 0} \, \bigg| \, \F_{-k} \bigg] \bigg|\\
&+ \frac{1}{\sqrt{|\F_k|}} \underbrace{\E \left[ \left( \frac{1 - e(X)}{e(X)} Z - [1 - Z] \right)^2 \{ \bar{\varrho}_+(X, 1) - \varrho_+(X, 1; \hat{Q}_+) \}^2 \, \bigg| \, \F_{-k} \right]^{1/2}}_{\precsim_P \Lambda || \bar{\varrho}_+(x, 1) - \varrho_+(\cdot, \cdot; \hat{Q}_+) ||_2 / \epsilon = o_P(1) \text{ by \Cref{lemma:nuisance_facts}.\ref{item:kappa_lipschitz}}}\\
&= o_P(n^{-1/2}).
\end{align*}
Meanwhile, the term (\ref{eq:lowerbound_error_product}) is $o_P(n^{-1/2})$ by the Cauchy-Schwarz inequality and our rate assumption.  

Now, we combine this conclusion over folds to obtain:
\begin{align*}
\sqrt{n}( \bar{\psi}_1^+ - \E[ \ell_{\infty}(X, Y, Z)] &= \sqrt{n}( \E_n[ \hat{\ell}_{i,n}] - \E[ \ell_{\infty}(X, Y, Z)])\\
&= \sqrt{n}( \E_n[ \ell_{\infty}(X, Y, Z)] - \E[ \ell_{\infty}(X, Y, Z)]) + o_P(n^{-1/2})\\
&\rightsquigarrow N(0, \Var( \ell_{\infty}(X, Y, Z))).
\end{align*}
It only remains to show that $\E[ \ell_{\infty}(X, Y, Z)] = \psi_1^+$.  This is done in the following calculation:
\begin{align*}
\E[ \ell_{\infty}(X, Y, Z)] &= \E \left[ ZY + \frac{1 - e(X)}{e(X)} Z [ \bar{Q}_+(X, 1) + \Lambda^{\sign \{ Y - Q_+(X, 1) \}} \{ Y - \bar{Q}_+(X, 1) \} \right]\\
&+ \underbrace{\E \left[ (1 - Z) \bar{\varrho}_+(X, 1) - \frac{1 - e(X)}{e(X)} Z \bar{\varrho}_+(X, 1) \right]}_{= 0}\\
&= \E \left[ ZY\right] + \E [ \{ 1 - e(X) \} \E[ \bar{Q}_+(X, 1) + \Lambda^{\sign \{ Y - Q_+(X, 1) \}} \{ Y - \bar{Q}_+(X, 1) \} \mid X, Z = 1] ]\\
&= \E[ZY] + \E[ \{ 1 - e(X) \} \E[ \Lambda^{\sign \{ Y - Q_+(X, 1) \}} Y \mid X, Z = 1]]\\
&+ \E[ \{ 1 - e(X) \} \bar{Q}_+(X, 1) \underbrace{\E[ 1 - \Lambda^{\sign \{ Y - Q_+(X, 1) \}} \mid X, Z = 1]}_{= 1 - \Lambda (1 - \tau) - \tau \Lambda^{-1} = 0}]\\
&= \E[ ZY] + \E[ \{ 1 - e(X) \} \E[ \Lambda^{\sign \{ Y - Q_+(X, 1) \}} Y \mid X, Z = 1]]\\
&+ \E[ \{ 1 - e(X) \} Q_+(X, 1)(1 - \Lambda^{\sign \{ Y - Q_+(X, 1) \}} \mid X, Z = 1]\\
&= \E[ZY] + \E[ \{ 1 - e(X) \} \underbrace{\E[ Q_+(X, 1) + \Lambda^{\sign \{ Y - Q_+(X, 1) \}} \{ Y - Q_+(X, 1) \} \mid X, Z = 1]}_{= \rho_+(X, 1) \text{ by \Cref{lemma:transformed_outcome_identity}}}]\\
&= \E[ZY + (1 - Z) \rho_+(X, 1)]\\
&= \psi_1^+.
\end{align*}
\end{proof}

\subsubsection{Putting things together for the continuous case} \label{sec:valid_inference_summary}

\begin{proof}
Now, we combine the results of the previous steps to prove \Cref{theorem:valid_inference} in the continuous case.  
\begin{align*}
\P( \psi_1^+ \leq \hat{\psi}_1^+ + z_{1 - \alpha}\hat{\sigma}_{1,+}) &= \P( \psi_1^+ \leq q_{1 - \alpha}( \E_n[ M_{i,n} \hat{\phi}_{i,n}]) + o_P(1/\sqrt{n})) &\text{\Cref{lemma:dvds_bootstrap_equivalence}}\\
&\leq \P( \psi_1^+ \leq q_{1 - \alpha}( \E_n[ M_{i,n} \hat{\ell}_{i,n}]) + o_P(1/\sqrt{n})) &\text{\Cref{lemma:bootstrap_stochastic_dominance}}\\
&= \P \left( \psi_1^+ \leq \bar{\psi}_1^+ + z_{1 - \alpha}\sqrt{\frac{\Var( \ell_{\infty}(X, Y, Z))}{n}} + o_P(1/\sqrt{n}) \right) &\text{\Cref{lemma:psibar_bootstrap_equivalence}}\\
&= \P( \sqrt{n}( \psi_1^+ - \bar{\psi}_1^+) \leq z_{1 - \alpha} \sqrt{\Var(\ell_{\infty}(X, Y, Z))} + o_P(1))\\
&\rightarrow 1 - \alpha &\text{\Cref{lemma:psibar_asymptotically_linear} + Slutsky}.
\end{align*}
In the last step, the assumptions in Condition \ref{condition:density} suffice to ensure $\Var( \ell_{\infty}(X, Y, Z)) > 0$ so that Slutsky's theorem may be applied.  Thus, $\liminf \P( \psi_1^+ \leq \hat{\psi}_1^+ + z_{1 - \alpha} \hat{\sigma}_{1,+}) \geq 1 - \alpha$.

We briefly mention why this guarantee holds uniformly over contiguous neighborhoods of $P$.  When $|| \bar{Q}_+(x, 1) - Q_+(x, 1) ||_2 \neq 0$, this is not surprising.  The proof of \Cref{theorem:double_valid_double_sharp} shows that $\hat{\psi}_1^+$ converges to $\E[ZY + (1 - Z) \bar{\varrho}_+(X, 1)]$ where $\bar{\varrho}_+$ is defined in \Cref{lemma:nuisance_facts}.  Since $Q_+(x, 1)$ is the unique minimizer of the map $q \mapsto \int q + \{ y - q \}_+ \, \d F(y \mid x, 1)$ and $\bar{Q}_+(x, 1)$ differs from $Q_+(x, 1)$ on a set of positive probability, we must have $\bar{\varrho}_+(X, 1) > \rho_+(X, 1)$ with positive probability.  Hence, $\E[ZY + (1 - Z) \bar{\varrho}_+(X, 1)] > \E[ZY + (1 - Z) \rho_+(X, 1)] = \psi_1^+$ so $\hat{\psi}_1^+$ converges to a limit that is strictly too large.  As a consequence, the coverage of the Wald upper confidence bound tends to one under $P$.  By the definition of contiguity, the same holds under any sequence $\{ P_n^n \}$ that is contiguous to $\{ P^n \}$.  

The case where $\bar{Q}_+ = Q_+$ but $|| \hat{Q}_+ - Q_+ ||_2 \neq o_P(n^{-1/4})$ is more subtle.  The key insight is that $\ell_{\infty}(x, y, z) = \phi_1^+(x, y, z; \eta)$ when $\bar{Q}_+ = Q_+$, so the lower bound $\bar{\psi}_1^+$ is semiparametrically efficient.  Therefore, the convergence in the final line of the preceding display is locally uniform by the regularity of $\bar{\psi}_1^+$.
\end{proof}

\subsubsection{Proof for the binary case} \label{sec:valid_inference_binary_proof}

\begin{proof}
In the binary case, the proof is considerably simpler because \Cref{lemma:nuisance_facts}.\ref{item:mu_to_Q} still assures quantile consistency, albeit without any rate.  In particular, the proof of \Cref{theorem:sharp_inference} still applies and guarantees that $n \hat{\sigma}_{1,+}$ consistently estimates $\Var( \phi_1^+(X, Y, Z; \eta))$.

To analyze the behavior of the estimator itself, start from the decomposition from \Cref{sec:asymptotic_linearity_continuous}:
\begin{align*}
\E_{n,k}[ \phi_1^+(X, Y, Z; \hat{\eta})] - \E_{n,k}[\phi_1^+(X, Y, Z; \eta)] &= (\ref{eq:change_q}) + (\ref{eq:change_e}) + (\ref{eq:change_kappa}) + (\ref{eq:error_product}).
\end{align*}
The arguments used in \Cref{sec:asymptotic_linearity_binary} to show that (\ref{eq:change_e}), (\ref{eq:change_kappa}), (\ref{eq:error_product}) are $o_P(n^{-1/2})$ did not rely on the $L^{\infty}$ rate condition, so they continue to apply under the weaker assumptions of \Cref{theorem:valid_inference}.  

Meanwhile, we will employ a version of the sign-matching argument from \Cref{sec:eform} to show that the term $(\ref{eq:change_q})$ is \emph{greater} than a quantity that is $o_P(n^{-1/2})$.  Define the function $S(x, y)$ as follows:
\begin{align*}
S(x, y) &= 
\left\{ 
\begin{array}{ll}
1 &\text{if } y > Q_+(x, 1)\\
\frac{1}{\log \Lambda} \left[ \log \left( \frac{1 - \mu(x, 1)/\Lambda)}{\log \mu} \right) \mathbb{I} \{ \mu(x, 1) > 1 - \tau \} + \log \left( \frac{1 - \Lambda \mu(x, 1)}{1 - \mu(x, 1)}\right) \mathbb{I} \{ \mu(x, 1) \leq 1 - \tau \} \right) &\text{if } y = Q_+(x, 1)\\
-1 &\text{if } y < Q_+(x, 1)
\end{array}
\right.
\end{align*}
The weight on the boundary is chosen so that $S(x, y) \in [-1, 1]$ and $\E[ 1 - \Lambda^{S(X, Y)} \mid X, Z = 1] = 0$.  Using this function, we may use \Cref{lemma:transformed_outcome_identity} to rewrite (\ref{eq:change_q}) as follows:
\begin{align*}
(\ref{eq:change_q}) &= \E_{n,k} \left[ \frac{1 - \hat{e}(X)}{\hat{e}(X)} Z (1 - \Lambda^{-1}) [ \hat{Q}_+(X, 1) + \Lambda^{\sign \{ Y - \hat{Q}_+(X, 1) \}}\{ Y - \hat{Q}_+(X, 1) \} \right]\\
&- \E_{n,k} \left[ \frac{1 - \hat{e}(X)}{\hat{e}(X)} Z (1 - \Lambda^{-1})[ Q_+(X, 1) + \Lambda^{\sign \{ Y - Q_+(X, 1) \}} \{ Y - Q_+(X, 1) \}] \right]\\
&\geq \E_{n,k} \left[ \frac{1 - \hat{e}(X)}{\hat{e}(X)} Z (1 - \Lambda^{-1}) [ \hat{Q}_+(X, 1) + \Lambda^{S(X, Y)} \{ Y - \hat{Q}_+(X, 1) \} \right]\\
&- \E_{n,k} \left[ \frac{1 - \hat{e}(X)}{\hat{e}(X)} Z (1 - \Lambda^{-1}) [ Q_+(X, 1) + \Lambda^{S(X, Y)} \{ Y - Q_+(X, 1) \}] \right]\\
&= \E_{n,k} \left[ \frac{1 - \hat{e}(X)}{\hat{e}(X)} Z (1 - \Lambda^{-1}) [ \hat{Q}_+(X, 1) - Q_+(X, 1)] [1 - \Lambda^{S(X, Y)}] \right].
\end{align*}
The lower bound in the final line of the preceding display is $o_P(n^{-1/2})$, which can be seen by applying Chebyshev's inequality conditionally on $\F_{-k}$.
\begin{align*}
&\E_{n,k} \left[ \frac{1 - \hat{e}(X)}{\hat{e}(X)} Z(1 - \Lambda^{-1}) [ \hat{Q}_+(X, 1) - Q_+(X, 1)] [1 - \Lambda^{S(X, Y)}] \right]\\
&\precsim_P \bigg| \E \bigg[ \frac{1 - \hat{e}(X)}{\hat{e}(X)} e(X) (1 - \Lambda^{-1}) [ \hat{Q}_+(X, 1) - Q_+(X, 1)] \underbrace{\E[ 1 - \Lambda^{S(X, Y)} \mid X, Z = 1]}_{= 0} \, \bigg| \, \F_{-k} \bigg] \bigg|\\
&+ \frac{1}{\sqrt{| \F_k|}} \underbrace{\E \left[ \left( \frac{1 - \hat{e}(X)}{\hat{e}(X)} Z (1 - \Lambda^{-1}) [ \hat{Q}_+(X, 1) - Q_+(X, 1)] [ 1 - \Lambda^{S(X, Y)}] \right)^2 \right]^{1/2}}_{\precsim (\Lambda/\epsilon) || \hat{Q}_+ - Q_+ ||_2 = o_P(1) \text{ by \Cref{lemma:nuisance_facts}.\ref{item:mu_to_Q}}}\\
&= o_P(n^{-1/2}).
\end{align*}
Thus we have shown $\E_{n,k}[ \phi_1^+(X, Y, Z; \hat{\eta})] \geq \E_{n,k}[ \phi_1^+(X, Y, Z; \eta)] - o_P(n^{-1/2})$.

Combining over folds gives the following conclusion:
\begin{align*}
\P( \psi_1^+ \leq \hat{\psi}_1^+ + z_{1 - \alpha} \hat{\sigma}_{1,+}) &\geq \P \left( \psi_1^+ \leq \E_n[ \phi_1^+(X, Y, Z; \eta)] + z_{1 - \alpha} \hat{\sigma}_{1,+} - o_P(1/\sqrt{n}) \right)\\
&= \P( \sqrt{n} ( \psi_1^+ - \E_n[ \phi_1^+(X, Y, Z; \eta)]) \leq z_{1 - \alpha}( n \hat{\sigma}_{1,+})^2 + o_P(1))\\
&= \P( \sqrt{n}( \psi_1^+ - \E_n[ \phi_1^+(X, Y, Z; \eta)]) \leq z_{1 - \alpha} (\Var( \phi_1^+(X, Y, Z; \eta) + o_P(1))^{1/2} + o_P(1))\\
&\rightarrow 1 - \alpha.
\end{align*}
Thus, valid inference holds also in the binary case.
\end{proof}

\subsection{Proof of \Cref{theorem:transformed_outcome_rate}}

The proof is broken into two parts.  The first part collects the implications of our assumptions for the curvature and concentration of the loss function.  The second part uses the standard empirical risk minimization results to get the rate of convergence.  

In both parts, we use the following quantities and notations.  $N$ denotes the size of the dataset $\mathcal{I}_2$, and $\E_N[\cdot] = \tfrac{1}{N} \sum_{i \in \mathcal{I}_2} [ \cdot ]_i$.  Moreover, we define the following functions:
\begin{align*}
\ell(x, y, z; g) &= \{ \Lambda^{-1} y + (1 - \Lambda^{-1})[ \hat{Q}_+(x, z) + \tfrac{1}{1 - \tau} \{ y - \hat{Q}_+(x, z) \}_+] - g(x, z) \}^2\\
\hat{L}_n(g) &= \E_N[ \ell(X, Y, Z; g)]\\
L_n(g) &= \int \ell(x, y, z; g) \, \d P(x, y, z).
\end{align*}
Finally, we let $\delta_n$ be a sequence with $n^{-1/\alpha} \ll \delta_n \ll n^{-1/\alpha} \log(n)^{2 \gamma / \alpha}$, where $a_n \ll b_n$ means $a_n = o(b_n)$.

\subsubsection{Curvature and concentration} \label{sec:curvature_concentration}

\begin{lemma} \label{lemma:argmin_convergence}
Assume the conditions of \Cref{theorem:transformed_outcome_rate}.  Let $g_n = \argmin_{g \in \mathcal{G}_n} L_N(g)$.  Then $|| g_n - \varrho_+(\cdot, \cdot; \hat{Q}_+) ||_2 \leq \delta_n$ with probability approaching one.
\end{lemma}

\begin{proof}
Since $|| g_n - \varrho_+(\cdot, \cdot; \hat{Q}_+) ||_2 \leq || g_n - \varrho_+(\cdot, \cdot; \bar{Q}_+) ||_2 + || \varrho_+(\cdot, \cdot; \bar{Q}_+) - \varrho_+(\cdot, \cdot; \hat{Q}_+) ||_2$, it suffices to show both terms in the upper bound are less than $\delta_n$ with high probability.

For the first term, this follows from our approximation condition:  $|| g_n - \varrho_+(\cdot, \cdot; \bar{Q}_+) ||_2 = O(n^{-1/\alpha}) = o(\delta_n)$.

For the second term, we consider two cases.  If $|| \hat{Q}_+ - Q_+ ||_2 = O_P(n^{-1/\alpha})$, then \Cref{lemma:nuisance_facts}.\ref{item:kappa_lipschitz} implies $|| \varrho_+(\cdot, \cdot; \bar{Q}_+) - \varrho_+(\cdot, \cdot; \hat{Q}_+) ||_2 \leq \Lambda || \hat{Q}_+ - \bar{Q}_+ ||_2 = O_P(n^{-1/\alpha}) = o_P(\delta_n)$.  If $|| \hat{Q}_+ - \bar{Q}_+ ||_4 = O_P(n^{-1/2\alpha})$, then \Cref{lemma:cvar_neyman_orthogonal} gives $|| \varrho_+(\cdot, \cdot; \bar{Q}_+) - \varrho_+(\cdot, \cdot; \hat{Q}_+) ||_2 \leq \Lambda M || \hat{Q}_+ - Q_+ ||_4^2 = O_P(n^{-1/\alpha}) = o_P(\delta_n)$.  Since the conditions of \Cref{theorem:transformed_outcome_rate} implies one of these must happen, we conclude $|| \varrho_+(\cdot, \cdot; \bar{Q}_+) - \varrho_+(\cdot, \cdot; \hat{Q}_+) ||_2 = o_P(\delta_n)$.
\end{proof}

\begin{lemma} \label{lemma:curvature}
Assume the conditions of \Cref{theorem:transformed_outcome_rate}.  Let $g_n = \argmin_{g \in \mathcal{G}_n} L_n(g)$.  Then with probability approaching one, the following holds:
\begin{align}
L_n(g) - L_n(g_n) \geq \tfrac{1}{2} || g - g_n ||_2^2 \text{ for all $g$ satisfying $|| g - g_n ||_2 \geq \delta_n$} \label{eq:curvature}
\end{align}
\end{lemma}

\begin{proof}
\Cref{lemma:argmin_convergence} implies that $|| g_n - \varrho_+(\cdot, \cdot; \hat{Q}_+) ||_2 \leq \delta_n/4$ with probability approaching one.  On this event, the following inequalities hold for all $g$ satisfying $|| g - g_n ||_2 \geq \delta_n$:
\begin{align*}
L_n(g) - L_n(g_n) &= [L_n(g) - L_n(\varrho_+(\cdot, \cdot; \hat{Q}_+))] - [L_n(g_n) - L_n( \varrho_+(\cdot, \cdot; \hat{Q}_+))]\\
&= || g - \varrho_+(\cdot, \cdot; \hat{Q}_+) ||_2^2 - || g_n - \varrho_+(\cdot, \cdot; \hat{Q}_+) ||_2^2\\
&\geq ( || g - g_n ||_2 - || g_n - \varrho_+(\cdot, \cdot; \hat{Q}_+) ||_2)^2 - || g_n - \varrho_+(\cdot, \cdot; \hat{Q}_+) ||_2^2\\
&= || g - g_n ||_2^2 - 2 || g_n - \varrho_+(\cdot, \cdot; \hat{Q}_+) ||_2 || g - g_n ||_2\\
&\geq \tfrac{1}{2} || g - g_n ||_2^2.
\end{align*}
\end{proof}

\begin{lemma}
Assume the conditions of \Cref{theorem:transformed_outcome_rate}.  Then for all large $n$, there exists a constant $C \equiv C(\alpha, \Lambda, B) < \infty$ such that (\ref{eq:modulus}) holds with probability one for all $\delta \in [\delta_n/2, 1]$.
\begin{align}
\E \bigg[ \sup_{\substack{|| g - g_n ||_2 \leq \delta\\ g \in \mathcal{G}_n}} | \sqrt{N}[\hat{L}_n(g) - L_n(g)] - \sqrt{n} [ \hat{L}_n(g_n) - L_n(g)]| \, \bigg| \mathcal{I}_1 \bigg] \leq C \log(n)^{\gamma} \delta^{2 - \alpha/2} \label{eq:modulus}
\end{align}
Here, we abuse notation and use $\mathcal{I}_1$ as shorthand for $\sigma(\{ X_i, Y_i, Z_i) \} \, : \, i \in \mathcal{I}_1 \}$.
\end{lemma}

\begin{proof}
Let $\Delta(x, y, z; g) = \ell(x, y, z; g) - \ell(x, y, z; g_n)$ and let $\mathcal{L}_n = \{ \Delta(\cdot, \cdot, \cdot; g) \, : \,  g \in \mathcal{G}_n \}$.  We start by bounding the bracketing entropy $H(s, \mathcal{L}_n)$.  For any $s > 0$ and any $g, g' \in \mathcal{G}_n$ we have:
\begin{align*}
&| \Delta(x, y, z; g) - \Delta(x, y, z; g') |\\
&\quad = |\{ 2[ \hat{Q}_+(x, z) + \Lambda^{\sign \{ y - \hat{Q}_+(x, z) \}} \{ y - \hat{Q}_+(x, z) \}]  - g(x, z) - g'(x, z) \} \{ g(x, z) - g'(x, z) \}|\\
&\quad \leq 2 \Lambda B | g(x, z) - g'(x, z)|
\end{align*}
Thus, any $s$-bracketing of $\mathcal{G}_n$ yields an $s/(2 \Lambda B)$ bracketing of $\mathcal{L}_n$, yielding the bound $H(s, \mathcal{L}_n) \leq (2 \Lambda B)^{\alpha - 2} \log(n)^{\gamma} / s^{\alpha - 2}$.  Plugging this bound in gives control of the bracketing integral of $\mathcal{L}_n$:
\begin{align*}
J_{[]}(\delta, \mathcal{L}_n, || \cdot ||_2) &= \int_0^{\delta} 1 + \sqrt{H(s, \mathcal{L}_n)} \, \d s \leq \delta + \frac{(2 \Lambda B)^{\alpha - 2}}{4 - \alpha} \log(n)^{\gamma} \delta^{2 - \alpha/2}
\end{align*}
For any $\delta \leq 1$ and all large $n$, we may further simplify the bound to $J_{[]}(\delta, \mathcal{L}_n, || \cdot ||_2) \leq \frac{2 (2 \Lambda B)^{\alpha - 2}}{4 - \alpha} \log(n)^{\gamma} \delta^{2 - \alpha/2}$.

Now, we apply \citet[Lemma 3.4.2]{vdv_wellner}, which implies that the following inequalities hold for some universal constant $c < \infty$:
\begin{align*}
&\E \bigg[ \sup_{\substack{|| g - g_n ||_2 < \delta\\ g \in \mathcal{G}_n}} | \sqrt{N}[ \hat{L}_n(g) - L_n(g)] - \sqrt{N} [ \hat{L}_n(g_n) - L_n(g)]| \, \bigg| \, \mathcal{I}_1 \bigg]\\
&\quad = \E \bigg[ \sup_{\substack{|| g - g_n ||_2 < \delta\\ g \in \mathcal{G}_n}} \sqrt{N} \left| \E_N[ \Delta(X, Y, Z; g)] - \int \Delta(x, y, z; g) \, \d P(x, y, z) \right| \, \bigg| \, \mathcal{I}_1  \bigg]\\
&\quad \leq \E \bigg[ \sup_{\substack{|| \Delta ||_2 < (2 \Lambda B) \delta\\ \Delta \in \mathcal{L}_n}} \sqrt{n} \left| \E_n[ \Delta(X, Y, Z; g)] - \int \Delta(x, y, z; g) \, \d P(x, y, z) \right| \, \bigg| \, \mathcal{I}_1 \bigg]\\
&\quad \leq c \left\{ J_{[]}(2 \Lambda B \delta, \mathcal{L}_n, || \cdot ||_2) + \frac{B J_{[]}(2 \Lambda B \delta, \mathcal{L}_n, || \cdot ||_2)^2}{(2 \Lambda B \delta)^2 \sqrt{N}} \right\} \\
&\quad \leq c \left\{ 
\frac{(2 \Lambda B)^{\alpha/2}}{4 - \alpha} \log(n)^{\gamma} \delta^{2 - \alpha/2} + \frac{4 B (2 \Lambda B)^{\alpha - 2}}{(4 - \alpha)^2} \frac{\log(n)^{2 \gamma}}{\sqrt{N}} \delta^{2 - \alpha} 
\right\}.
\end{align*}
For any $\delta < \delta_n$ and all large $n$, the second term in the upper bound is smaller than the first.  Therefore, for large enough $n$, this upper bound can be further simplified to $C(\alpha, \Lambda, B) \log(n)^{\gamma} \delta^{2 - \alpha/2}$.
\end{proof}

\subsubsection{Proof of \Cref{theorem:transformed_outcome_rate}}

\begin{proof}
First, we consider the convergence of $|| \hat{\rho}_+ - g_n ||_2$ to zero, which can be handled by standard M-estimator theory.  Specifically, we apply \citet[Theorem 3.4.1]{vdv_wellner}, with $M_n = L_n$, $\mathbb{M}_n = \hat{L}_n$, $\eta = \infty$, $\phi_n(\delta) = \log(n)^{\gamma} \delta^{2 - \alpha/2}$, and $\delta_n$ as introduced earlier in this section to conclude:
\begin{align*}
|| \hat{\rho}_+ - g_n ||_2 = O_P( n^{-1/\alpha} \log(n)^{2 \gamma / \alpha}).
\end{align*}
The conditions of that Theorem are verified in \Cref{sec:curvature_concentration}, except the result of \Cref{lemma:curvature} holds only with high probability whereas \citet[Theorem 3.4.1]{vdv_wellner} requires it to hold surely.  Nevertheless, the proof still goes through by handling the event where curvature fails in the same way that the proof of \citet[Theorem 3.2.5]{vdv_wellner} handles the event $2d(\hat{\theta}_n, \theta_0) \geq \eta$.

Now, we consider three cases.  First, suppose that $|| \hat{Q}_+ - Q_+ ||_2 = O_P(n^{-1/\alpha})$.  Then $\hat{Q}_+$ is correctly-specified meaning $\bar{\varrho}_+ = \rho_+$, which allows us to write:
\begin{align*}
|| \hat{\rho}_+ - \rho_+ ||_2 &= || \hat{\rho}_+ - \bar{\varrho}_+ ||_2\\
&\leq || \hat{\rho}_+ - g_n ||_2 + || g_n - \bar{\varrho}_+ ||_2\\
&= || \hat{\rho}_+ - g_n ||_2 + \inf_{g \in \mathcal{G}_n} || g - \bar{\varrho}_+ ||_2\\
&= O_P(n^{-1/\alpha} \log(n)^{2 \gamma / \alpha}) + O(n^{-1/\alpha})\\
&= o_P(n^{-1/4}).
\end{align*}
Note that in this case, the convergence rate $|| \hat{Q}_+ - Q_+ || = O_P(n^{-1/\alpha})$ is only used in \Cref{lemma:curvature} to verify the condition of \citet[Theorem 3.4.1]{vdv_wellner}.

Next, consider the case where $|| \hat{Q}_+ - Q_+ ||_4 = O_P(n^{-1/2 \alpha})$.  Then we may write:
\begin{align*}
|| \hat{\rho}_+ - \varrho_+(\cdot, \cdot; \hat{Q}_+) ||_2 &\leq || \hat{\rho}_+ - g_n ||_2 + || g_n - \bar{\varrho}_+ ||_2 + || \bar{\varrho}_+ - \varrho_+(\cdot, \cdot; \hat{Q}_+) ||_2\\
&= O_P(n^{-1/\alpha} \log(n)^{2 \gamma / \alpha}) + O_P(n^{-1/\alpha}) + O_P( || \hat{Q}_+ - Q_+ ||_4^2) &\text{\Cref{lemma:cvar_neyman_orthogonal}}\\
&= O_P(n^{-1/\alpha} \log(n)^{2 \gamma /\alpha}) + O_P(n^{-1/\alpha}) + O_P(n^{-1/\alpha})\\
&= o_P(n^{-1/4}).
\end{align*}

Finally, consider the misspecified case $|| \hat{Q}_+ - \bar{Q}_+ ||_2 = O_P(n^{-1/\alpha})$.
\begin{align*}
|| \hat{\rho}_+ - \varrho_+(\cdot, \cdot; \hat{Q}_+) ||_2 &= || \hat{\rho}_+ - g_n ||_2 + || g_n - \bar{\varrho}_+ ||_2 + || \bar{\varrho}_+ - \varrho_+(\cdot, \cdot; \hat{Q}_+) ||_2\\
&= O_P(n^{-1/\alpha} \log(n)^{2 \gamma / \alpha}) + O_P(n^{-1/\alpha}) + O_P( || \hat{Q}_+ - \bar{Q}_+ ||_2) &\text{\Cref{lemma:nuisance_facts}.\ref{item:kappa_lipschitz}}\\
&= O_P(n^{-1/\alpha} \log(n)^{2 \gamma / \alpha}) + O_P(n^{-1/\alpha}) + O_P(n^{-1/\alpha})\\
&= o_P(n^{-1/4})
\end{align*}

Thus, in all the cases considered in \cref{theorem:transformed_outcome_rate}, least-squares regression on the estimated transformed outcome achieves a convergence rate faster than $n^{-1/4}$.
\end{proof}

\end{document}